\documentclass[twoside,11pt]{article}

\usepackage{arxiv_template}

\usepackage{mathtools} 
\usepackage{bbm}
\usepackage{amsfonts}
\usepackage{enumitem}
\usepackage{booktabs}
\usepackage{multirow}
\usepackage{makecell}
\usepackage{diagbox}

\usepackage{algorithm}
\usepackage{algpseudocode}

\usepackage[caption=false]{subfig}

\usepackage{csquotes}
\usepackage{numprint}
\usepackage{leftindex}
\usepackage{centernot}

\usepackage[nameinlink, capitalize]{cleveref}

\newcommand{\mcR}{\mathcal{R}}

\newcommand{\bE}{\boldsymbol{E}}
\newcommand{\PFER}{\textnormal{PFER}}

\newcommand{\FDR}{\textnormal{FDR}}
\newcommand{\FDP}{\textnormal{FDP}}
\newcommand{\NFD}{\textnormal{FD}}
\newcommand{\BC}{\textnormal{BC}}
\newcommand{\RB}{\textnormal{RB}}
\newcommand{\RWC}{\textnormal{RWC}}
\newcommand{\ph}{\textnormal{ph}}
\newcommand{\dph}{\textnormal{dph}}
\newcommand{\ebh}{\textnormal{ebh}}
\newcommand{\kn}{\textnormal{kn}}
\newcommand{\phnu}{\textnormal{ph-}\nu}

\DeclareMathOperator*{\argmax}{argmax}

\newcommand{\indep}{\perp \!\!\! \perp}

\newtheorem{fact}{Fact}

\usepackage{lastpage}
\ShortHeadings{Choosing the nominal level post-hoc with knockoffs}{Fischer and Sechidis}
\firstpageno{1}

\begin{document}

\title{Choosing the nominal level post-hoc\\ with knockoffs using e-values}

\author{\name Lasse Fischer\thanks{These authors contributed equally to this work.}
\email fischer1@uni-bremen.de \\
\addr Competence Center for Clinical Trials Bremen\\
       University of Bremen\\
       Bremen, Germary
\AND
\name Konstantinos Sechidis\footnotemark[1] 
\email kostas.sechidis@novartis.com\\
\addr Advanced Methodology and Data Science\\
       Novartis Pharma AG\\
       Basel, Switzerland
       }


\maketitle

\begin{abstract}
The knockoff filter is a powerful tool for controlled variable selection with false discovery rate (FDR) control. In this paper, we leverage e-values to allow the nominal FDR level to be switched post-hoc, after looking at the data and applying the knockoff procedure. This approach addresses a significant limitation of standard knockoffs: while frequently used in high-dimensional regressions, they often lack power in low-dimensional and sparse signal settings. One of the main reasons for this is that the knockoff filter requires a minimum number of selections that depends strictly on the nominal FDR level. By utilizing e-values, we can increase the nominal level in cases where the original procedure makes no discoveries, or decrease it to improve precision when discoveries are abundant. These improvements come without any costs, meaning the results of our post-hoc procedure are always more informative than those of the original knockoff filter. We extend this methodology to recently proposed derandomized knockoff procedures and demonstrate its utility in variable selection problems relevant to drug development using real clinical trial data.
\end{abstract}


\section{Introduction\label{sec:intro}}

In exploratory analyses, the identification of relevant variables is a critical step toward understanding complex data-generating mechanisms and uncovering meaningful associations.\footnote{We refer to the task as variable selection, noting that equivalent terminology includes feature selection, biomarker discovery, or covariate identification, depending on the application domain.} Over the past years, a wide range of methods that leverage modern machine learning techniques have been proposed to quantify the importance of candidate variables. While these approaches provide powerful tools for ranking variables according to their apparent influence, they typically lack rigorous guarantees for controlled variable selection. In particular, translating variable importance scores into a selected subset of variables while controlling false discoveries remains a major methodological challenge. To bridge this gap, recent methods have aimed to develop inference procedures that combine statistical rigor with modeling flexibility, with the knockoff framework emerging as a particularly effective approach \citep{barber2015controlling, candes2018panning}.

\subsection{Controlled variable selection and knockoffs}

In this work, we consider the problem of variable selection under rigorous false discovery control. 
Suppose we observe a response $Y$ and variables $\mathbf{X}=(X_1,\dots,X_p)$. 
For each variable $X_i$, we consider the conditional independence null hypothesis
\[
H_0^i : X_i \perp\!\!\!\perp Y \mid \mathbf{X}_{-i},
\]
where $\mathbf{X}_{-i}$ denotes all variables except $X_i$. Variables for which $H_0^i$ holds are called \emph{null}, and the remaining are \emph{non-null} (or \emph{relevant}). We denote these sets by
\[
I_0 = \{ i : H_0^i \text{ is true} \}, \qquad
I_1 = [p] \setminus I_0.
\]

A variable selection procedure outputs a set of discoveries/rejections $R \subseteq [p]$, and its quality is measured by the {false discovery rate} (FDR), which is defined as the expected number of false discovery proportion (FDP):
\[
\mathrm{FDR} = \mathbb{E}\!\left[\mathrm{FDP}\right],~~\mathrm{with}~~
\mathrm{FDP} = \frac{|R \cap I_0|}{|R| \vee 1},
\]
where $a \vee b=\max(a,b)$ avoids division by zero. 

The task of identifying the relevant variables, that is, those in $ I_1 $, naturally leads to the question of how to design a variable selection procedure that can recover as many non-nulls as possible while rigorously controlling false discoveries. The knockoff framework of \citet{barber2015controlling}, refined and generalized by \citet{candes2018panning}, addresses this goal by enabling controlled variable selection: it yields a data-driven rejection set of putative discoveries while provably controlling the FDR at a user-specified level $ \alpha \in (0,1) $, meaning $\FDR\leq \alpha$. Methods based on knockoffs (KO) construct synthetic variables that act as negative controls, enabling variable selection with provable FDR guarantees, while making minimal assumptions on the conditional distribution of $Y\mid \mathbf{X}$. A full description of the knockoff construction and filtering procedure is deferred to Section~\ref{sec:knockoffs}.

\subsection{Motivation of our work}
While the knockoff framework was primarily developed for high-dimensional problems, it is equally valuable to have methods that provide controlled variable selection in low- and moderate-dimensional settings, particularly when the underlying signal is expected to be sparse. In many scientific applications, the number of truly relevant variables is small, even if a large number of variables are measured. When identifying effect modifiers (i.e., predictive variables) in clinical trials, the total number of available clinical variables is typically modest, often on the order of 10–50 patient characteristics, reflecting the scale of routinely collected baseline data. Within this set, it is generally assumed that only a small subset truly drives heterogeneity in treatment effect, leading to a sparse predictive structure in most methodological developments and simulation studies \citep{Lipkovich2017,sechidis2018distinguishing,Lipkovich2024, Sun2024}.
Interestingly, in such settings, the vanilla knockoff procedure exhibits certain limitations that motivate further methodological development. In our work we will focus on two main limitations of the original knockoffs:
\begin{description}
    \item[(a)] The minimum number of discoveries required to achieve statistical significance is determined by the pre-specified FDR level $ \alpha $, and it is $1/\alpha,$ which can be overly conservative when only a few signals are present \citep{zimmermann2024, luo2025improving}. As a result, the procedure may fail to flag any discoveries even when a small but genuine signal exists.
    \item[(b)] If the difficulty of the problem (e.g., signal strength or sparsity) is not well understood beforehand --- which is not rare in exploratory analyses, especially when new data are being studied --- choosing the FDR level $ \alpha $ a priori can be challenging. Once the analysis is conducted, there is no valid post-hoc adjustment of $ \alpha $ to adapt to the observed evidence, which may lead to either low power or weak error guarantee.  
\end{description}

These limitations underscore the need for methods that maintain the rigorous error control of the knockoff framework while providing greater flexibility and improved sensitivity in sparse or low-signal settings. In the remainder of this paper, we develop and assess such approaches, demonstrating how post-hoc adjustment of the nominal FDR level can overcome these challenges.

\subsection{Contribution}
Our paper addresses the aforementioned two key limitations of standard FDR-controlling knockoff methods by leveraging e-values to adjust the nominal FDR level after observing the data.

\paragraph{Switching $\alpha$ post-hoc with knockoffs:} 
E-values \citep{ramdas2024hypothesis, ramdas2023game} are a novel tool for hypothesis testing that has proven particularly useful for sequential anytime-valid inference \citep{grunwald2020safe, shafer2021testing, waudby2020confidence, howard2021time}, testing of composite nulls in irregular models \citep{wasserman2020universal, larsson2024numeraire}, multiple testing \citep{wang2022false, ignatiadis2024asymptotic, vovk2021values, xu2025bringing}, and, most recently, for hypothesis testing at data-dependent nominal levels \citep{grunwald2024beyond, koning2023post, gauthier2025values, xu2025bringing} (see Section~\ref{sec:e-closure} for a detailed description of FDR control with data-dependent $\alpha$). In this paper, we exploit the two latter aspects of e‑values to construct a knockoff procedure, described in Algorithm \ref{alg:posthoc-alpha}, that allows the nominal FDR level to be adjusted post‑hoc. This improves the original knockoff method \citep{barber2015controlling, candes2018panning} in two ways.

\begin{enumerate}
    \item If the original knockoff method does not make any discoveries, we can increase the nominal level to potentially obtain rejections.
    \item If the original knockoff method makes discoveries, we can often decrease the nominal level while keeping the same discovery set, thus increasing the precision of the results.
\end{enumerate}


Figure \ref{fig:example} illustrates the benefits of the proposed post‑hoc adjustment.  In Figure~\ref{fig:example_power}, we observe that the post-hoc procedure dramatically increases power: across all settings for the number of relevant variables, our method successfully identifies nearly all relevant variables (i.e., power $\approx 1$). In contrast, the original knockoff procedure exhibits very low power when the number of relevant variables is small. This is primarily because the minimum number of discoveries is determined by the nominal FDR level $\alpha$, which in our case corresponds to $1/\alpha = 5$. To illustrate this, when the actual number of relevant variables is 3, the original procedure produces no discoveries in 1490 out of 2000 iterations and exactly 5 discoveries in 238 iterations, which are the two most frequent outcomes. By contrast, the post-hoc method most frequently selects 3 variables (924 iterations) and 4 variables (499 iterations), corresponding to the two most common outcomes, closely matching the true number of relevant variables.

The gains in power come with only a small sacrifice in the (averaged across simulations) nominal $\alpha$ value (solid lines), as illustrated in Figure~\ref{fig:example_false_discoveries} for scenarios with a low number of relevant variables (3 and 4). Interestingly, the average FDP acrosss the simulations (i.e., empirical FDR) remains below 0.20 for both methods (dashed lines). Moreover, our post-hoc adjustment can reduce the effective nominal level $\alpha$ when the problem is easier than anticipated, while maintaining the same power as the original procedure.

\begin{figure}[htbp]
\vspace{-0.3cm}
\centering
  \subfloat[Power vs number of relevant variables\label{fig:example_power}]{%
    \includegraphics[width=0.45\textwidth]{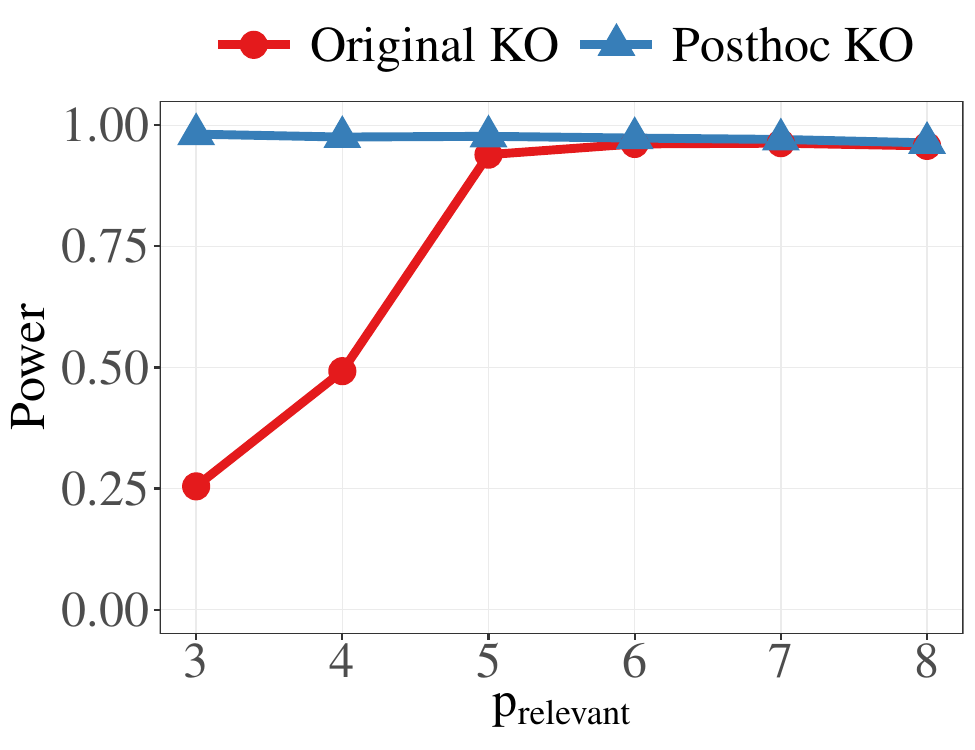}
  }
  \hfill
  \subfloat[FDR and $\alpha$ vs number of relevant variables\label{fig:example_false_discoveries}]{%
    \includegraphics[width=0.45\textwidth]{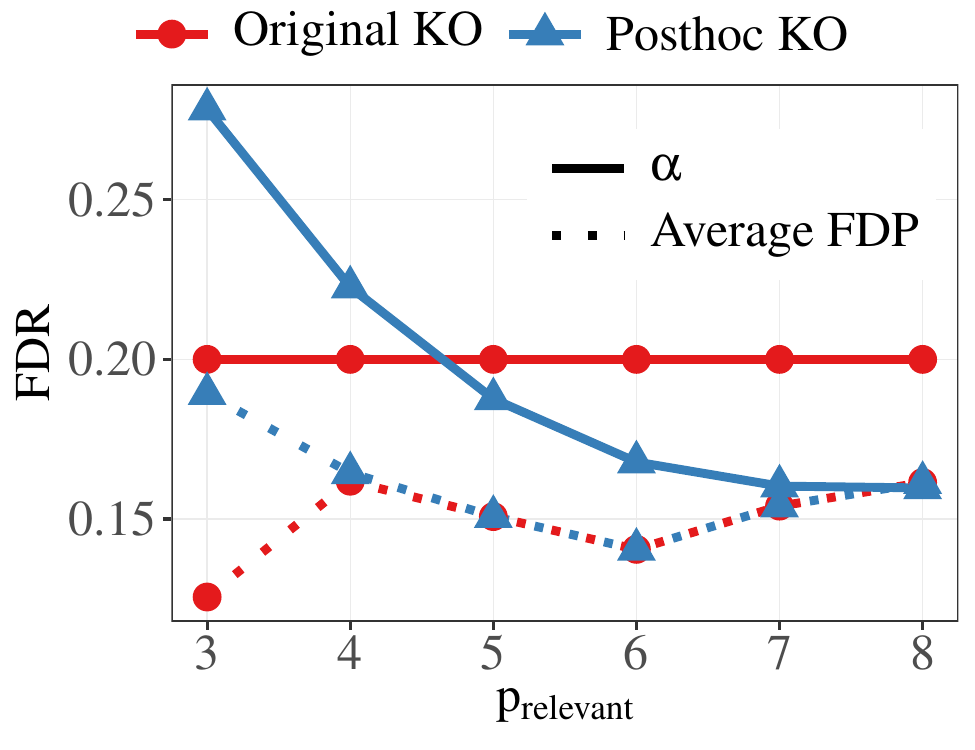}
  }
  \caption{This figure illustrates the main contribution of our work, namely that the proposed post-hoc adjustment of the $\alpha$ values substantially increases statistical power compared to the original approach. For small values of $p_{\text{relevant}}$, the $\alpha$ level is slightly increased to obtain rejections in cases where the original method does not make any discoveries, whereas for larger $p_{\text{relevant}}$, the post-hoc method can even reduce $\alpha$ while maintaining similar (or better) power. 
  The nominal level of $\alpha$ for the original KO method is set to $0.20.$ Results are averaged over 2000 runs, and a full description of the data‑generating process, parameter choices, and simulation protocol is provided in Section \ref{sec:sims}.}
  \label{fig:example}
\end{figure}

It should be noted that our improvement is a \enquote{free-lunch}, meaning our approach never makes less discoveries and it only reports a larger significance level $\alpha$ in cases where the original method did not make any rejections. Therefore, the results reported by our method are always more informative. To further illustrate these properties, we examine the setting with 3 relevant variables, as detailed in Table \ref{tbl:example}. Out of the 1490 runs where the original method failed to yield any discoveries, the post-hoc approach successfully identified an average of 3.27 variables. While this was achieved at a higher average nominal level ($0.313$), it highlights the post-hoc method's ability to extract information where the original method remains uninformative. Conversely, in the 510 runs where the original method successfully identified variables, our post-hoc approach matched the original discoveries, but did so while maintaining tighter error control, achieving an average nominal level of $0.178$ compared to the original $0.20$.
\begin{table}[htbp]
\vspace{-0.3cm}
\centering
\caption{Comparison of original vs posthoc KO number of discoveries and nominal $\alpha$ levels}
\label{tbl:example}
\small 
\begin{tabular}{cccc}
\toprule
\makecell{\textbf{Original KO}\\at nominal level $\alpha = 0.20$\\predefined } &  
\makecell{\textbf{Post-hoc KO}\\ avg. number\\ of discoveries} & 
\makecell{\textbf{Post-hoc KO}\\ avg. nominal level $\alpha$\\ post-hoc determined} \\ \midrule
No discoveries  in  1490 runs out of 2000 & 
3.27 & 0.313 \\
$\ge 5$ discoveries in  510 runs out of 2000  
& Same as original KO & 0.178 \\ \bottomrule
\end{tabular}
\vspace{-0.5cm}
\end{table}

\paragraph{Post-hoc $\alpha$ with derandomized knockoffs:} 
Since the knockoff procedure is based on the random generation of knockoff variables, the results are heavily exposed to external randomness. In order to reduce dependence on the random data generator, \citet{ren2023derandomizing} and \citet{ren2024derandomised} introduced derandomized knockoff procedures by sampling multiple knockoff copies for each variable. Specifically, \citet{ren2023derandomizing} focuse on controlling the per-family error rate (PFER), and \citet{ren2024derandomised} provide a framework for FDR control. These derandomized procedures require the pre-specification of an additional parameter, on which the performance and guarantee of the procedure significantly depend. We show that the parameter can be chosen arbitrarily based on the entire data while maintaining valid error control, thus increasing the efficiency and utility of the derandomized methods enormously. Furthermore, we show how the derandomized procedures can be further uniformly improved.

\subsection{Outline}
In the following, we start to briefly summarize the knockoff procedure and the e-Closure Principle to derive FDR procedures with post-hoc error control. Afterwards, in Section~\ref{sec:our_approach}, we introduce our post-hoc knockoff procedure and derive its error control. In Section~\ref{sec:derandomized}, we consider the derandomized post-hoc procedures and our proposed improvements. In Section~\ref{sec:sims}, we provide extensive simulation results that demonstrate the effectiveness of our approach compared to existing methods.\footnote{The code implementing the methods described in this paper is publicly available at \url{https://github.com/sechidis/posthoc_knockoffs}.} Finally, in Section~\ref{sec:clinical_trial}, we illustrate how to use our proposed methods in clinical trial data, demonstrating its effectiveness in identifying prognostic and predictive variables, two tasks very critical in modern drug development.

\section{Preliminaries}

\subsection{The \texorpdfstring{$\BC$}{BC}-knockoff procedure\label{sec:knockoffs}}

The knockoff procedure introduced by \citet{barber2015controlling} and extended by \citet{candes2018panning} is a powerful framework for variable selection while controlling the FDR. Since we consider multiple variations of their initial approach in this paper, we call it $\BC$-knockoff procedure in the following. The strategy is composed of three steps:

\begin{enumerate}
    \item \textbf{Generate Knockoff Variables:} For each input variable $X_i$, an artificial variable $\tilde{X}_i$ called a knockoff is generated to act as a negative control for assessing the importance of the original variable. To ensure the knockoff variables can serve this purpose, they are constructed with the following two critical properties:
    \begin{description}
        \item[Property (a) - Exchangeability.] The original and the knockoff variables are pairwise exchangeable, in other words, the joint distribution of $(\mathbf{X}, \tilde{\mathbf{X}})$ is invariant under the swapping of any $X_i$ and $\tilde{X}_i.$  

        \item[Property (b) - Conditional Independence.] Given the original variables $\mathbf{X}$, the knockoffs $\tilde{\mathbf{X}}$ are conditionally independent of the response variable $Y$. That is:
        \[
        \tilde{\mathbf{X}} \indep Y \, | \, \mathbf{X}.
        \]
        This property is automatically satisfied if the response $Y$ is not used in the knockoff generation process.
    \end{description}
    Various algorithms have been proposed to construct knockoffs that obey the exchangability property \citep{Sesia2018hmm,Romano2019,Jordon2019,Bates2020,huang2020relaxing,spector2022powerful}.

    \item \textbf{Estimate Knockoff Statistics:} For each variable $X_i$, we calculate knockoff statistics $W_i$ to measure how much more important the original variable is compared to its knockoff $\tilde{X}_i$. A positive value of $W_i$ indicates that the original variable is more important, where the magnitude captures the strength of the evidence. The derivation of knockoff statistics follows a two-step procedure:

    \begin{description}
        \item[Step 2a - Calculate Variable Importance Scores.] Using the dataset that includes both original variables $\mathbf{X}$ and their knockoffs $\tilde{\mathbf{X}}$, along with the outcome variable $Y$, variable importance scores are computed for both original and knockoff variables. Specifically, we calculate a statistic of statistics $(Z, \tilde{Z})$, representing the importance measures for each pair of original and knockoff variables:
        \begin{align}
           (Z, \tilde{Z}) = (Z_1, \dots, Z_p, \tilde{Z}_1, \dots, \tilde{Z}_p) = q \left\{(\mathbf{X}, \tilde{\mathbf{X}}), Y \right\}, \notag
        \end{align}
        where $Z_i$ and $\tilde{Z}_i$ capture importance for the original variable $X_i$ and its knockoff $\tilde{X}_i$, respectively. The function $q$ represents an appropriate model or statistical procedure that assigns importance scores. These statistics are constructed so that swapping any variable with its knockoff results in the corresponding components of the statistic being swapped as well. Most variable importance score estimation procedures, such as those based on regression coefficients and permutation importance measures, inherently satisfy this symmetry property \citep{candes2018panning}. This exchangeability is a key requirement for maintaining the validity of knockoffs.

        \item[Step 2b - Derive Knockoff Statistics.] Using the variable importance scores from Step 2a, the knockoff statistic $W_i$ is defined as:
        \begin{align}
            W_i = f(Z_i, \tilde{Z}_i), \notag
        \end{align}
        where $f$ is an antisymmetric function satisfying $f(u, v) = -f(v, u)$. This property ensures that swapping $X_i$ with its knockoff $\tilde{X}_i$ flips the sign of $W_i$. Common choices for $f$ include:
        \begin{align}
            f(u, v) = u - v, \quad \text{or} \quad f(u, v) = |u| - |v|.  \notag
        \end{align}
        A large, positive $W_i$ provides evidence against the null hypothesis that $X_i$ is an unimportant variable, allowing it to be included in the rejection set in Step 3. 
    \end{description}

    The last years, many powerful knockoff statistics have been introduced in the literature, see for example \citet{spector2024asymptoticallyoptimalknockoffstatistics}. Furthermore, there are works that derive knockoff statistics targeting effect modification in randomized trials: for example, \citet{sechidis2021usingknockoffs} develop predictive knockoff statistics that quantify the extent to which baseline variables modify treatment effects in RCTs, enabling FDR-controlled selection of predictive variables (i.e., effect modifiers) rather than purely prognostic ones.

    \item \textbf{Threshold Knockoff Statistics:} Let $\epsilon_i$ be the sign of $W_i$ and $\mathcal{W}=\sigma(|W_1|,\ldots,|W_p|)$. Define the rejection set $R^{\BC}$ for a predefined level $\alpha^{\kn}\in (0,1]$ as:
    \begin{align}
        R^{\BC} &\coloneqq \{i\in [p]: W_i \geq T_{\alpha^{\kn}}^{\BC} \}, \notag\\ 
        \text{where } 
        T_{\alpha^{\kn}}^{\BC}&\coloneqq \inf\left\{t>0: \frac{\sum_{i\in [p]} \mathbbm{1}\{W_i\geq t\}}{1+\sum_{i\in [p]} \mathbbm{1}\{W_i\leq -t\}}\geq 1/\alpha^{\kn} \right\}. 
        \label{eq:stop_BC}
    \end{align}
    If there does not exist a $t>0$ satisfying the above condition, then $T_{\alpha^{\kn}}^{\BC}=\infty$. 
    The key theoretical result that forms the foundation of this methodology is stated in this Lemma:
    \begin{fact}[Lemma 3.3 in \citep{candes2018panning}]\label{fact:candes}
        Conditional on $\mathcal{W}$, the signs corresponding to true hypotheses $(\epsilon_i)_{i \in I_0}$ are i.i.d. coin flips.
    \end{fact}
\end{enumerate}

With the latter result, one can show that the $\BC$-knockoff procedure controls the FDR at level $\alpha=\alpha^{\kn}$ \citep{barber2015controlling, candes2018panning}. Similar to \citet{ren2024derandomised}, we distinguish between the FDR level $\alpha$ and the parameter $\alpha^{\kn}$  used for the knockoff stopping time in \eqref{eq:stop_BC}, since they can deviate later in the paper. Also, in this paper we propose to choose the significance level $\alpha$ post-hoc, however, the parameter $\alpha^{\kn}$ needs to be specified in advance and can be interpreted as an \textit{initial} or \textit{targeted} significance level.

\subsection{Post-hoc choice of \texorpdfstring{$\alpha$}{alpha}\label{sec:e-closure}}
The Closure Principle of \citet{marcus1976closed} has long been known as a fundamental approach to multiple testing with familywise error rate (FWER) control.
Recently, \citet{xu2025bringing} introduced an e-Closure Principle for FDR control. Given a \textit{local e-value} $E_S$ for each intersection hypothesis $H_0^S=\bigcap_{i\in S} H_0^i$, $S\subseteq [p]$, meaning $\mathbb{E}_{H_0^S}[E_S]\leq 1$, we consider the \textit{closure set}
\begin{align}
    \mcR_{\alpha}(\bE)\coloneqq \left\{R\subseteq [p]: E_S\geq  \frac{\FDP_S(R)}{\alpha} \text{ for all } S\subseteq [p]\right\}, \label{eq:closure_set}
\end{align}
where $\bE=(E_S)_{S\subseteq [p]}$ and $\FDP_S(R)=\frac{|S\cap R|}{|R|\lor 1}$. \citet{xu2025bringing} showed that any rejection set $R\in \mcR_{\alpha}(\bE)$ controls the FDR and any procedure that controls the FDR at level $\alpha$ can be constructed in this way. 

Furthermore, \citet{xu2025bringing} showed that these closure sets even permit a more sophisticated error rate control that goes beyond the classical notion of FDR control, i.e.,  $\alpha$ and $R\in \mcR_{\alpha}(\bE)$ can be chosen post-hoc!  Given a family of e-values $\bE=(E_S)_{S\subseteq [p]}$, they proved that 
\begin{align}
    & \mathbb{E}\left[\sup_{\alpha >0} \max_{R\in \mcR_{\alpha}(\bE)} \frac{\FDP_{I_0}(R)}{\alpha}\right] \leq 1. \label{eq:sup_alpha}
\end{align}
This particularly implies that 
\begin{align}
    \mathbb{E}\left[ \frac{\FDP_{I_0}(\tilde{R})}{\tilde{\alpha}}\right] \leq 1 \text{ for any data-dependent } \tilde{\alpha}\in (0,1], \tilde{R}\in \mcR_{\tilde{\alpha}}(\bE). \label{eq:post_hoc_alpha}
\end{align} 
Therefore, after calculating the e-values $\bE$, we can construct the closure sets $\mcR_{\alpha}(\bE)$ for multiple $\alpha$'s and then decide completely data-dependent which $\tilde{\alpha}\in (0,1]$ and $\tilde{R}\in \mcR_{\tilde{\alpha}}(\bE)$ to report while guaranteeing \eqref{eq:post_hoc_alpha}.

Note that the post-hoc error rate control in \eqref{eq:post_hoc_alpha} is slightly different from classical FDR control (where the FDR is bounded by $\alpha$), since $\tilde{\alpha}$ is inside the expectation. If $\tilde{\alpha}$ was fixed in advance, then the two notions would be equivalent due to the linearity of the expected value. However, since $\tilde{\alpha}$ is allowed to be a random variable depending on the data, the classical notion needs to be adapted by drawing $\tilde{\alpha}$ inside the expectation. This notion of post-hoc validity was also considered in \citep{grunwald2024beyond, koning2023post} for single hypotheses.

Nevertheless, \eqref{eq:post_hoc_alpha} provides a sensible notion of error rate control for post-hoc $\tilde{\alpha}$: If we choose a small $\tilde{\alpha}$, we will (on average) only have the option to choose a rejection set $\tilde{R}$ with small $\FDP_{I_0}(\tilde{R})$, since \eqref{eq:post_hoc_alpha} would be violated otherwise. 
Note that \eqref{eq:post_hoc_alpha} cannot be trivially satisfied by applying a fixed-level multiple testing procedure at different levels $\alpha$ and choosing the best result, as we illustrate by the following example.  

\begin{example}
    Suppose all null hypotheses are true ($I_0=[p]$). In that case $\FDP_{I_0}(R)$ becomes $1$ if $R\neq \emptyset$ and $0$ otherwise. Suppose we have a multiple testing procedure $R_\alpha$ for each $\alpha\in (0,1]$ that has exact error rate control, meaning $\mathbb{E}[\FDP_{I_0}(R_{\alpha})]=\alpha$ (for example, this is satisfied by the Benjamini-Hochberg procedure \citep{benjamini1995controlling} with independent and uniformly distributed null p-values). Now set $\tilde{\alpha}=\inf\{\alpha \in (0,1]: R_{\alpha}\neq \emptyset]$. Then, we have
    $$
    \mathbbm{E}\left[\frac{\FDP_{I_0}\left(R_{\tilde{\alpha}} \right)}{\tilde{\alpha}} \right]=\int_0^1 1/x\ dx =\infty.
    $$
    Therefore, cherry-picking the significance level with fixed-level procedures can inflate the expected value in \eqref{eq:post_hoc_alpha} indefinitely.
\end{example}



In this paper, we will develop a knockoff procedure that gains power by allowing the significance level to be chosen post-hoc, while ensuring that \eqref{eq:sup_alpha} (and thus \eqref{eq:post_hoc_alpha}) is fulfilled.

\section{Switching \texorpdfstring{$\alpha$}{alpha} post-hoc with knockoffs\label{sec:our_approach}}

Considering the stopping time $T_{\alpha^{\kn}}^{\BC}$ in \eqref{eq:stop_BC}, it is apparent that at least $1/\alpha^{\kn}=1/\alpha$ variables with a strong signal are needed to make rejections with $\BC$-knockoff, as $T_{\alpha^{\kn}}^{\BC}=\infty$ otherwise. For this reason, $\BC$-knockoff often performs poorly in low dimensional settings (see Figure~\ref{fig:example}).
In this section, we show how we can exploit the e-Closure Principle (Section~\ref{sec:e-closure}) in such cases to increase the significance level $\alpha$ (but not the parameter $\alpha^{\kn}$) post-hoc and potentially make some rejections while guaranteeing FDR control in the sense of \eqref{eq:post_hoc_alpha}. Furthermore, we also show that in cases where $\BC$-knockoff makes rejections, we can often decrease the significance level, thus increasing the precision.



\subsection{Increase \texorpdfstring{$\alpha$}{alpha} to gain power}

First, let us define the stopping time
\begin{align}
    T_{\alpha^{\kn}}^{\ph}\coloneqq \inf\left\{t>0: \frac{\sum_{i\in [p]} \mathbbm{1}\{W_i\geq t\}}{1+\sum_{i\in [p]} \mathbbm{1}\{W_i\leq -t\}}\geq 1/\alpha^{\kn} \quad \text{or} \quad \sum_{i\in [p]} \mathbbm{1}\{W_i\leq -t\}=0\right\}, \label{eq:stop_ph}
\end{align}
where the parameter $\alpha^{\kn}\in (0,1]$ can be interpreted as the prespecified initial significance level. The difference of $T_{\alpha^{\kn}}^{\ph}$ to the stopping time $T_{\alpha^{\kn}}^{\BC}$ of \citet{barber2015controlling} (see equation \eqref{eq:stop_BC}) is the second condition $\sum_{i\in [p]} \mathbbm{1}\{W_i\leq -t\}=0$. Be aware that this condition only comes into play if the first condition is not satisfied for any $t$, meaning $T_{\alpha^{\kn}}^{\BC}=T_{\alpha^{\kn}}^{\ph}$ if $R^{\BC}\neq \emptyset$. However, in case of $T_{\alpha^{\kn}}^{\BC}=\infty$, we have $T_{\alpha^{\kn}}^{\ph} < \infty$. Similar to $\BC$-knockoff, our procedure, which we call post-hoc knockoff (ph-knockoff), rejects all hypotheses $H_0^i$ with $W_i\geq T_{\alpha^{\kn}}^{\ph}$. Since $T_{\alpha^{\kn}}^{\ph}\leq T_{\alpha^{\kn}}^{\BC}$, ph-knockoff is uniformly more powerful than $\BC$-knockoff by \citet{barber2015controlling}. However, in case of $T_{\alpha^{\kn}}^{\BC}=\infty$, we can only report the rejections at a larger significance level $\alpha > \alpha^{\kn}$, as we show in the next section (see Theorem~\ref{theo:posthoc_alpha}).

\begin{remark}\label{remark:RB_stop}
     \citet{luo2025improving} and \citet{ren2024derandomised} used a modification of the stopping time $T_{\alpha^{\kn}}^{\BC}$ that is similar to $T_{\alpha^{\kn}}^{\ph}$, adding the condition $\sum_{i\in [p]} \mathbbm{1}\{W_i\geq t\} <1/\alpha^{\kn}$ instead of $\sum_{i\in [p]} \mathbbm{1}\{W_i\leq -t\}=0$. They used this adjusted stopping time to improve their calibrated and derandomized knockoff procedures without considering a post-hoc choice of $\alpha$. While our approach of switching $\alpha$ post-hoc would also work with their stopping time, our stopping time $T_{\alpha^{\kn}}^{\ph}$ usually leads to smaller post-hoc levels and therefore we think $T_{\alpha^{\kn}}^{\ph}$ is more appropriate.  
\end{remark}

\subsection{Decrease \texorpdfstring{$\alpha$}{alpha} to improve precision}

Up to this point, we have only considered to choose a level $\alpha>\alpha^{\kn}$ to potentially obtain more rejections. However, it turns out that in cases where $\BC$-knockoff makes rejections ($R^{\BC}\neq \emptyset$), it actually controls the FDR at a smaller level $\alpha$ than the prespecified $\alpha^{\kn}$. Since rejections at a smaller significance level indicate more evidence against the null hypothesis, we can decrease $\alpha$ compared to the initial significance level $\alpha^{\kn}$ post-hoc to improve the precision of the results. 

In Algorithm~\ref{alg:posthoc-alpha}, we summarize our general procedure, where we denote by $\tilde{\alpha}^{\ph}$ our post-hoc significance level and by $R^{\ph}$ the corresponding rejection set. If $\BC$-knockoff does not make any rejections, we increase the significance level ($\tilde{\alpha}^{\ph}>\alpha^{\kn}$), and thus potentially obtain rejections. If $\BC$-knockoff makes rejections, we decrease the significance level ($\tilde{\alpha}^{\ph}\leq \alpha^{\kn}$) to increase the evidence. In this way, the claims made by $\ph$-knockoff are always stronger than the claims made by $\BC$-knockoff. 

\begin{algorithm}
\caption{$\ph$-knockoff}\label{alg:posthoc-alpha}
 \textbf{Input:} Initial significance level $\alpha^{\kn}$ and knockoff statistics $W_1,\ldots,W_p$.\\ 
 \textbf{Output:} Rejection set $R^{\ph}$ and post-hoc level $\tilde{\alpha}^{\ph}$.
\begin{algorithmic}[1]
\State Calculate $T_{\alpha^{\kn}}^{\ph}$ from \eqref{eq:stop_ph}
\State $R^{\ph}=\{i\in [p]:W_i\geq T_{\alpha^{\kn}}^{\ph}\}$
\If{$R^{\ph}\neq \emptyset$}
\State $\tilde{\alpha}^{\ph}=\frac{1+\sum_{i\in [p]} \mathbbm{1}\{W_i\leq -T_{\alpha^{\kn}}^{\ph}\}}{\sum_{i\in [p]} \mathbbm{1}\{W_i\geq T_{\alpha^{\kn}}^{\ph}\}}$
\Else 
\State $\tilde{\alpha}^{\ph}=\alpha^{\kn}$
\EndIf
\State \Return $R^{\ph}$, $\tilde{\alpha}^{\ph}$
\end{algorithmic}
\end{algorithm}

To prove the guarantee of our procedure, we define the local e-values
\begin{align}
    E_S^{\ph}=\frac{\sum_{i\in S}  \mathbbm{1} \{W_i \geq T_{\alpha^{\kn}}^{\ph}\}}{1+\sum_{i\in [p]}\mathbbm{1}\{W_i \leq -T_{\alpha^{\kn}}^{\ph}\}}, S\subseteq [p]. \label{eq: local_e_ph}
\end{align}
The validity of these local e-values (i.e., that $\mathbb{E}_{H_0^S}[E_S^{\ph}]\leq 1$) follows since $T_{\alpha^{\kn}}^{\ph}$ is a stopping time with respect to the filtration generated by a masked version of $(W_i)_{i\in [p]}$, meaning the event $\{T_{\alpha^{\kn}}^{\ph}\leq t\}$ is determined by the absolute values $|W_i|$ for all $i$ and  $\sum_{i:|W_i|\geq s} \mathbbm{1}\{W_i>0\}$ for all $s\leq t$ (see Remark~4 in \citep{ren2024derandomised}).

\begin{theorem}\label{theo:posthoc_alpha}
    Let $\bE^{\ph}=(E_S^{\ph})_{S\subseteq [p]}$ be the family of local e-values defined in \eqref{eq: local_e_ph} and set
    \begin{align}
        \tilde{\alpha}^{\ph}=\frac{1+\sum_{i\in [p]} \mathbbm{1}\{W_i\leq -T_{\alpha^{\kn}}^{\ph}\}}{\sum_{i\in [p]} \mathbbm{1}\{W_i\geq T_{\alpha^{\kn}}^{\ph}\}} \quad \text{with the convention } \tilde{\alpha}^{\ph}=\alpha^{\kn} \text{ if } \sum_{i\in [p]} \mathbbm{1}\{W_i\geq T_{\alpha^{\kn}}^{\ph}\}=0. \notag
    \end{align}
    Then $R^{\ph}\in \mcR_{\tilde{\alpha}^{\ph}}(\bE^{\ph})$, where
    $R^{\ph}=\{i\in [p]:W_i\geq T_{\alpha^{\kn}}^{\ph}\}$. 
    Therefore, it holds that 
    \begin{align}
    \mathbb{E}\left[ \frac{\FDP_{I_0}(R^{\ph})}{\tilde{\alpha}^{\ph}}\right] \leq 1. \label{eq:ph_kn_guarantee}
\end{align} 
Furthermore, $R^{\ph}\supseteq R^{\BC}$, meaning $\ph$-knockoff always makes as least as many rejections as $\BC$-knockoff and $\tilde{\alpha}^{\ph}\leq \alpha^{\kn}$ on the event $\{R^{\BC}\neq \emptyset\}$, meaning if $\BC$-knockoff rejects any hypothesis then $\ph$-knockoff does not report a larger significance level.
\end{theorem}
\begin{proof}
    We need to show that $E_S^{\ph}\geq \FDP_S(R^{\ph})/\tilde{\alpha}^{\ph}$ for all $S\subseteq [p]$. First, note that in case of $\sum_{i\in [p]} \mathbbm{1}\{W_i\geq T_{\alpha^{\kn}}^{\ph}\}=0$, we have $R^{\ph}=\emptyset$ and therefore the condition is trivially fulfilled. Hence, assume $R^{\ph}\neq \emptyset$. Then,
\begin{align}
    E_S^{\ph}=\frac{|S\cap R^{\ph}|}{1+\sum_{i\in [p]}\mathbbm{1}\{W_i \leq -T_{\alpha^{\kn}}^{\ph}\}}=\frac{\FDP_S(R^{\ph}) |R^{\ph}|}{1+\sum_{i\in [p]}\mathbbm{1}\{W_i \leq -T_{\alpha^{\kn}}^{\ph}\}}=\frac{\FDP_S(R^{\ph})}{\tilde{\alpha}^{\ph}}. \notag
\end{align}
    
\end{proof}

\begin{remark}
Note that $\tilde{\alpha}^{\ph}$ equals $1$ in case of $\sum_{i\in [p]} \mathbbm{1}\{W_i\geq T_{\alpha^{\kn}}^{\ph}\}=1$. While rejections at a level of $1$ seems like a trivial claim in the fixed-$\alpha$ regime, it is not for random $\alpha$ \citep{koning2023post}. However, if such large values for $\tilde{\alpha}^{\ph}$ should be avoided, one could just set $\tilde{\alpha}^{\ph}=\alpha^{\kn}$ and $R^{\ph}=\emptyset$ in those cases. This obviously maintains \eqref{eq:ph_kn_guarantee}.
\end{remark}



\section{Post-hoc \texorpdfstring{$\alpha$}{alpha} with derandomized knockoffs\label{sec:derandomized}}

\citet{ren2024derandomised} introduced a derandomized knockoff procedure, which we call $\RB$-knockoff in the following, by generating multiple knockoff copies and merging the results using an average of (compound) e-values. In the following, we briefly recap their method and demonstrate how it can be improved by using a post-hoc significance level. It should be noted that we present the $\RB$-knockoff procedure with our stopping time $T_{\alpha^{\kn}}^{\ph}$ to be consistent with the previous part, even though \citet{ren2024derandomised} proposed a slightly different one (see Remark~\ref{remark:RB_stop}). However, our approach generally works for any stopping time.

Let $W_i^{(1)},\ldots,W_i^{(k)}$, $i\in [p]$, be $k$ knockoff statistics for the $i$-th variable. It is only required that the vector $W^{(j)}=(W_1^{(j)}, \ldots, W_p^{(j)})$ satisfies the property in Fact~\ref{fact:candes} for each $j\in [k]$, but the dependence structure between the different $W^{(j)}$ can be arbitrary. Now let $T_{\alpha^{\kn}}^{\ph,(j)}$ be the stopping time \eqref{eq:stop_ph} applied on $W^{(j)}$ and define for each $i\in [p]$,
\begin{align}
    E_i^{\text{avg}}=\frac{1}{k} \sum_{j=1}^k E_i^{(j)}, \quad \text{where } E_i^{(j)}=p\frac{\mathbbm{1}\{W_i^{(j)}\geq T_{\alpha^{\kn}}^{\ph,(j)}\}}{1+\sum_{l\in [p]} \mathbbm{1}\{W_l^{(j)}\leq -T_{\alpha^{\kn}}^{\ph,(j)}\} }. \label{eq:e-values_derand}
\end{align}
\citet{ren2024derandomised} showed that $E_1^{\text{avg}},\ldots, E_p^{\text{avg}}$ are compound e-values, which means $\mathbb{E} \left[ \sum_{i\in I_0} E_i^{\text{avg}} \right] \leq p$. The $\RB$-knockoff procedure is defined by e-BH \citep{wang2022false} applied on $E_1^{\text{avg}},\ldots, E_p^{\text{avg}}$ at prespecified level $\alpha^{\ebh}\in (0,1]$ and controls the FDR at level $\alpha^{\ebh}$. 

In case of $k=1$ and $\alpha^{\ebh}=\alpha^{\kn}$, $\RB$-knockoff coincides with $\BC$-knockoff \citep{ren2024derandomised}. However, \citet{ren2024derandomised} showed that $\alpha^{\ebh}=\alpha^{\kn}$ often leads to zero power in case of $k>1$ and therefore recommend to choose $\alpha^{\ebh}>\alpha^{\kn}$. The choice of the parameter $\alpha^{\ebh}$ is crucial. If $\alpha^{\ebh}$ is too small, one might not be able to make any rejections even though the signal is strong. If $\alpha^{\ebh}$ is too large, the results become imprecise. \citet{ren2024derandomised} derive some concrete choices like $\alpha^{\ebh}=2\alpha^{\kn}$ based on simulations and heuristics, but they do not propose an optimal parameter. We note that it is possible to choose $\alpha^{\ebh}$ completely data-dependent while maintaining \eqref{eq:post_hoc_alpha}. In this way, we can just choose the $\alpha^{\ebh}$ that leads to the most desirable rejection set. 
\begin{theorem}\label{theo:derandomized}
    Let $E_{(1)}^{\text{avg}}\geq \ldots \geq  E_{(p)}^{\text{avg}}$ be the ordered compound e-values defined in \eqref{eq:e-values_derand} and $R_{\alpha^{\ebh}}^{\RB}=\{i\in [p]: E_i^{\text{avg}} \geq p/(\alpha^{\ebh} i^{\ebh})\}$, where $i^{\ebh}=\max\{i\in [p]:E_{(i)}^{\text{avg}}\geq p/(\alpha^{\ebh} i)\}$ with the convention $\max(\emptyset)=0$. Then,
    \begin{align}
        \mathbb{E}\left[ \sup_{\alpha^{\ebh}\in (0,1]} \frac{\FDP_{I_0}(R_{\alpha^{\ebh}}^{\RB})}{\alpha^{\ebh}}  \right]\leq 1. \notag
    \end{align}
    Hence, the parameter $\alpha^{\ebh}$ can be chosen based on the data while maintaining FDR control in the sense of \eqref{eq:post_hoc_alpha}.
\end{theorem}
The proof of the Theorem~\ref{theo:derandomized} follows immediately from the equivalence of e-BH with compound e-values $E_1^{\text{avg}},\ldots, E_p^{\text{avg}}$ and the e-Closure Principle with the following local e-values (see Theorem 23 in \citep{xu2025bringing})
\begin{align}
    E_S=\frac{1}{p} \sum_{i\in S} E_i^{\text{avg}}. \label{eq:local_avg}
\end{align}

There are multiple reasonable choices to pick $\alpha^{\ebh}$ based on the data that could lead to different rejection sets (in the following we refer to our data-dependent choice as $\tilde{\alpha}^{\dph}$). In order to compare our derandomized $\ph$-knockoff procedure  to $\RB$-knockoff in the following, we start with some initial level $\alpha^{\ebh}\in [0,1]$ that should be chosen as in the $\RB$-knockoff procedure. If $\RB$-knockoff does not make any rejections with that $\alpha^{\ebh}$, we choose $\tilde{\alpha}^{\dph}>\alpha^{\ebh}$ as small as possible such that $R_{\tilde{\alpha}^{\dph}}^{\RB}$ is not empty. If $R_{\alpha^{\ebh}}^{\RB}\neq \emptyset$, then we choose $\tilde{\alpha}^{\dph}\leq \alpha^{\ebh}$ as small as possible while $R_{\tilde{\alpha}^{\dph}}^{\RB}=R_{\alpha^{\ebh}}^{\RB}$. In this way, our derandomized $\ph$-knockoff procedure always rejects as least as many hypotheses as $\RB$-knockoff. Furthermore, if $\RB$-knockoff makes rejections, then $\tilde{\alpha}^{\dph}\leq \alpha^{\ebh}$. Our general procedure is summarized in Algorithm~\ref{alg:posthoc-alpha_derand}.


\begin{remark}
    In case of $k=1$ the derandomized $\ph$-knockoff method (Algorithm~\ref{alg:posthoc-alpha_derand}) and the $\ph$-knockoff procedure (Algorithm~\ref{alg:posthoc-alpha}) are equivalent regardless of the chosen $\alpha_{\ebh}$. To see this, note that for $k=1$ we have $$E_i^{\text{avg}}=p\frac{\mathbbm{1}\{W_i \geq T_{\alpha^{\kn}}^{\ph}\}}{1+\sum_{l\in [p]} \mathbbm{1}\{W_l\leq -T_{\alpha^{\kn}}^{\ph}\}}.$$
    Hence, if $R_{\alpha^{\ebh}}^{\RB}\neq \emptyset$, then $R^{\dph}=R_{\alpha^{\ebh}}^{\RB}=\{i\in [p]: W_i\geq T_{\alpha^{\kn}}^{\ph}\}=R^{\ph}$ and $\tilde{\alpha}^{\dph}=p/(i^{\ebh} E_{(i^{\ebh})}^{\text{avg}})= \frac{1+\sum_{i\in [p]} \mathbbm{1}\{W_i\leq -T_{{\alpha^{\kn}}}^{\ph}\}}{\sum_{i\in [p]} \mathbbm{1}\{W_i\geq T_{{\alpha^{\kn}}}^{\ph}\}}=\tilde{\alpha}^{\ph}$. Furthermore, if $R_{\alpha^{\ebh}}^{\RB}= \emptyset$, it holds that $i^*=\sum_{i\in [p]} \mathbbm{1}\{W_i \geq T_{{\alpha^{\kn}}}^{\ph}\}$, which implies $\tilde{\alpha}^{\dph}=\tilde{\alpha}^{\ph}$ and $R^{\dph}=R^{\ph}$ in the same manner as before. Hence, the post-hoc framework unifies the non-derandomized and the derandomized knockoff method.
\end{remark}

\begin{algorithm}[H]
\caption{Derandomized $\ph$-knockoff}\label{alg:posthoc-alpha_derand}
 \textbf{Input:} Knockoff level $\alpha^{\kn}\in (0,1]$, initial level $\alpha_{\ebh}\in [0,1]$ and knockoff statistics $W_1^{(j)},\ldots,W_p^{(j)}$ for each $j\in [k]$.\\ 
 \textbf{Output:} Rejection set $R^{\dph}$ and post-hoc level $\tilde{\alpha}^{\dph}$.
\begin{algorithmic}[1]
\State Calculate $E_1^{\text{avg}},\ldots, E_p^{\text{avg}}$ as in \eqref{eq:e-values_derand}
\State Set $i^{\ebh}=\max\{i\in [p]:E_{(i)}^{\text{avg}}\geq p/(\alpha^{\ebh} i)\}$
\State Define the $\RB$-knockoff rejection set $R_{\alpha^{\ebh}}^{\RB}=\{i\in [p]: E_i^{\text{avg}} \geq p/(\alpha^{\ebh} i^{\ebh})\}$
\If{$R_{\alpha^{\ebh}}^{\RB}\neq \emptyset$}
\State $\tilde{\alpha}^{\dph}=p/(i^{\ebh} E_{(i^{\ebh})}^{\text{avg}}) $
\State $R^{\dph}=R_{\alpha^{\ebh}}^{\RB}$
\Else
\State $i^*=\argmax_{i\in [p]} i E_{(i)}^{\text{avg}}$  (take largest index if maximum is not unique)
\If{$i^* E_{(i^*)}^{\text{avg}} < p$}
\State $\tilde{\alpha}^{\dph}=\alpha^{\kn}$
\State $R^{\dph}=\emptyset$
\Else
\State $\tilde{\alpha}^{\dph}=p/(i^* E_{(i^*)}^{\text{avg}})$ 
\State $R^{\dph}=\{i\in [p]: E_i^{\text{avg}}\geq p/(i^*\tilde{\alpha}^{\dph}) \}$
\EndIf
\EndIf
\State \Return $R^{\dph}$, $\tilde{\alpha}^{\dph}$
\end{algorithmic}
\end{algorithm}


\subsection{Choosing the rejection set before the \texorpdfstring{$\alpha$}{alpha}\label{sec:rej_first}}

Until now, we have still followed the classic order in multiple testing, where one first defines the $\alpha$ (in our case, possibly data-adaptive) and then determines the largest possible rejection set. In practice, it may well be interesting to reverse this. This means that, based on the data and contextual information like explainable relevance of a variable, we determine a rejection set $R$ that is of interest to us. Then, we choose the smallest $\alpha$ for which $R$ is contained in the closure $\mcR_{\alpha}(\bE)$ (see \eqref{eq:closure_set} for the definition of the closure set). Due to the simultaneous control of the e-Closure Principle over rejection sets and significance levels \eqref{eq:sup_alpha}, this combination of $R$ and $\alpha$ provides post-hoc FDR control. In the following proposition we derive the smallest such $\alpha$ for the derandomized knockoffs above.
\begin{proposition}
\label{proposition:rej_first}
    Let $\bE$ be the family of e-values defined by \eqref{eq:local_avg} and $R$ be a rejection set depending on the data in an arbitrary way. Then,
    \begin{align}
        \inf\left\{\alpha >0: R\in \mcR_{\alpha}(\bE)\right\}= \frac{p}{|R| \min_{i\in R} E_i^{\text{avg}}} \eqqcolon \tilde{\alpha}^R  \notag
    \end{align}
    with the convention $\tilde{\alpha}^R=\infty$ if $\min_{i\in R} E_i^{\text{avg}}=0$. Therefore, we have
    \begin{align}
        \mathbb{E}\left[\frac{\FDP_{I_0}(R)}{\tilde{\alpha}^R} \right]\leq 1. \notag
    \end{align}
\end{proposition}
\begin{proof}
   We want to find the smallest $\alpha>0$ such that $R\in \mcR_{\alpha}(\bE)$, which is equivalent to
   \begin{align}
    \frac{1}{p} \sum_{i\in S} E_i^{\text{avg}} \geq \frac{\FDP_S(R)}{\alpha} \text{ for all } S\subseteq [p]. \notag
   \end{align}
   Since adding indices $i$ with $i\notin R$ to some $S$ keeps the right-hand side of the inequality the same but increases the left-hand side of the inequality, we can restrict to $S\subseteq R$. Let $E_{(1:R)}\geq \ldots \geq E_{(|R|:R)}$ be the sorted e-values with index in $R$. Then we can further reduce the property $R\in \mcR_{\alpha}(\bE)$ to 
   \begin{align} 
   \alpha \geq \frac{p\cdot s}{|R| \sum_{i=1+|R|-s}^{|R|} E_{(i:R)}} \text{ for all } s\in [|R|]. \notag
   \end{align}
   Since the right-hand side of the inequality is maximized for $s=1$, the assertion follows.
\end{proof}

In Section~\ref{sec:clinical_trial_predictive}, we illustrate this idea of choosing the rejection set first based on a real data trial.

\subsection{A uniform improvement of the derandomized \texorpdfstring{($\ph$-)}{ph-}knockoff procedure\label{sec:uniform_improvement}}
In this paragraph, we introduce a uniform improvement of the derandomized knockoff procedure for fixed $\alpha^{\ebh}$ by \citet{ren2024derandomised} and therefore also of the derandomized $\ph$-knockoff procedure. However, this uniform improvement is based on the e-Closure Principle and we have not found a short-cut. This makes this uniform improvement computationally inefficient, which is why we still recommend Algorithm~\ref{alg:posthoc-alpha_derand} for practical purposes. 

To derive the uniform improvement, consider the local e-values $(E_S)_{S\subseteq [p]}$ that recover the $\RB$-knockoff procedure, defined in \eqref{eq:local_avg}. One can write $E_S$, $S\subseteq [p]$, as
\begin{align}
    E_S=\frac{1}{k}\left(\sum_{j=1}^k \frac{\sum_{i\in S} \mathbbm{1}\{W_i^{(j)} \geq T_{{\alpha^{\kn}}}^{\ph, (j)}\} }{1+\sum_{i\in [p]} \mathbbm{1}\{W_i^{(j)}\leq -T_{{\alpha^{\kn}}}^{\ph, (j)} \}} \right). \notag
\end{align}
Now, due to the proof for FDR control of $\BC$-knockoff by \citet{barber2015controlling}, it follows that one can also define local e-values by changing $[p]$ in the denominators to $S$:
\begin{align}
    \tilde{E}_S=\frac{1}{k}\left(\sum_{j=1}^k \frac{\sum_{i\in S} \mathbbm{1}\{W_i^{(j)} \geq T_{{\alpha^{\kn}}}^{\ph, (j)}\} }{1+\sum_{i\in S} \mathbbm{1}\{W_i^{(j)}\leq -T_{{\alpha^{\kn}}}^{\ph, (j)} \}} \right). \label{eq:improved_derand}
\end{align}
Obviously, $\tilde{E}_S\geq E_S$ for all $S\subseteq [p]$ and thus the e-Closure Principle with the local e-values $(\tilde{E}_S)_{S\subseteq [p]}$, called closed knockoff in the following, never rejects less hypotheses than the $\RB$-knockoff procedure. This was already noted by \citet{xu2025bringing} for $k=1$ (knockoff without derandomization). However, while they showed that in case of $k=1$  closed knockoff is just equivalent to $\BC$-knockoff, it becomes a real uniform improvement for $k>1$ as we illustrate in the following example.

\begin{example}
    Let $p=5$, $k=2$ and $\alpha^{\ebh}=1/2$. Suppose $W_1^{(1)}\leq -T_{{\alpha^{\kn}}}^{\ph, (1)}$, $W_i^{(1)}\geq T_{{\alpha^{\kn}}}^{\ph, (1)}$ for $i\in \{2,3,4,5\}$,  $W_2^{(2)}\leq -T_{{\alpha^{\kn}}}^{\ph, (2)}$ and $W_i^{(2)}\geq T_{{\alpha^{\kn}}}^{\ph, (2)}$ for $i\in \{1,3,4,5\}$. One can easily check that in this case $E_1^{\text{avg}}=E_2^{\text{avg}}=5/4$ and $E_3^{\text{avg}}=E_4^{\text{avg}}=E_5^{\text{avg}}=5/2$ and thus $R_{\alpha^{\ebh}}^{\RB}=\emptyset$. However, $\{3,4,5\} \in \mathcal{R}_{\alpha^{\ebh}}(\boldsymbol{\tilde{E}})$, where $\boldsymbol{\tilde{E}}=(\tilde{E}_S)_{S\subseteq [p]}$ is given by \eqref{eq:improved_derand}. To see this, note that $\min_{S\subseteq [p]: |S\cap \{4,5,6\}|=3} \tilde{E}_S=2 $, $\min_{S\subseteq [p]: |S\cap \{4,5,6\}|=2} \tilde{E}_S=3/2 $ and $\min_{S\subseteq [p]: |S\cap \{4,5,6\}|=1} \tilde{E}_S=1 $ and therefore $\tilde{E}_S\geq 2\FDP_S(\{4,5,6\})$ for all $S\subseteq [p]$.
\end{example}

\subsection{Derandomized knockoffs for PFER control \label{sec:derandomized_PFER}}

\citet{ren2023derandomizing} introduced a derandomized version of a knockoff procedure \citep{janson2016familywise} for PFER control. We call this procedure $\RWC$-knockoff in the following.  The PFER is defined as the expected number of false discoveries $\PFER =\mathbb{E}[|R\cap I_0|]=\mathbb{E}[\NFD_{I_0}(R)]$, where $\NFD_{I_0}(R)=|R\cap I_0|$ for some rejection set $R$.  

Let $\nu>0$ be an integer and $$T_{\nu}^{\text{PFER}}=\inf\left\{t>0:\sum_{i\in [p]} \mathbbm{1}\{W_i\leq -t\}=\nu-1 \right\} $$ be the stopping time for the PFER knockoff procedure introduced by \citet{janson2016familywise}. Again, we denote by $T_{\nu}^{\text{PFER}, (j)}$, $j\in [k]$, the stopping time in the $j$-th knockoff run. Then $\RWC$-knockoff is given by the rejection set 
\begin{align}
    R_{\eta}^{\RWC}=\left\{i\in [p]: \left(\frac{1}{k}\sum_{j=1}^k \mathbbm{1}\{W_i^{(j)}\geq T_{\nu}^{\text{PFER}, (j)} \} \right) \geq \eta \right\},  \label{eq:rej_PFER}
\end{align}
where $\eta\in (0,1]$ is some prespecified parameter that determines after how many \enquote{rejections} in the different knockoff runs a hypothesis is finally rejected. \citet{ren2023derandomizing} showed that $\RWC$-knockoff controls the PFER at level $\nu/\eta $. Again, there is no optimal choice for $\eta$ and they propose sample splitting to obtain a data-driven parameter. We show that $\eta$ can be chosen entirely based on the data while retaining a sensible guarantee.

\begin{theorem}\label{theo:PFER}
    For some prespecified integer $\nu>0$ and the rejection set $R_{\eta}^{\RWC}$ in \eqref{eq:rej_PFER} it holds that
    \begin{align}
    \mathbb{E}\left[\sup_{\eta\in (0,1]} \NFD_{I_0}(R_{\eta}^{\RWC})\eta  \right]\leq \nu.  \notag  
    \end{align}
\end{theorem}

\begin{proof}
    For the proof we implicitly use a modified version of the e-Closure Principle for PFER control by \citet{xu2025bringing}. 
    Define the local e-values $E_S$, $S\subseteq [p]$, as in \eqref{eq:local_avg} but with the stopping times $T_{\nu}^{\text{PFER}, (j)}$ instead of $T_{\alpha^{\kn}}^{\ph,(j)}$. To see that $\mathbb{E}_{H_0^S}[E_S]\leq 1$, see for example Remark 4 in \citep{ren2024derandomised}. For any $\eta \in (0,1]$ it holds that
    \begin{align}
        \frac{\eta}{\nu} \NFD_{I_0}(R_{\eta}^{\RWC}) &= \frac{\eta}{\nu} \sum_{i\in I_0} \mathbbm{1}\left\{\left(\frac{1}{k}\sum_{j=1}^k \mathbbm{1}\{W_i^{(j)}\geq T_{\nu}^{\text{PFER}, (j)} \} \right) \geq \eta \right\} \notag  \\
        &\leq \frac{1}{\nu} \sum_{i\in I_0} \frac{1}{k} \sum_{j=1}^k \mathbbm{1}\{W_i^{(j)}\geq T_{\nu}^{\text{PFER}, (j)} \} = E_{I_0}, \notag
    \end{align}
    where the inequality follows from $\mathbbm{1}\{x\geq 1\} \leq x$ for all $x\geq 0$. Since $\mathbb{E}[E_{I_0}]\leq 1$, the claim follows.
\end{proof}
Hence, we can choose $\tilde{\eta}$ based on the data while maintaining
    \begin{align}
    \mathbb{E}\left[\frac{\NFD_{I_0}(R_{\tilde{\eta}}^{\RWC})\tilde{\eta}}{\nu}\right]\leq 1. \label{eq:post_hoc_eta}
\end{align} 

One possibility to choose $\tilde{\eta}$ in practice would be to just calculate $R_{\eta}^{\RWC}$ for every $\eta \in \{r/k: r\in [k]\}$ and then choose $\tilde{\eta}=\eta$ for the $\eta$ that leads to the most desirable $(R_{\eta}^{\RWC}, \eta)$ pair. In Algorithm~\ref{alg:posthoc-PFER} we propose to choose the one that maximizes the product $R_{\eta}^{\RWC}\eta$. It should be noted that this implementation does not guarantee a \enquote{free-lunch} improvement over the original method, rejection set is not always larger. However, if the $(R_{\eta}^{\RWC}, \eta)$ combination of the original method seems more desirable, we could also switch to that post-hoc.

\begin{algorithm}
\caption{Derandomized $\ph$-knockoff for PFER control}\label{alg:posthoc-PFER}
 \textbf{Input:} Knockoff level $\nu \in \mathbb{N}$ and knockoff statistics $W_1^{(j)},\ldots,W_p^{(j)}$ for each $j\in [k]$.\\ 
 \textbf{Output:} Rejection set $R^{\phnu}$ and post-hoc level $\tilde{\eta}^{\phnu}$.
\begin{algorithmic}[1]
\State Define $\mathcal{K}=\{r/k: r\in [k]\}$
\State Define $\tilde{\eta}^{\phnu}=\argmax_{\eta \in 
 \mathcal{K}} \{R_{\eta}^{\RWC} \eta\}$, where $R_{\eta}^{\RWC}$ is defined as in \eqref{eq:rej_PFER}
 \State Set \begin{align*}
    R^{\phnu}=\left\{i\in [p]: \left(\frac{1}{k}\sum_{j=1}^k \mathbbm{1}\{W_i^{(j)}\geq T_{\nu}^{\text{PFER}, (j)} \} \right) \geq \tilde{\eta}^{\phnu} \right\}  
\end{align*}

\If{$R^{\phnu} = \emptyset$}
\State $\tilde{\eta}^{\phnu}=1/2$

\EndIf
\State \Return $R^{\phnu}$, $\tilde{\eta}^{\phnu}$
\end{algorithmic}
\end{algorithm}

\begin{remark}
    In the same way as in Section~\ref{sec:uniform_improvement} we can also uniformly improve $\RWC$-knockoff by using $(\tilde{E}_S)_{S\subseteq [p]}$ in \eqref{eq:improved_derand} with $T_{\nu}^{\text{PFER}, (j)}$ instead of $T_{\alpha^{\kn}}^{\ph,(j)}$ and then apply the e-Closure Principle for PFER control by \citet{xu2025bringing}. 
\end{remark}

Furthermore, as described in Section~\ref{sec:rej_first} for FDR control, we can also choose the rejection set first and then select the $\eta$ afterwards. 
\begin{proposition}
    Let $R$ be a rejection set depending on the data in an arbitrary way and $$\tilde{\eta}^R\coloneqq \min_{i\in R} \frac{1}{k} \sum_{j=1}^k \mathbbm{1}\{W_i^{(j)} \geq T_{\nu}^{\text{PFER}, (j)} \}.$$
    Then $\mathbb{E}[\NFD_{I_0}(R)\tilde{\eta}^R]\leq \nu$.
\end{proposition}
\begin{proof}
    Let $E_{I_0}$ be defined as in the proof of Theorem~\ref{theo:PFER}. Then,
    \begin{align}
        \mathbb{E}[\NFD_{I_0}(R)\tilde{\eta}^R/\nu]&=
        \mathbb{E}\left[\frac{1}{\nu}\sum_{i\in I_0} \mathbbm{1}\{i\in R\} \min_{l\in R} \frac{1}{k} \sum_{j=1}^k \mathbbm{1}\{W_l^{(j)} \geq T_{\nu}^{\text{PFER}, (j)} \}\right] \notag \\
        &\leq \mathbb{E}\left[\frac{1}{\nu}\sum_{i\in I_0} \mathbbm{1}\{i\in R\}  \frac{1}{k} \sum_{j=1}^k \mathbbm{1}\{W_i^{(j)} \geq T_{\nu}^{\text{PFER}, (j)} \}\right] \notag  \\
        &= \mathbb{E}\left[\frac{1}{\nu}\sum_{i\in I_0} \frac{1}{k} \sum_{j=1}^k \mathbbm{1}\{W_i^{(j)} \geq T_{\nu}^{\text{PFER}, (j)} \}\right] \notag  \\
        &= \mathbb{E}\left[E_{I_0}\right] \leq 1. \notag
    \end{align}
\end{proof}





\section{Simulated experiments\label{sec:sims}}

\subsection{Generating data}
\label{seq:sim_generating_data}
Our simulation framework builds on the setup of \citet{ren2024derandomised}. 
As in the original work, we simulate data from both a Gaussian linear model and a logistic regression model. 
To reflect the two simulation regimes explored in Section~\ref{sec:sims_results}, low-dimensional and high-dimensional, we describe each separately below.

Across all experiments, the variable vector $\mathbf{X} \in \mathbb{R}^p$ is generated as
\[
\mathbf{X} \sim \mathcal{N}(0,\Sigma), \qquad \Sigma_{jk} = 0.5^{|j-k|}.
\]
The response conditional on $\mathbf{X}$ is drawn from either
\[
\textnormal{Gaussian model: } Y \mid \mathbf{X} \sim \mathcal{N}(\mathbf{X}^\top \beta, 1), 
\qquad
\textnormal{Logistic model: } 
Y \mid \mathbf{X} \sim \textnormal{Bernoulli}\!\left(\frac{e^{\mathbf{X}^\top \beta}}{1 + e^{\mathbf{X}^\top \beta}} \right).
\]

\paragraph{Construction of $\beta$.}
To define the regression coefficients, we first sample a signal vector 
$\bar{\beta} \in \mathbb{R}^{p_{\text{relevant}}}$ with i.i.d. entries from $\mathcal{N}(A,1)$, 
where $A$ controls signal amplitude. 
The full vector $\beta \in \mathbb{R}^p$ is obtained by embedding these non-zero entries into $p$ coordinates at approximately equally spaced locations, alternating their signs and scaling them by $\sqrt{n}$:
\[
\beta =
\big(
\underbrace{0, \ldots, 0}_{z}, \tfrac{\bar{\beta}_1}{\sqrt{n}},
\underbrace{0, \ldots, 0}_{z}, -\tfrac{\bar{\beta}_2}{\sqrt{n}}, \ldots
\big).
\]
The spacing parameter $z$ is chosen as a function of $p$ and $p_{\text{relevant}}$, ensuring that signals are distributed evenly throughout the variable vector. 
As in \citet{ren2024derandomised}, this alternating-sign, evenly spaced construction reduces correlation between active variables while maintaining signal heterogeneity.

\subsubsection{Low-dimensional regime (Sections~\ref{sec:sims_original}, \ref{sec:sims_derandomized_PFER} and \ref{sec:sims_calibrated})}
To investigate settings where classical knockoffs experience loss in power, we first consider the low-dimensional case, where
\[
p = 50, \qquad n = 250, \qquad p_{\text{relevant}} \in \{2,\ldots,10\}.
\]
For each choice of $p_{\text{relevant}}$, the spacing $z$ is adapted to distribute non-null variables evenly.
We vary the signal amplitude on the grid
\[
A \in \{2,\dots,10\} \text{ (Gaussian)}, \qquad
A \in \{6,\dots,16\} \text{ (Logistic)}.
\]
Results for this regime are presented in:
\begin{itemize}[itemsep=0.1pt,parsep=0pt,topsep=2pt]
    \item Figures~\ref{fig:posthoc_vs_traditional_p_relevant} and \ref{fig:posthoc_vs_traditional_amplitude},  where we compare our post-hoc knockoffs (Algorithm \ref{alg:posthoc-alpha}) to the original knockoff procedure of \citet{candes2018panning}.
    \item Figure~\ref{fig:posthoc_vs_calibration_models}, where we compare our post-hoc knockoffs (Algorithm \ref{alg:posthoc-alpha}) to the calibration knockoffs procedure of \citet{luo2025improving}.
    \item Figure~\ref{fig:posthoc_vs_derandomized_PFER}, where we compare our derandomized post-hoc knockoffs for PFER control (Algorithm \ref{alg:posthoc-PFER})  against the original derandomized procedure for PFER proposed by \citet{ren2023derandomizing}.
\end{itemize}

\subsubsection{High-dimensional regime (Section~\ref{sec:sims_derandomized_FDR})}
To enable direct comparison with the derandomized knockoff procedure of \citet{ren2024derandomised}, we also simulate $n=1000$ samples in the same high-dimensional settings, namely
\[
p = 800 \text{ (Gaussian)}, \qquad p = 600 \text{ (Logistic)}.
\]
The original study considered $p_{\text{relevant}}=80$ with $z=9$ for the Gaussian model and $p_{\text{relevant}}=60$ with $z=11$ for the logistic model.
We adopt the same parameterization, and further study sparser cases by varying
\[
p_{\text{relevant}} \in \{20, 80\} \text{ (Gaussian)}, \qquad
p_{\text{relevant}} \in \{15, 60\} \text{ (Logistic)}.
\]
The signal amplitude ranges are $A \in \{2,\dots,8\}$ for the Gaussian and $A \in \{4,\dots,16\}$ for the Logistic models.

Results for this regime appear in Figure ~\ref{fig:posthoc_vs_derandomized_FDR}, where we benchmark our derandomized post-hoc procedure for controlling FDR (Algorithm \ref{alg:posthoc-alpha_derand}) against the original derandomized knockoffs for FDR proposed by \citet{ren2024derandomised}.

\subsection{Evaluation measures}
For each model and each signal amplitude setting, the dataset $\mathcal{D} = (\mathbf{X}, Y) $ is generated independently D times, with D = 2000 for the low-dimensinal regime simulations (Sections~\ref{sec:sims_original}, \ref{sec:sims_derandomized_PFER} and \ref{sec:sims_calibrated}), and D = 200 for the high-dimensional (Section~\ref{sec:sims_derandomized_FDR}).  For each draw of the dataset, we apply the knockoffs procedures that we are comparing. For each model, each signal amplitude, and each of the two methods, we measure performance in various different ways. First, for all the methods we estimate the power as
\[
\textnormal{Power} = \frac{1}{D}\sum_{d = 1}^{D} 
\frac{| R_{d} \cap I_1|}{|I_1|},
\]
where $R_{d}$ denotes the rejected set in the $d$-th simulation run, and $I_1$ is the set of \emph{non-null} (or \emph{relevant}) variables.

Next, to assess Type~I error control, we report different metrics depending on whether the procedure targets the FDR or the PFER.

\subsubsection{For the methods that control the FDR (Sections \ref{sec:sims_original}, \ref{sec:sims_derandomized_FDR} and \ref{sec:sims_calibrated})}
We measure the average false discovery proportion (i.e., empirical FDR) across the simulations:
\[
\textnormal{Average FDP} = \frac{1}{D}\sum_{d = 1}^{D} 
\frac{| R_{d} \cap I_0|}{|R_{d}| \vee 1}, 
\]
where $I_0$ is the set of \emph{null} variables. Furthermore,  we consider the deviation of the ratio FDP/${\alpha}$ from $1$ to validate that  eq. \eqref{eq:post_hoc_alpha} is indeed satisfied and investigate the tightness of the control.  For the traditional knockoffs ${\alpha}$ is the prespecified parameter level $\alpha^{\kn}$, for the derandomized knockoffs it is $\alpha^{\ebh},$  while for our post-hoc approaches it is the data-dependent level $\tilde{\alpha}^{\ph}$ without derandomization and $\tilde{\alpha}^{\dph}$ in the derandomized case, returned by Algorithms \ref{alg:posthoc-alpha} and \ref{alg:posthoc-alpha_derand} respectively.
More precisely, we consider
\begin{align}
   \textnormal{Average ratio FDP/}{\alpha} = \frac{1}{D}\sum_{d = 1}^{D}  \frac{| R_{d} \cap I_0|}{|R_{d}| \vee 1}  \frac{1}{{\alpha}_d}, \notag
\end{align} 
where $\alpha_d$ is the significance level $\alpha$ obtained in the $d$-th simulation run.

\subsubsection{For the methods that control the PFER (Section \ref{sec:sims_derandomized_PFER})}
We measure the average number of false discoveries across the simulations:
\[
\textnormal{Average FD} = \frac{1}{D}\sum_{d = 1}^{D} 
{| R_{d} \cap I_0|}.
\]
Furthermore,  we consider the deviation of the ratio 
$\frac{FD {\eta}}{\nu}$ from $1$ to validate that  eq. \eqref{eq:post_hoc_eta} is indeed satisfied and investigate the tightness of the control. For the derandomized knockoffs that control the PFER, following the configuration in \citet{ren2023derandomizing}, we set $\eta = 0.50$. For our post-hoc approach, the corresponding level is the data-dependent quantity $\tilde{\eta}^{\ph},$ returned by Algorithm \ref{alg:posthoc-PFER}. More precisely, we consider
\begin{align}
   \textnormal{Average ratio FD} {\eta}/\nu = \frac{1}{100}\sum_{d = 1}^{100}  \frac{| R_{d} \cap I_0|{\eta}_{d}}{\nu}, \notag
\end{align} 
where ${\eta}_{d}$ is the parameter $\eta$ obtained in the $d$-th simulation run.

\subsection{Results}
\label{sec:sims_results}
This section is structured as follows. In Subsection~\ref{sec:sims_original}, we begin by comparing our post-hoc knockoff procedure with the original knockoff filter of \citet{candes2018panning}. In Subsection~\ref{sec:sims_calibrated}, we then compare our method with the recently proposed calibrated knockoff procedure of \citet{luo2025improving}, which overcomes a key limitation of the original approach by mitigating the thresholding phenomenon. Next, in Subsection~\ref{sec:sims_derandomized_FDR}, we compare our derandomized post-hoc variant with the derandomized knockoffs of \citet{ren2024derandomised} designed for FDR control. Finally, in Subsection~\ref{sec:sims_derandomized_PFER}, we complete the analysis by evaluating our method against the derandomized knockoffs of \citet{ren2023derandomizing} for controlling PFER.

\subsubsection{Comparing against original knockoffs  \texorpdfstring{\citep{candes2018panning}}{Candes et al. (2018)}}
\label{sec:sims_original}
In this section we illustrate the performance of our proposed \textbf{post-hoc knockoff} (presented in Algorithm \ref{alg:posthoc-alpha}) and compare to the \textbf{original knockoff} procedure from \citet{candes2018panning}. The simulation setup is the low-dimensional regime presented in  Section~\ref{seq:sim_generating_data}. Through the simulation we set the $\alpha^{\kn}$ level to $0.20$, which corresponds to a Type~I error (FDR) of $0.20$ for the traditional knockoffs, and also serves as the initial significance level for our post-hoc knockoffs (Algorithm~\ref{alg:posthoc-alpha}).

In Figure~\ref{fig:posthoc_vs_traditional_p_relevant}, we evaluate the performance of the different methods as a function of $p_{\text{relevant}}$, with signal strength held fixed. When the number of truly relevant variables is small, traditional methods suffer from very low power. In contrast, post-hoc knockoffs offer the flexibility to substantially increase power by increasing the nominal level when the original method does not make any rejections. Across all settings, the average FDP remains close to the nominal target level $\alpha^{\kn}$, similar to that of the original knockoff procedure. Furthermore, when power is similar between the two methods (e.g., Gaussian setting for $p_{\text{relevant}} > 6$), the post-hoc approach can achieve stricter error control, resulting in effective $\alpha$ values less than $0.20$. In the last row, we plot the average $\textnormal{FDP/}{\alpha}$ ratio, which remains consistently below 1. This confirms that the controlled error rate, as defined in equation (\ref{eq:post_hoc_alpha}), is indeed maintained.

Similar behavior is observed when $p_{\text{relevant}}$ is held constant and the signal amplitude is varied, as shown in Figure~\ref{fig:posthoc_vs_traditional_amplitude}. Post-hoc knockoffs consistently achieve higher power than traditional methods, especially at moderate signal strengths. For high signal amplitudes, our method achieves high power, and at the same time, it can provide more stringent control, i.e., $\alpha<0.20.$ Interestingly, the average FDP remains below the nominal $\alpha$ level in all scenarios. Furthermore, the average $\textnormal{FDP/}{\alpha}$ ratio remains below 1, again confirming that the error control in equation (\ref{eq:post_hoc_alpha}) holds.
\begin{figure}[htb]
\vspace{-0.5cm}
\centering
\subfloat[\textbf{Gaussian}\label{fig:gauss}]{%
    \begin{minipage}[t]{0.45\textwidth}
        \centering
        \includegraphics[width=\textwidth]{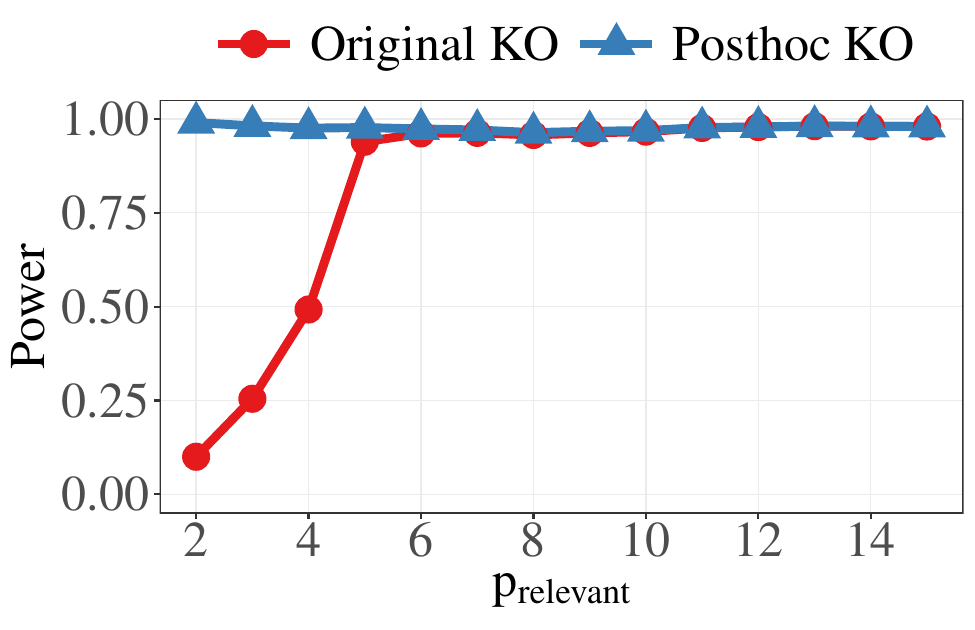}\\
        \vfill
        \includegraphics[width=\textwidth]{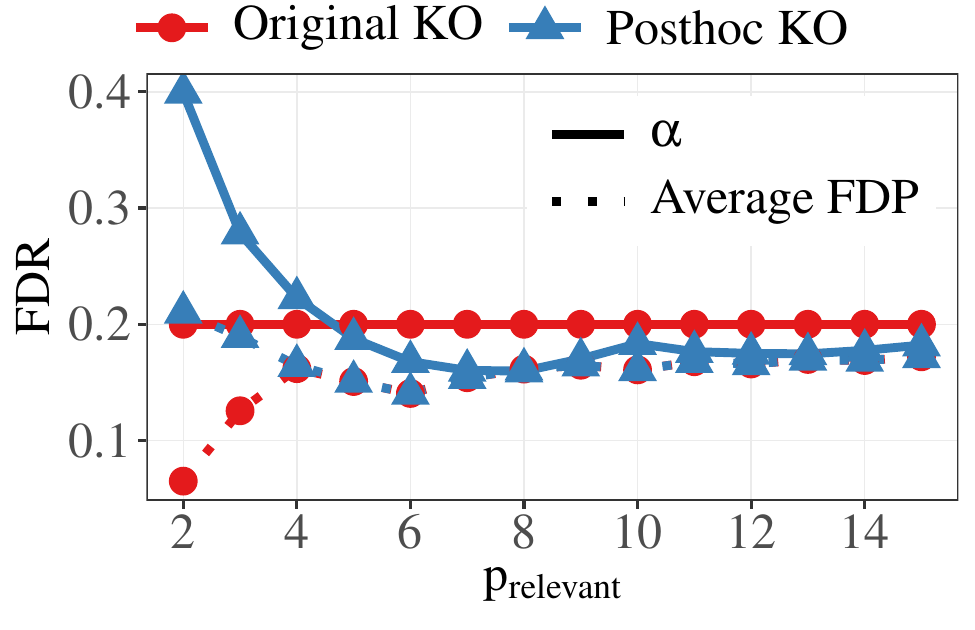}\\
        \vfill
        \includegraphics[width=\textwidth]{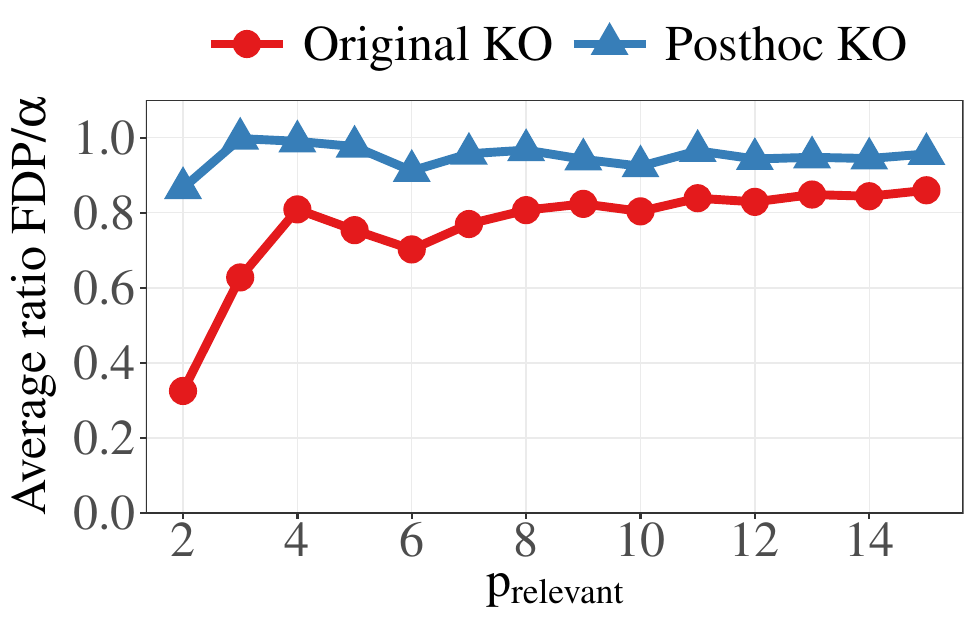}
    \end{minipage}%
}
\ \vrule \ 
\subfloat[\textbf{Logistic}\label{fig:logistic}]{%
    \begin{minipage}[t]{0.45\textwidth}
        \centering
        \includegraphics[width=\textwidth]{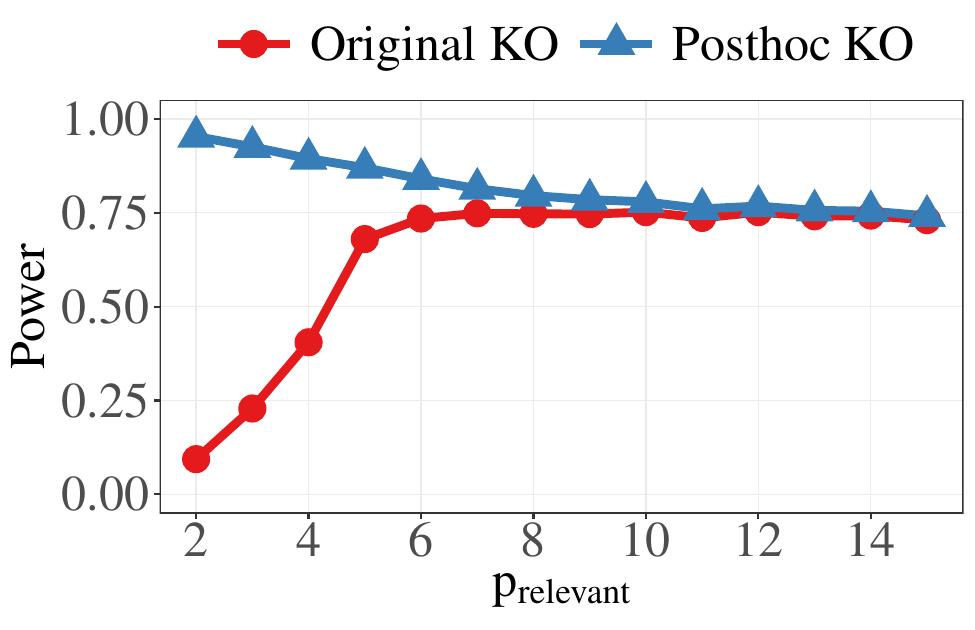}\\
        \vfill
        \includegraphics[width=\textwidth]{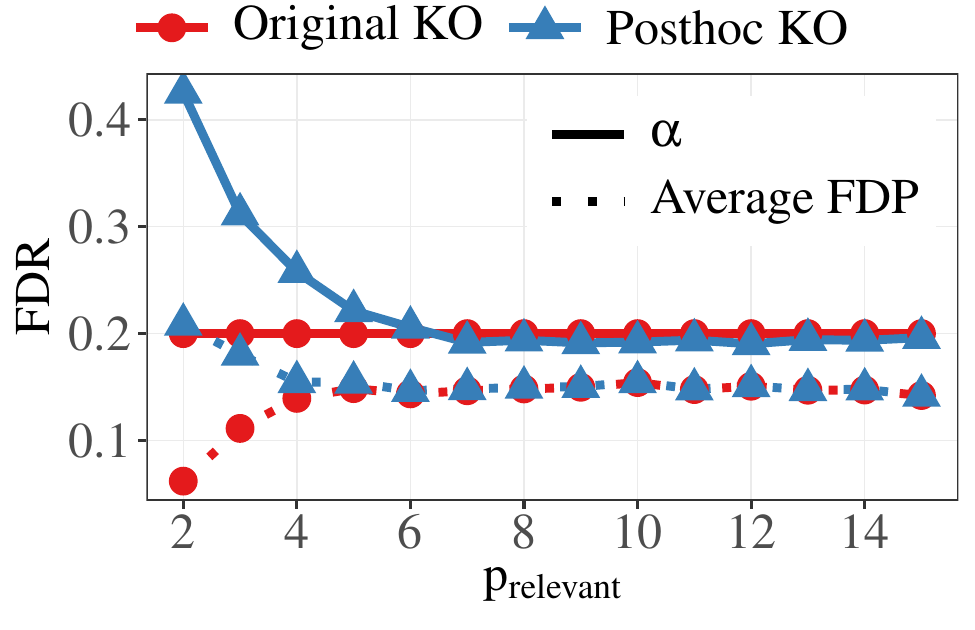}\\
        \vfill
        \includegraphics[width=\textwidth]{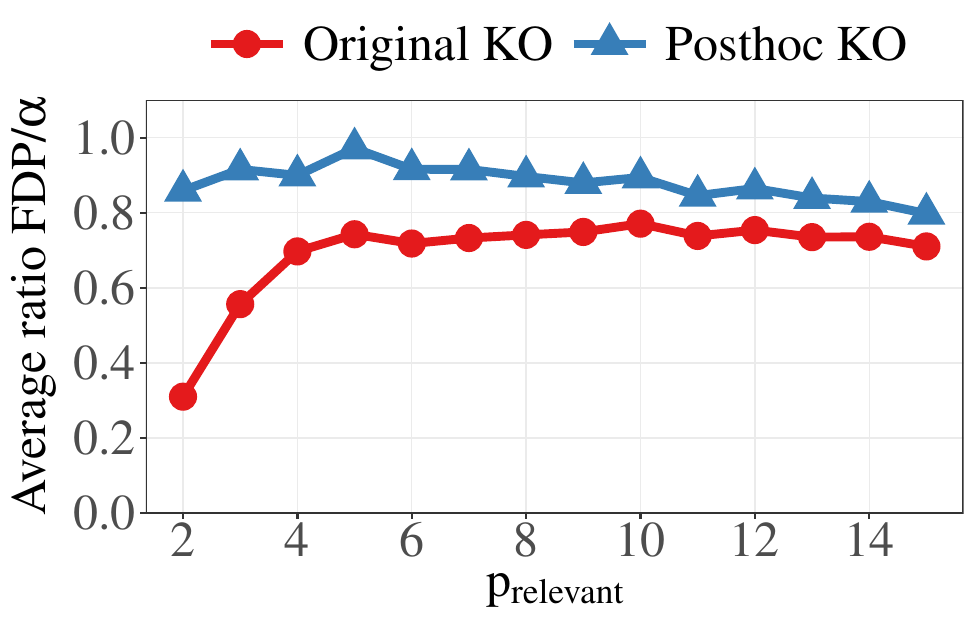}
    \end{minipage}%
}
\caption{Performance comparison between the proposed derandomized post-hoc knockoff procedure and the original derandomized knockoffs method of \citet{candes2018panning} across various numbers of actual relevant variables $p_{\text{relevant}}$. The signal amplitude is fixed at $A = 8$ and $A = 14$ for the Gaussian and Logistic model respectively.  The first row reports the realized power, the second row reports the nominal level $\alpha$ (for the original method this is the user specified level $\alpha^{\kn}$, whereas for the post-hoc knockoff procedure it is the data-dependent level $\tilde{\alpha}^{\ph}$ returned by Algorithm~\ref{alg:posthoc-alpha}) together with the estimated average FDP, and the last row reports the average ratio $\textnormal{FDP/}{\alpha}$. All values are averaged over 2000 runs.} 
\label{fig:posthoc_vs_traditional_p_relevant}
\vspace{-0.5cm}
\end{figure}
\begin{figure}[htb]
\vspace{-0.5cm}
\centering
\subfloat[\textbf{Gaussian}\label{fig:gauss_ampl}]{%
    \begin{minipage}[t]{0.45\textwidth}
        \centering
        \includegraphics[width=\textwidth]{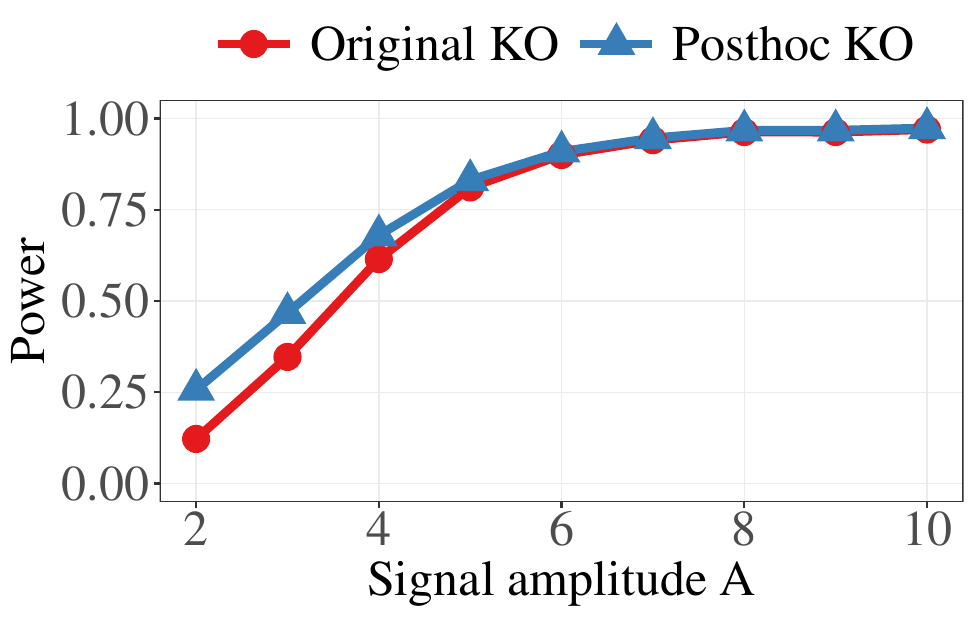}\\
        \vfill
        \includegraphics[width=\textwidth]{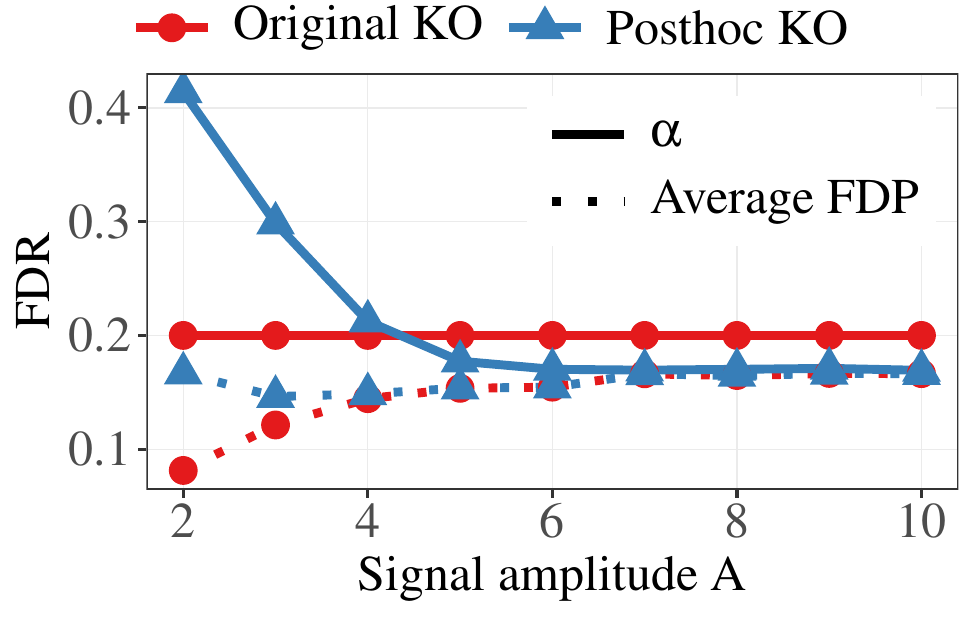}\\
        \vfill
        \includegraphics[width=\textwidth]{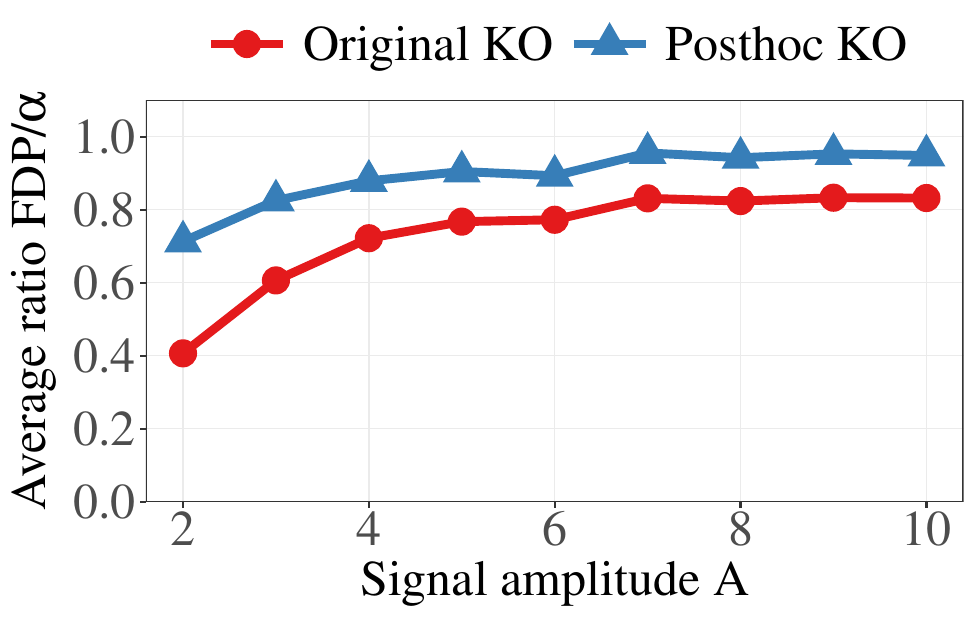}
    \end{minipage}%
}
\ \vrule \ 
\subfloat[\textbf{Logistic}\label{fig:logistic_ampl}]{%
    \begin{minipage}[t]{0.45\textwidth}
        \centering
        \includegraphics[width=\textwidth]{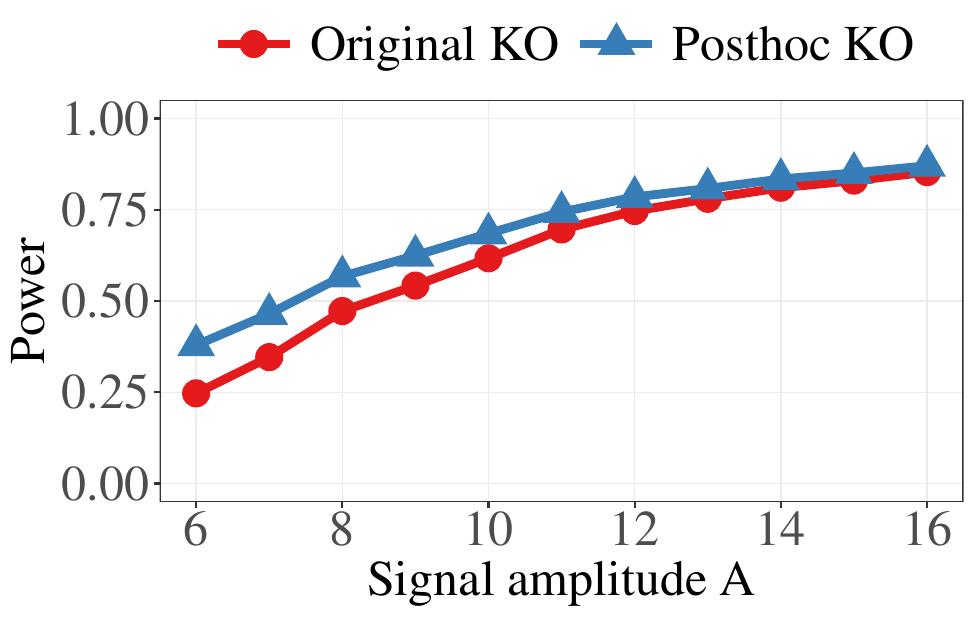}\\
        \vfill
        \includegraphics[width=\textwidth]{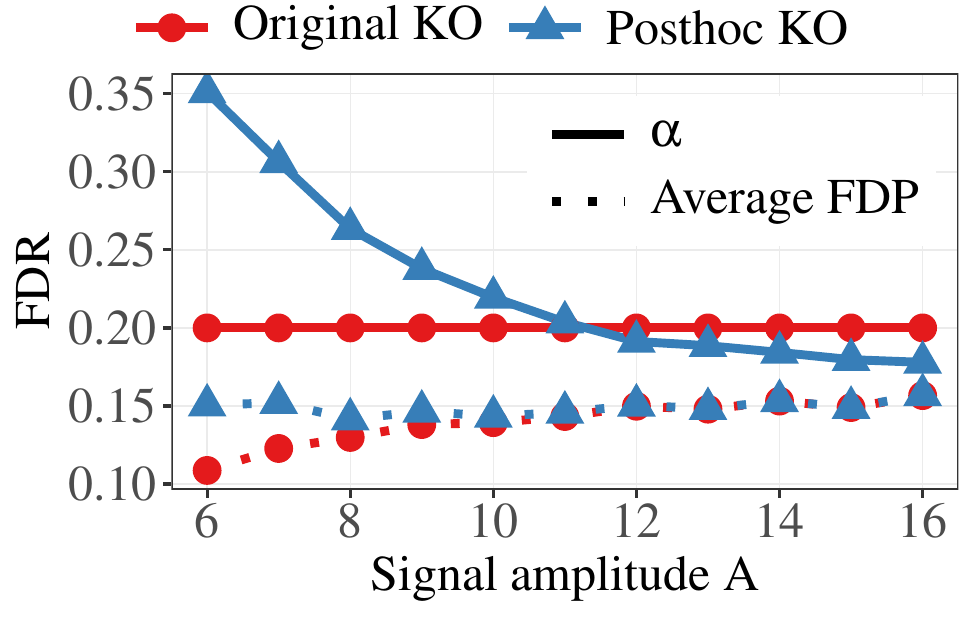}\\
        \vfill
        \includegraphics[width=\textwidth]{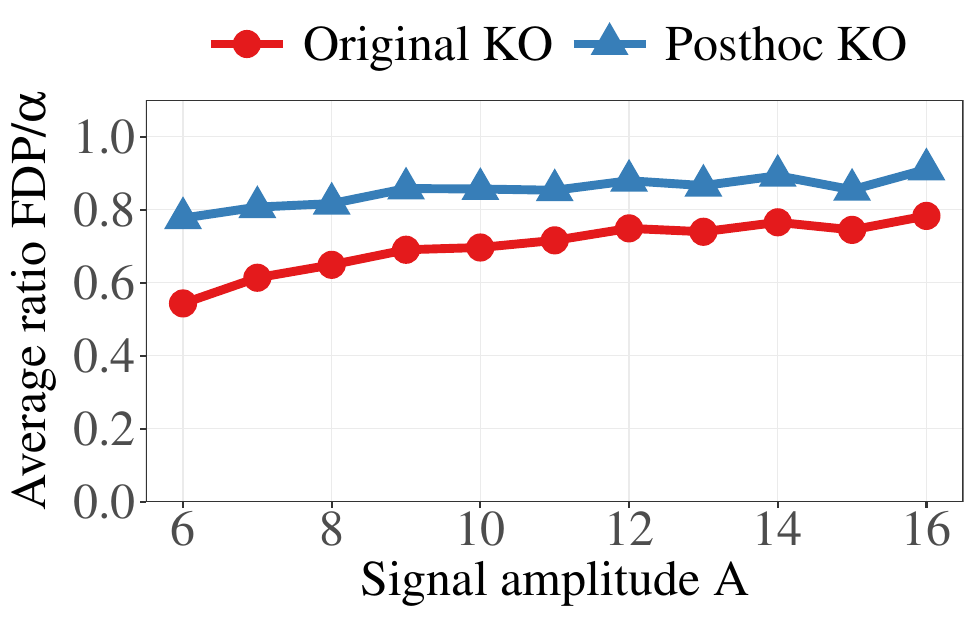}
    \end{minipage}%
}
\caption{Performance comparison between the proposed derandomized post-hoc knockoff procedure and the original derandomized knockoffs method of \citet{candes2018panning} across various signal amplitude values $A$. The number of relevant variables is fixed at $p_{\text{relevant}} = 6$.  
The first row reports the realized power, the second row reports the nominal level $\alpha$ (for the original method this is the user specified level $\alpha^{\kn}$, whereas for the post-hoc knockoff procedure it is the data-dependent level $\tilde{\alpha}^{\ph}$ returned by Algorithm~\ref{alg:posthoc-alpha}) together with the estimated average FDP, and the last row reports the average ratio $\textnormal{FDP/}{\alpha}$. All values are averaged over 2000 runs.}
\label{fig:posthoc_vs_traditional_amplitude}
\vspace{-0.5cm}
\end{figure}

\textbf{Key takeaways from the comparison with the original knockoffs:} Our post‑hoc derandomized knockoff procedure lead to power gains in weak‑signal regimes by increasing the significance level in cases where the original method does not make any discoveries. In addition, in scenarios where both methods attain similar power, the post‑hoc approach can still provide tighter error control, reflected in lower $\alpha$ values. It should be noted that our improvement is a \enquote{free-lunch}: It only reports a larger significance level if the original method did not produce any rejections.

\subsubsection{Comparing against calibrated knockoffs  \texorpdfstring{\citep{luo2025improving}}{Luo et al. (2025)}}
\label{sec:sims_calibrated}
In this section we illustrate the performance of our proposed \textbf{post-hoc knockoff} (presented in Algorithm~\ref{alg:posthoc-alpha}) and compare to the \textbf{calibration knockoffs} method of \citet{luo2025improving}. Calibration knockoffs were proposed to address a limitation of the original knockoff filter: because the selection threshold is data-dependent, it can become overly conservative, leading to an empty or very small rejection set (the so-called \emph{thresholding phenomenon}) even when the signal is not negligible. The calibration knockoffs approach mitigates this issue by calibrating the selection threshold using additional randomness, aiming to stabilize the cutoff and increase the probability of making discoveries while maintaining the desired Type~I error control. The simulation setup is the low-dimensional regime presented in  Section~\ref{seq:sim_generating_data}. Through the simulation we set the $\alpha^{\kn}$ level to $0.20$ FDR, i.e., this is the nominal level of the calibration knockoffs, and the initial significance level of our post-hoc knockoffs.

Figure~\ref{fig:posthoc_vs_calibration_models} reports results under the Gaussian linear and logistic regression models with $p=50$, for two sparsity levels $p_{\text{relevant}}\in\{10,15\}$. Across all configurations, the two procedures exhibit very similar power curves as the signal amplitude increases. For weak signals (Gaussian: amplitudes $\le 3$; Logistic: amplitudes $\le 7$) the posthoc method leads to higher $\alpha$ levels of FDR, while for moderate to large signals (Gaussian: amplitudes $>3$; Logistic: amplitudes $> 7$), the post-hoc approach typically yields more stringent error control, leading to smaller $\alpha$ levels. We emphasize that the estimated average FDP is very similar for the two methods and remains below the nominal level $0.20$ throughout. Furthermore, both methods ensure that the average ratio $\mathrm{FDP}/\alpha$ remains below~1.
\begin{figure}[htb]
\centering
\subfloat[\textbf{Gaussian}\label{fig:gauss_calib}]{%
    \begin{minipage}[t]{0.49\textwidth}
        \centering
        \includegraphics[width=\textwidth]{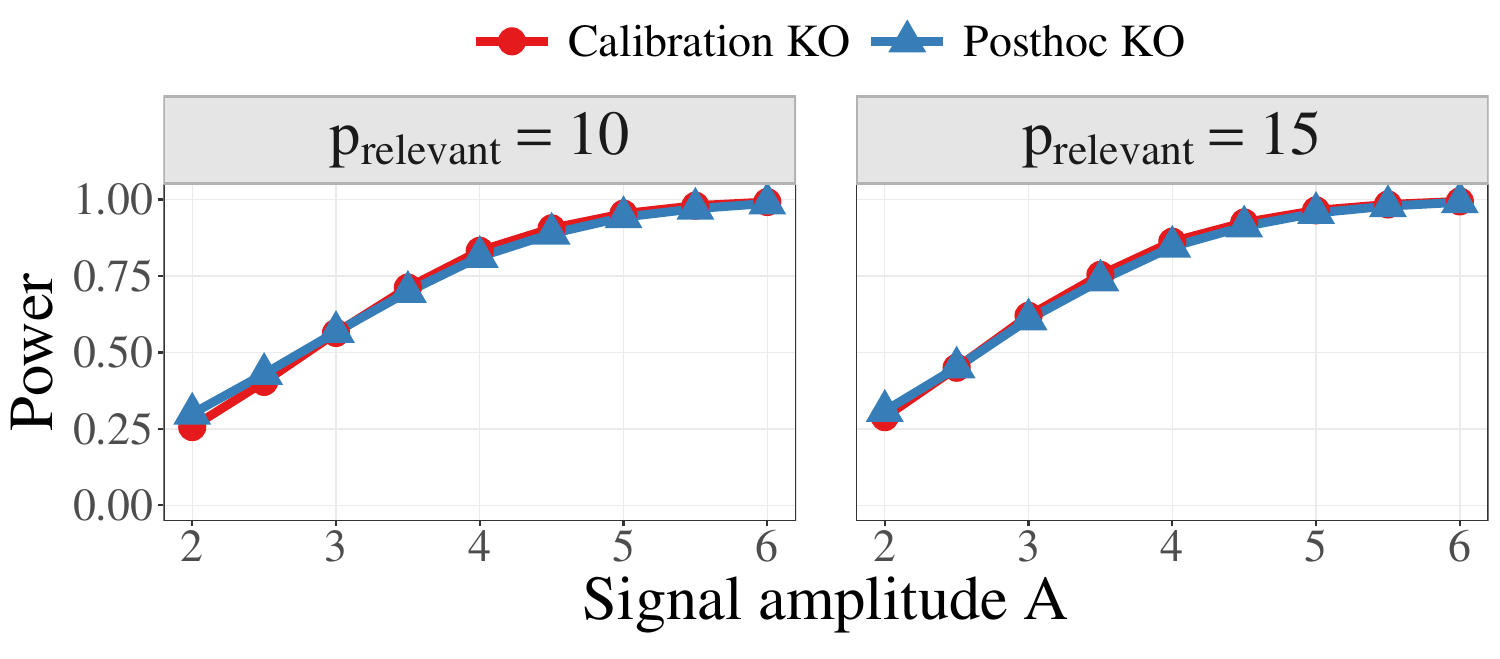}\\
        \vspace{1mm}
        \includegraphics[width=\textwidth]{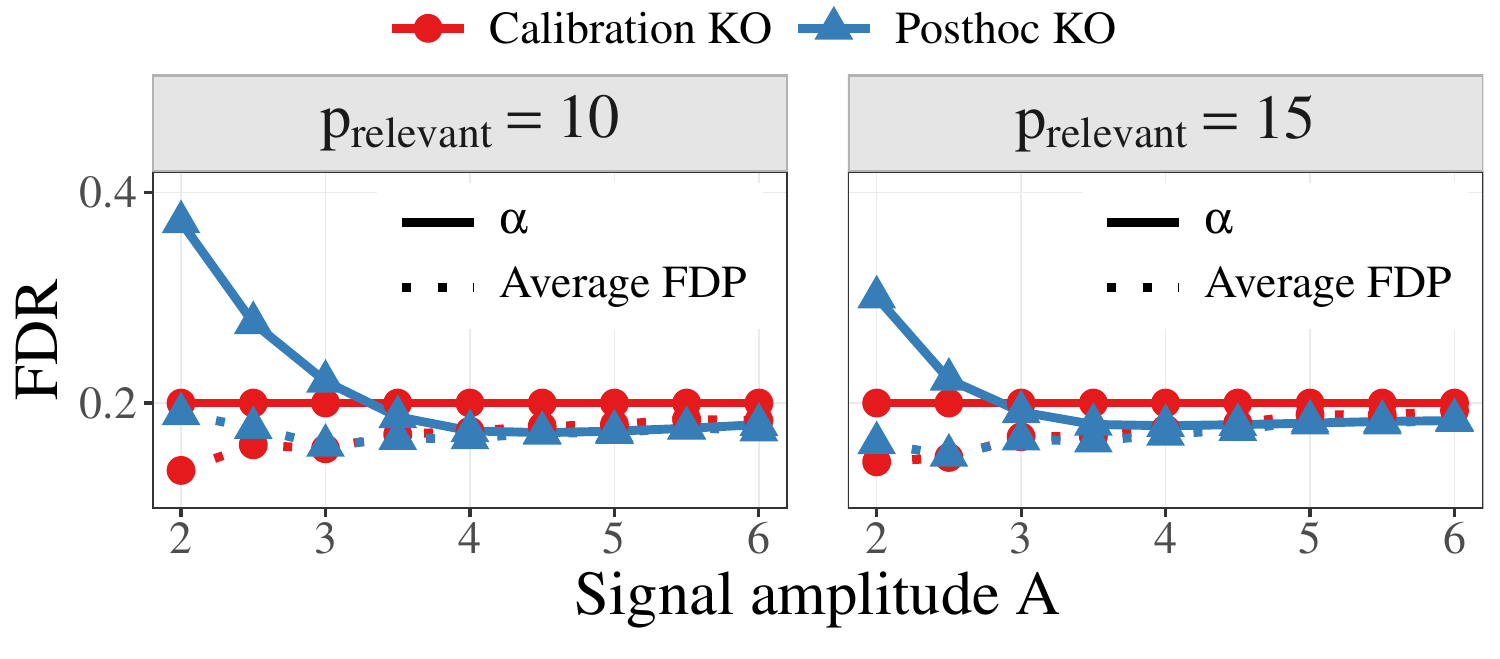}\\
        \vspace{1mm}
        \includegraphics[width=\textwidth]{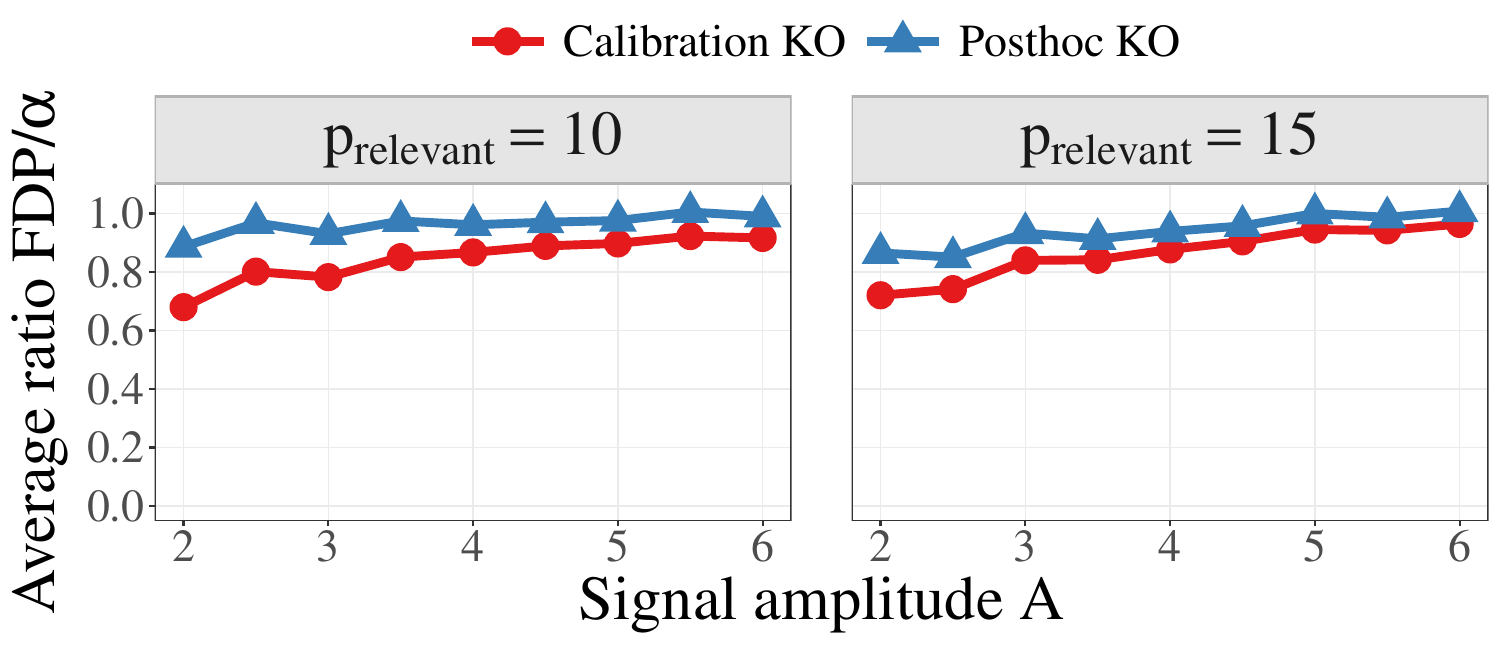}
    \end{minipage}%
}
\ \vrule \ 
\subfloat[\textbf{Logistic}\label{fig:logistic_calib}]{%
    \begin{minipage}[t]{0.49\textwidth}
        \centering
        \includegraphics[width=\textwidth]{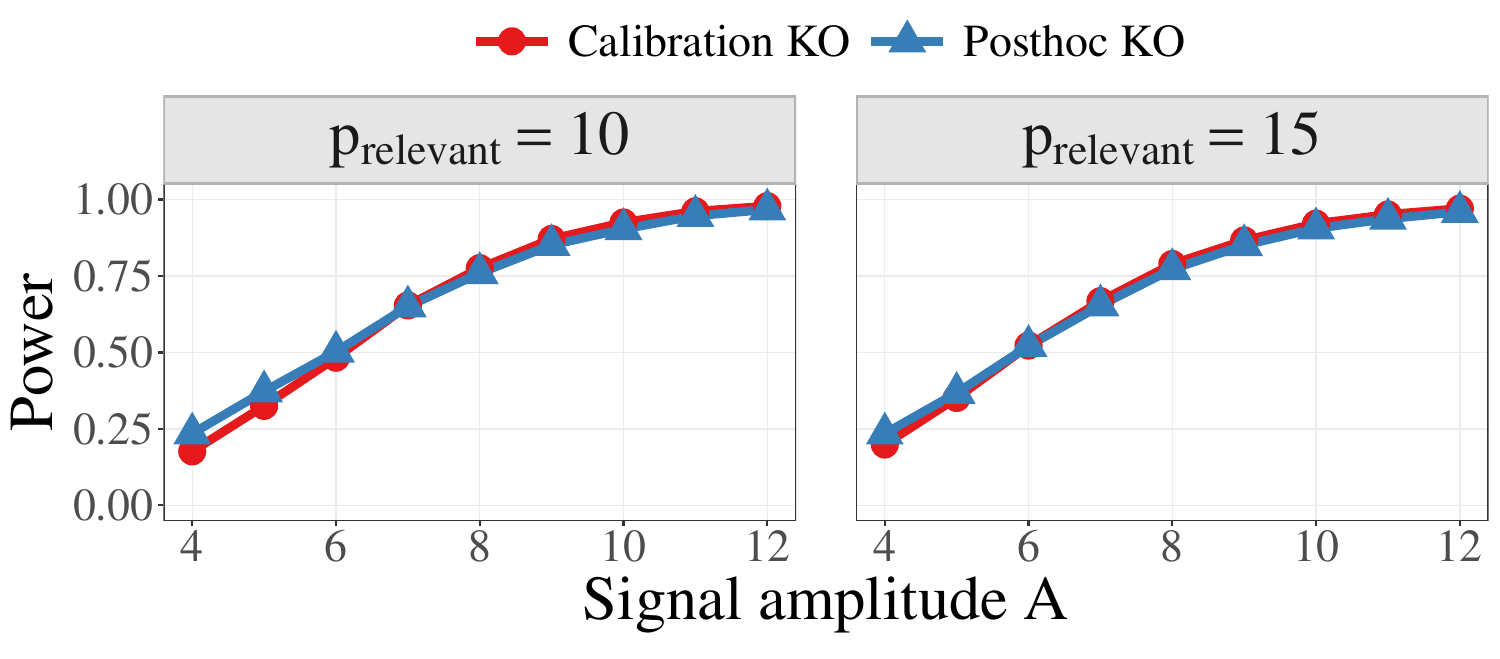}\\
        \vspace{1mm}
        \includegraphics[width=\textwidth]{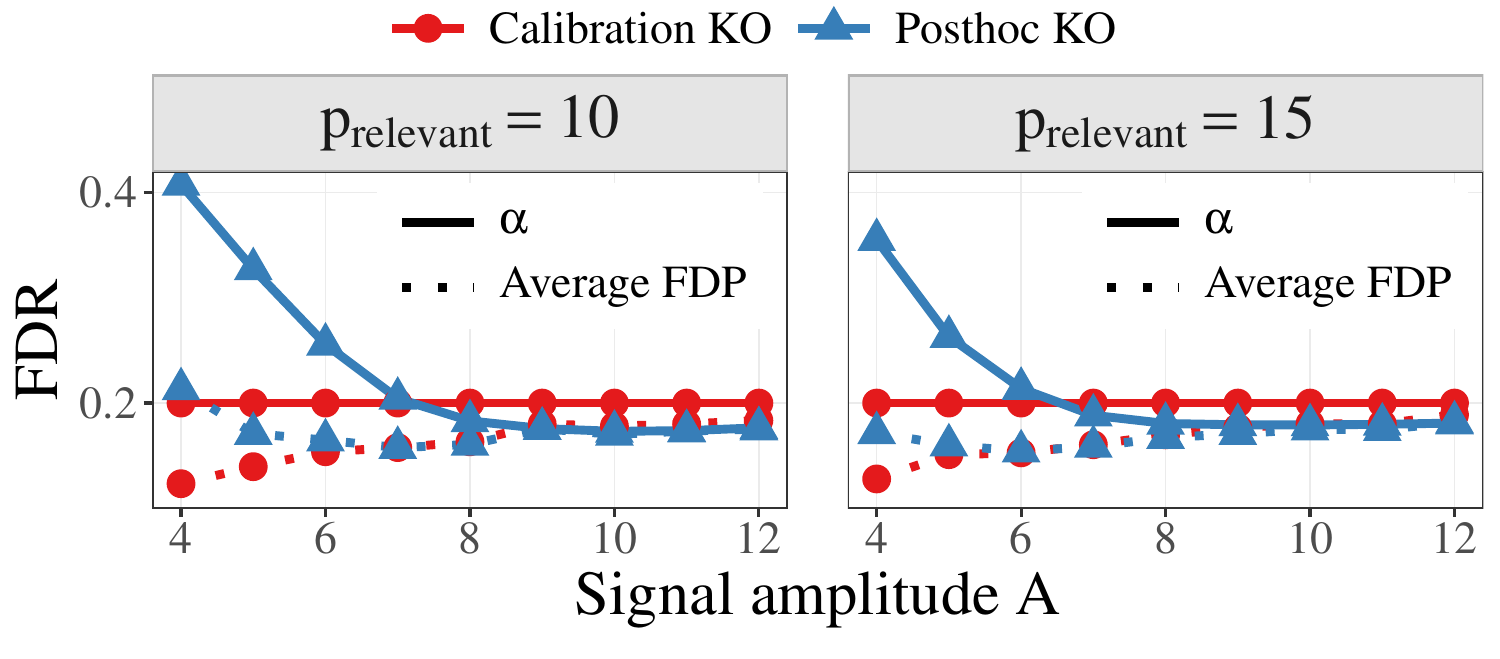}\\
        \vspace{1mm}
        \includegraphics[width=\textwidth]{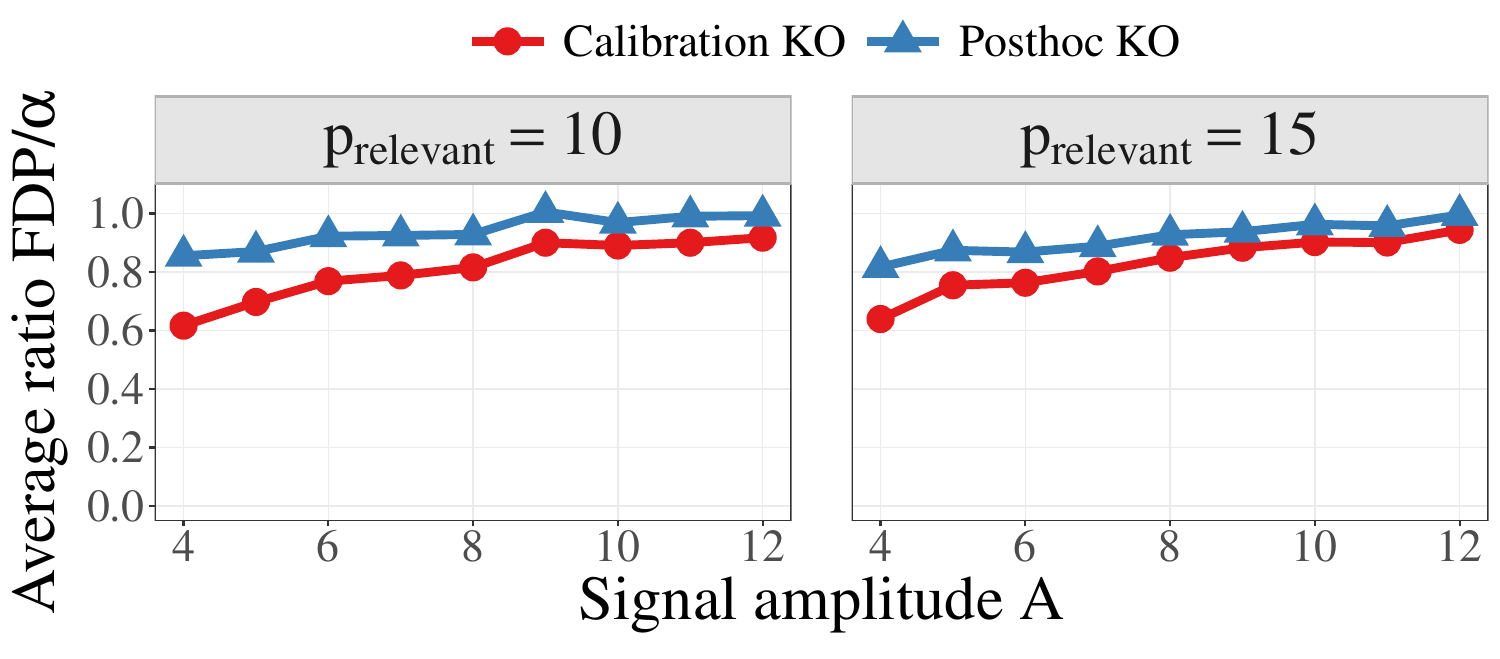}
    \end{minipage}%
}
\caption{Performance comparison between the proposed post-hoc knockoff procedure and the calibration knockoffs method of \citet{luo2025improving} under Gaussian and Logistic models with $p=50$. Results are reported across signal amplitudes for varying sparsity levels, given by the number of relevant variables $p_{\text{relevant}} = 10$ and $15$.
The first row reports the realized power, the second row reports the nominal level $\alpha$ (for the calibration method this is the user specified level $\alpha^{\kn}$, whereas for the post-hoc knockoff procedure it is the data-dependent level $\tilde{\alpha}^{\ph}$ returned by Algorithm~\ref{alg:posthoc-alpha}) together with the estimated average FDP, and the last row reports the average ratio $\textnormal{FDP/}{\alpha}$. All values are averaged over 2000 runs.}
\label{fig:posthoc_vs_calibration_models}
\end{figure}

\textbf{Key takeaways from the comparison with the calibrated knockoffs:}
Across all scenarios, post‑hoc knockoffs achieve power nearly identical to that of calibrated knockoffs. With regard to the Type I error rate $\alpha,$ our post-hoc method leads to higher level for weak signals, but it also leads to stricter error control for moderate to large signals. Apart from this performance, our method is offering several practical and methodological advantages. First, our procedure is straightforward to apply: in contrast to calibrated knockoffs, which implement conditional calibration via per‑variable fallback tests, conditional‑null sampling, and budget accounting, our method requires neither calibration nor auxiliary randomness and does not rely on additional distributional assumptions. Hence, it is compatible with, and can be straightforwardly used with, any knockoff statistics, including nonparametric scores such as those derived from random‑forest importance measures. Second, in contrast to the calibrated knockoff method, the computational overhead of post‑hoc knockoffs is negligible, adding virtually no cost beyond running the original knockoff filter. Most importantly, the method admits a natural derandomization, and in the next two subsections we examine how this derandomized versions compare with other state‑of‑the‑art derandomized procedures.

\subsubsection{Comparing against original derandomized knockoffs for controlling fdr \texorpdfstring{\citep{ren2024derandomised}}{Ren et al. (2024)}}
\label{sec:sims_derandomized_FDR}

In this section, we illustrate the performance of our proposed \textbf{derandomized post-hoc knockoff} procedure for controlling the \textbf{FDR}  (presented in Algorithm~\ref{alg:posthoc-alpha_derand}) and compare it with the \textbf{original derandomized knockoff} procedure from \citet{ren2024derandomised}. The simulation setup is the high-dimensional regime presented in  Section~\ref{seq:sim_generating_data}.  For the nominal levels, we adopt the configurations from \citet{ren2024derandomised}: we set the FDR level $\alpha = 0.20$, which corresponds to the $\alpha_{\ebh}$ level of the derandomized knockoffs and serves as the initial significance level in our post-hoc derandomized knockoffs procedure (Algorithm \ref{alg:posthoc-alpha_derand}), while in both methods we set  $\alpha^{\kn} = \frac{\alpha_{\ebh}}{2},$ a configuration suggested by \citet{ren2024derandomised}.

Figure~\ref{fig:posthoc_vs_derandomized_FDR}(a) presents the results under the Gaussian linear model, where the dimensionality is $p=800$. 
When $p_{\text{relevant}} = 80$, which is the setting considered in \citet{ren2024derandomised}, both methods exhibit comparable performance overall. However, our post-hoc procedure yields a modest power gain at lower signal amplitudes by increasing the nominal level in scenarios where the original derandomized method does not make any rejections. This becomes more pronounced in sparser settings, where $p_{\text{relevant}} = 20$. Notably, in moderately sparse scenarios, our method can also achieve a {lower} average Type~I error rate while maintaining higher power. For example, when $p_{\text{relevant}} = 20$ and the signal amplitude exceeds $5$, our procedure attains similar power with improved Type~I error control. The results are similar for the logistic model provided in Figure~\ref{fig:posthoc_vs_derandomized_FDR}(b).

\begin{figure}[htb]
\vspace{-1cm}
\centering
\subfloat[\textbf{Gaussian}\label{fig:gauss_derand}]{%
    \begin{minipage}[t]{0.49\textwidth}
        \centering
        \includegraphics[width=\textwidth]{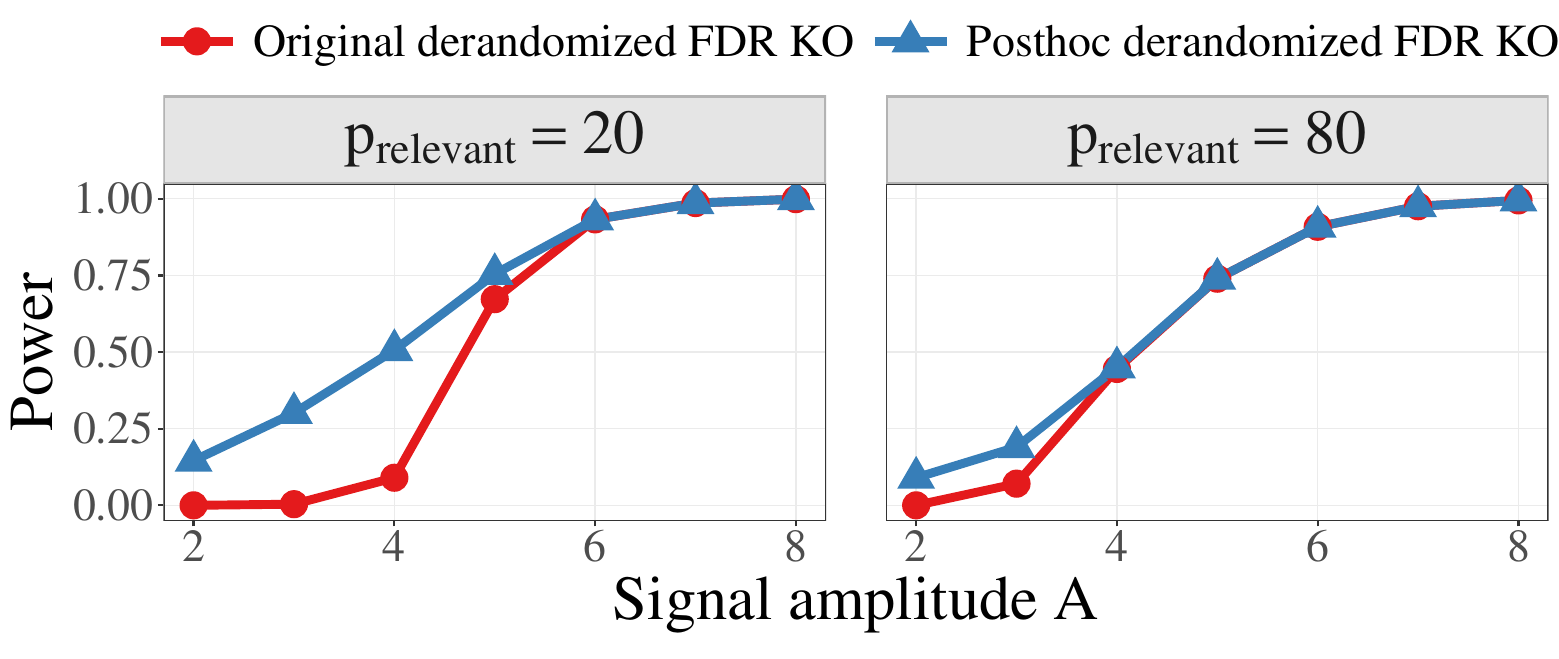}\\
        \vspace{1mm}
        \includegraphics[width=\textwidth]{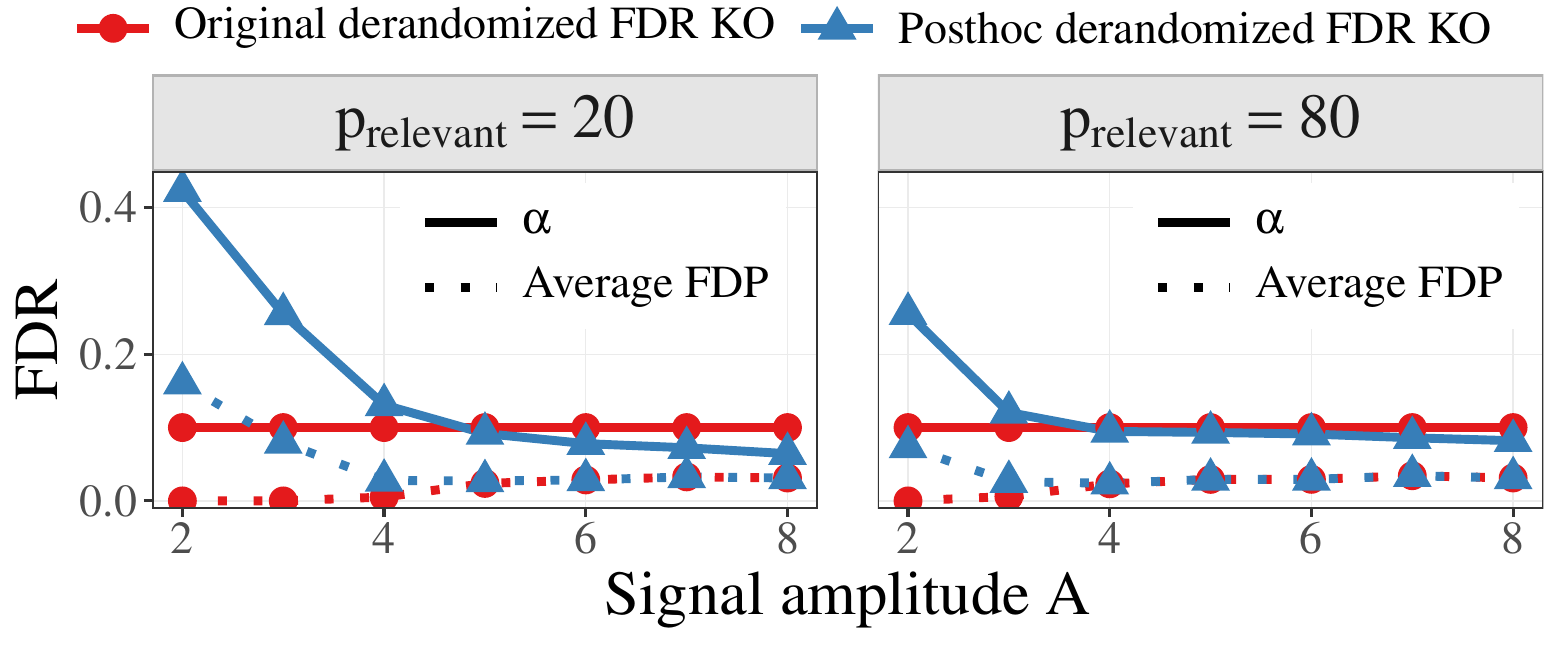}\\
        \vspace{1mm}
        \includegraphics[width=\textwidth]{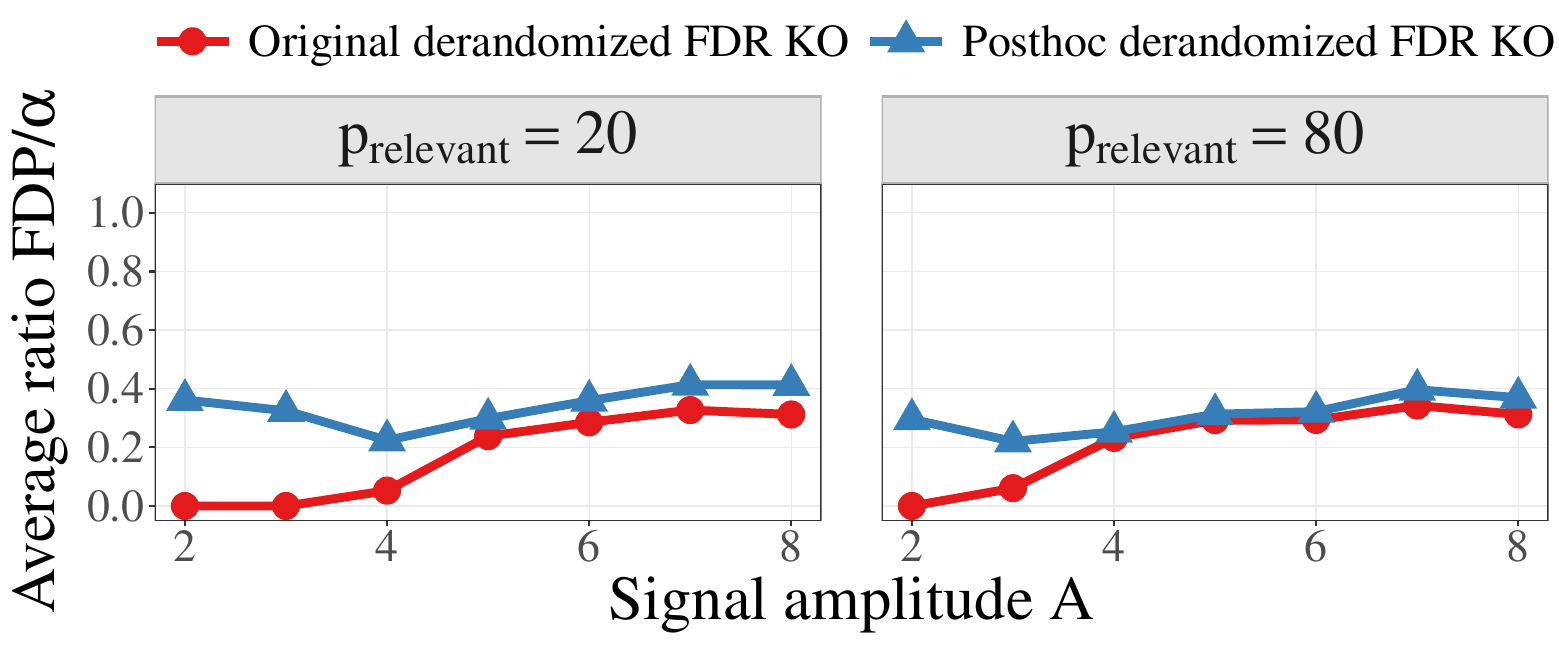}
    \end{minipage}%
}
\ \vrule \ 
\subfloat[\textbf{Logistic}\label{fig:logistic_derand}]{%
    \begin{minipage}[t]{0.49\textwidth}
        \centering
        \includegraphics[width=\textwidth]{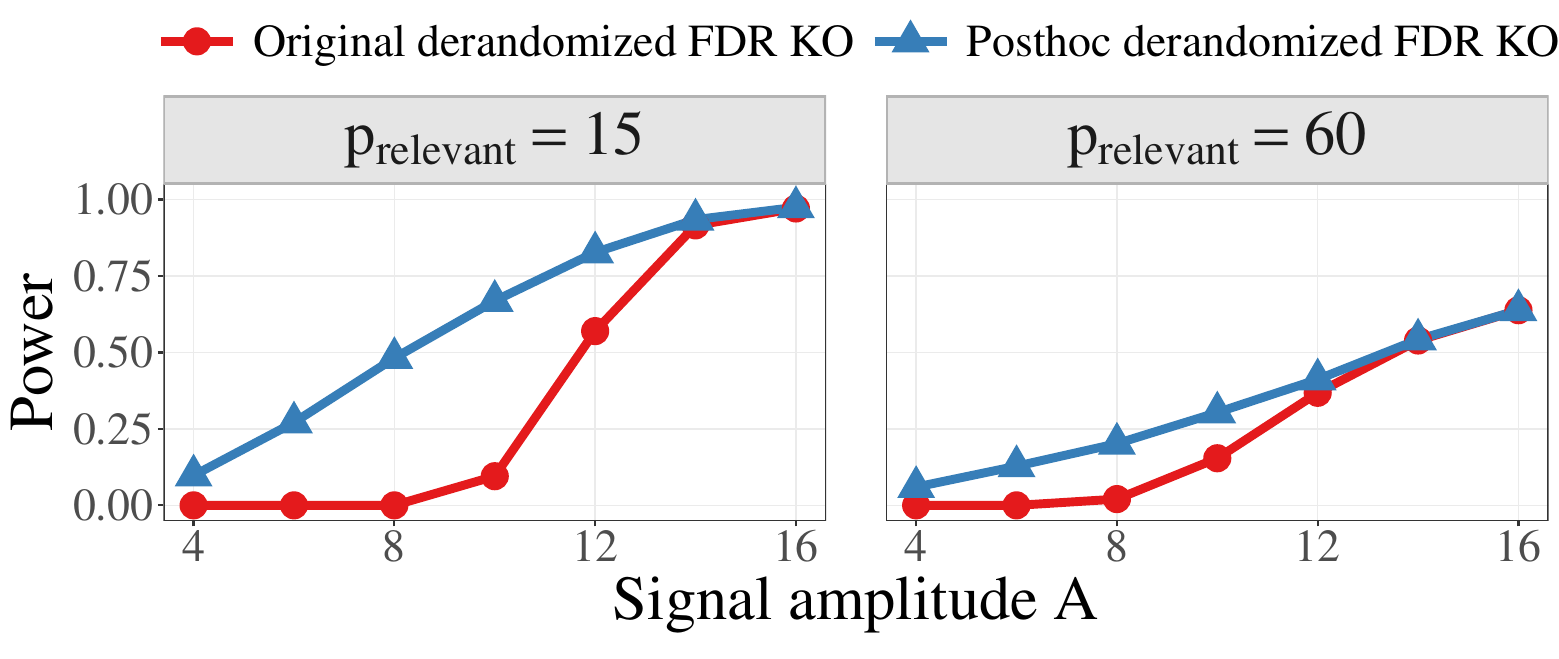}\\
        \vspace{1mm}
        \includegraphics[width=\textwidth]{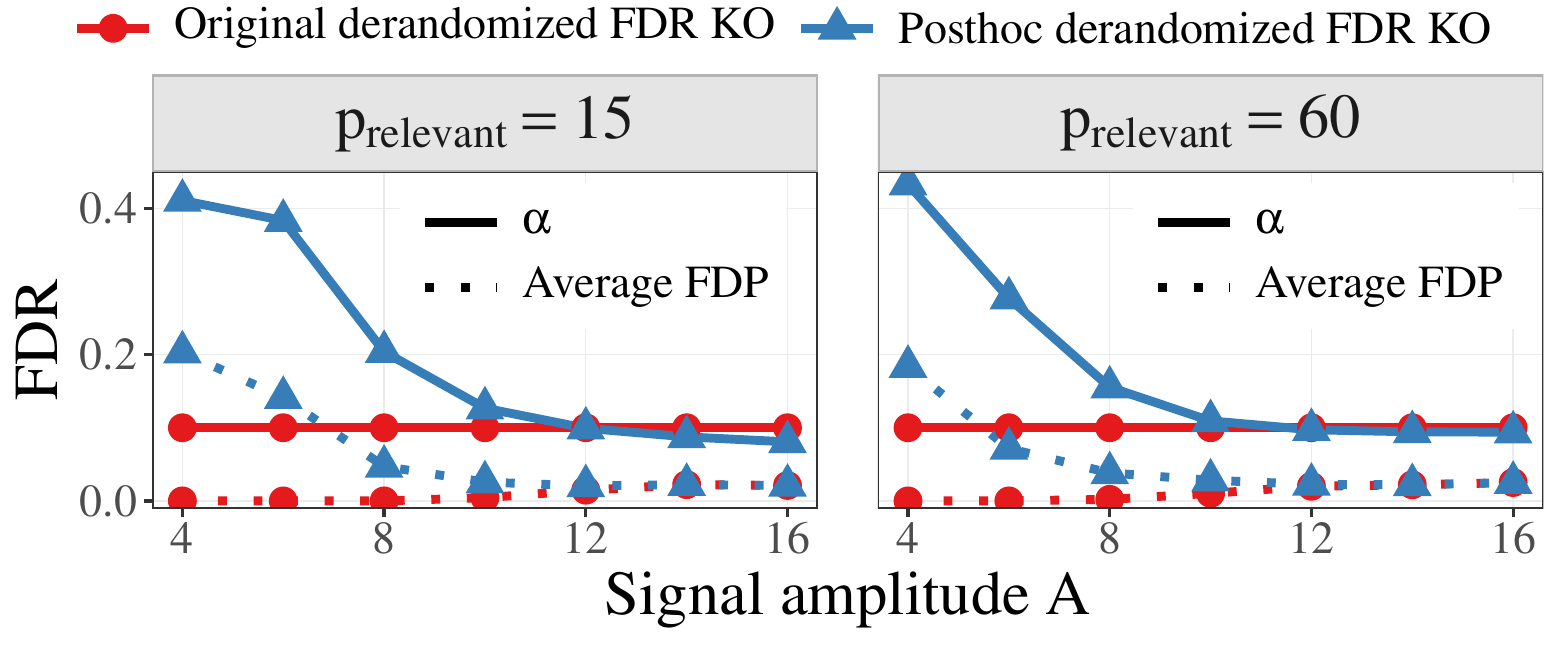}\\
        \vspace{1mm}
        \includegraphics[width=\textwidth]{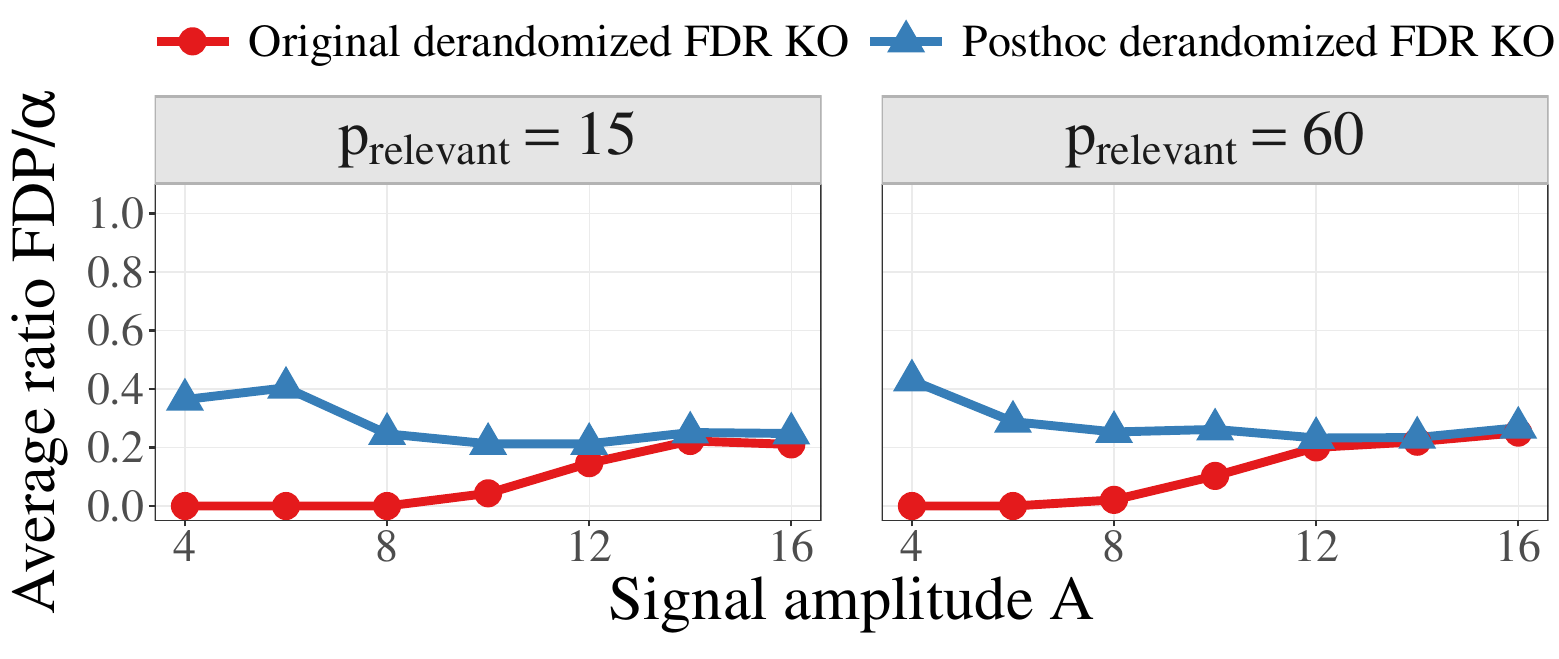}
    \end{minipage}%
}
\caption{Performance comparison between the proposed derandomized post-hoc knockoff procedure for controlling FDR against the original derandomized knockoffs method by \citet{ren2024derandomised} under Gaussian ($p=800$) and Logistic models ($p=600$). Results are reported across signal amplitudes for varying sparsity levels, given by the number of relevant variables $p_{\text{relevant}}$, including the original setting: $p_{\text{relevant}}=80$ for the gaussian, and $p_{\text{relevant}}=60$ for the logistic.
The first row reports the realized power, the second row reports the nominal level $\alpha$ (for the calibration method this is the user specified level $\alpha^{\ebh}$, whereas for the post-hoc knockoff procedure it is the data-dependent level $\tilde{\alpha}^{\dph}$ returned by Algorithm~\ref{alg:posthoc-alpha_derand}) together with the estimated average FDP, and the last row reports the average ratio $\textnormal{FDP/}{\alpha}$.  All values are averaged over 200 runs.}
\label{fig:posthoc_vs_derandomized_FDR}
\vspace{-1cm}
\end{figure}


\textbf{Key takeaways from the comparison with the derandomized knockoffs for controlling FDR:} Our post‑hoc derandomized knockoff procedure yields meaningful power improvements in weak‑signal settings by increasing the nominal level in scenarios where the original method does not make any discoveries. Furthermore, in scenarios where both methods achieve comparable power, the post‑hoc approach can yield tighter error control. As described before, our improvement can be seen as a \enquote{free-lunch}.

\subsubsection{Comparing against original derandomized knockoffs for controlling pfer \texorpdfstring{\citep{ren2023derandomizing}}{Ren et al. (2023)}}
\label{sec:sims_derandomized_PFER}
In this section, we illustrate the performance of our proposed \textbf{derandomized post-hoc knockoff} procedure for controlling the \textbf{PFER} (presented in Algorithm \ref{alg:posthoc-PFER}) and compare it with the \textbf{original derandomized knockoffs} procedure of \citet{ren2023derandomizing}. The simulation setup is the low-dimensional regime presented in  Section~\ref{seq:sim_generating_data}.  For the nominal levels, we set the expected number of false discoveries $\nu = 1$, and, following the configurations from \citet{ren2023derandomizing}, we set $\eta = 0.50$. In the original derandomized PFER knockoffs procedure, the parameter $\eta\in(0,1]$ is a prespecified threshold that determines how consistently a variable must be selected across the $k$ knockoff runs to be included in the final rejection set $R_\eta^{\RWC}$ in \eqref{eq:rej_PFER}; thus, smaller $\eta$ makes rejections easier and larger $\eta$ makes them more stringent. \citet{ren2023derandomizing} showed that $\RWC$-knockoff controls the PFER at level $\nu/\eta$, i.e., in our setting this is equal to $1/0.50 = 2$. The parameter $\eta$ is not treated as a fixed design parameter in our post-hoc derandomized knockoffs framework: instead, it is adjusted in a data-dependent way, as permitted by the guarantee in Theorem~\ref{theo:PFER}.

Figure~\ref{fig:posthoc_vs_derandomized_PFER} reports results under the Gaussian linear and logistic regression models with $p=50$, for two sparsity levels $p_{\text{relevant}}\in\{5,10\}$. Across all configurations, the two procedures exhibit very similar power curves as the signal amplitude increases.
\begin{figure}[htb]
\centering
\subfloat[\textbf{Gaussian}\label{fig:gauss_pfer}]{%
    \begin{minipage}[t]{0.49\textwidth}
        \centering
        \includegraphics[width=\textwidth]{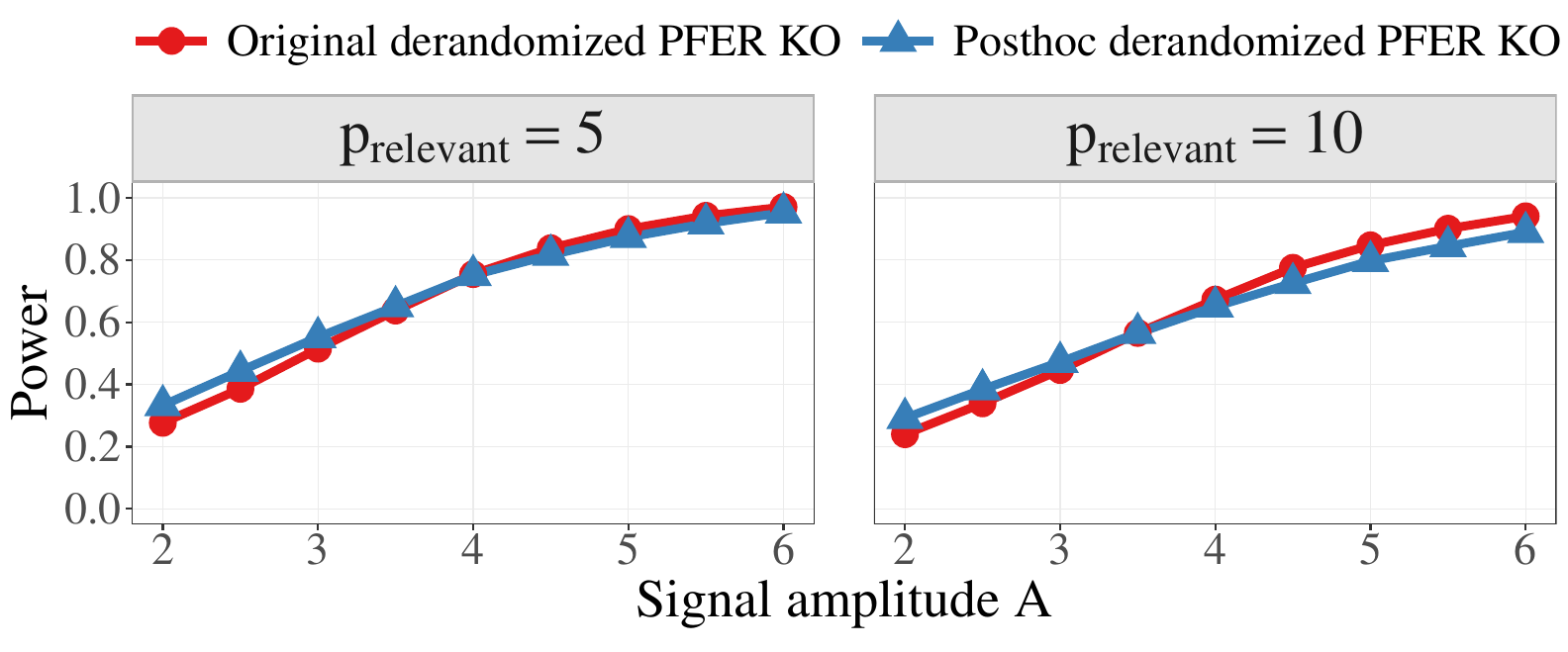}\\
        \vspace{1mm}
        \includegraphics[width=\textwidth]{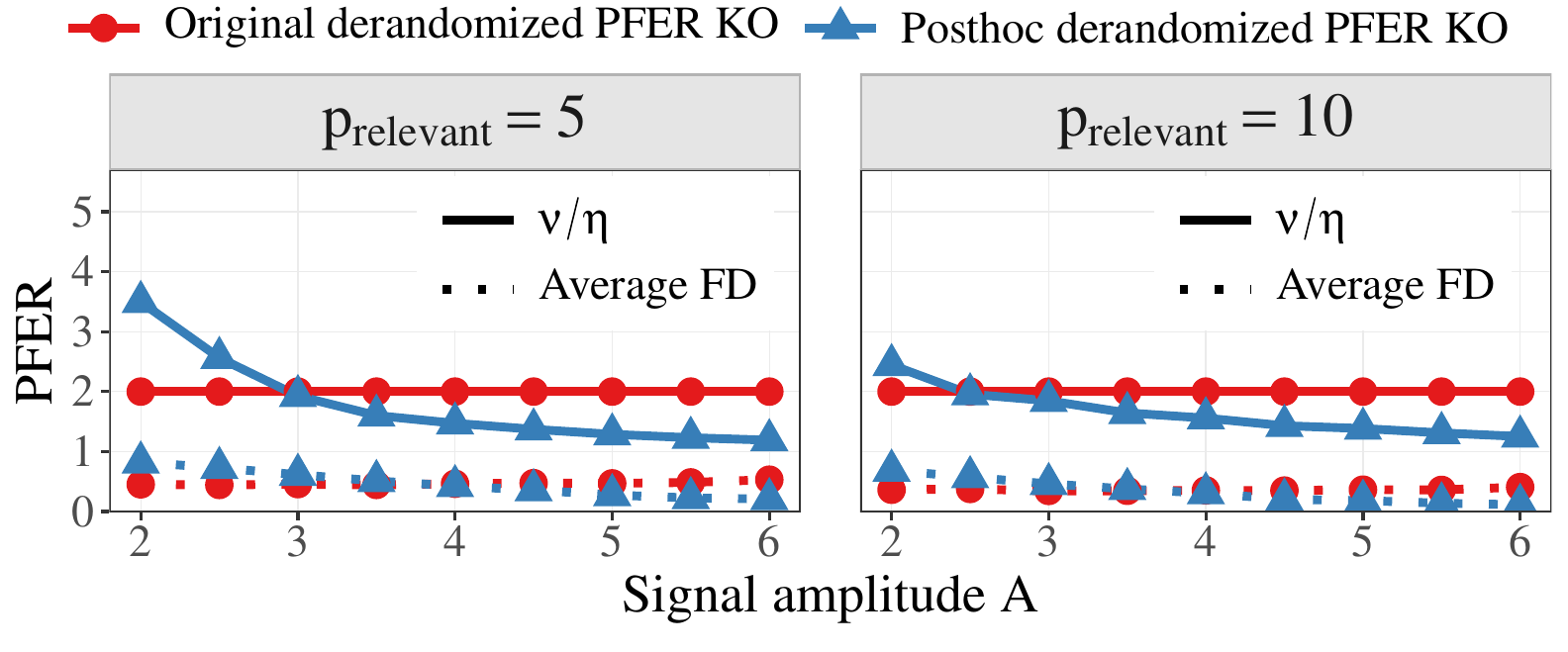}\\
        \vspace{1mm}
        \includegraphics[width=\textwidth]{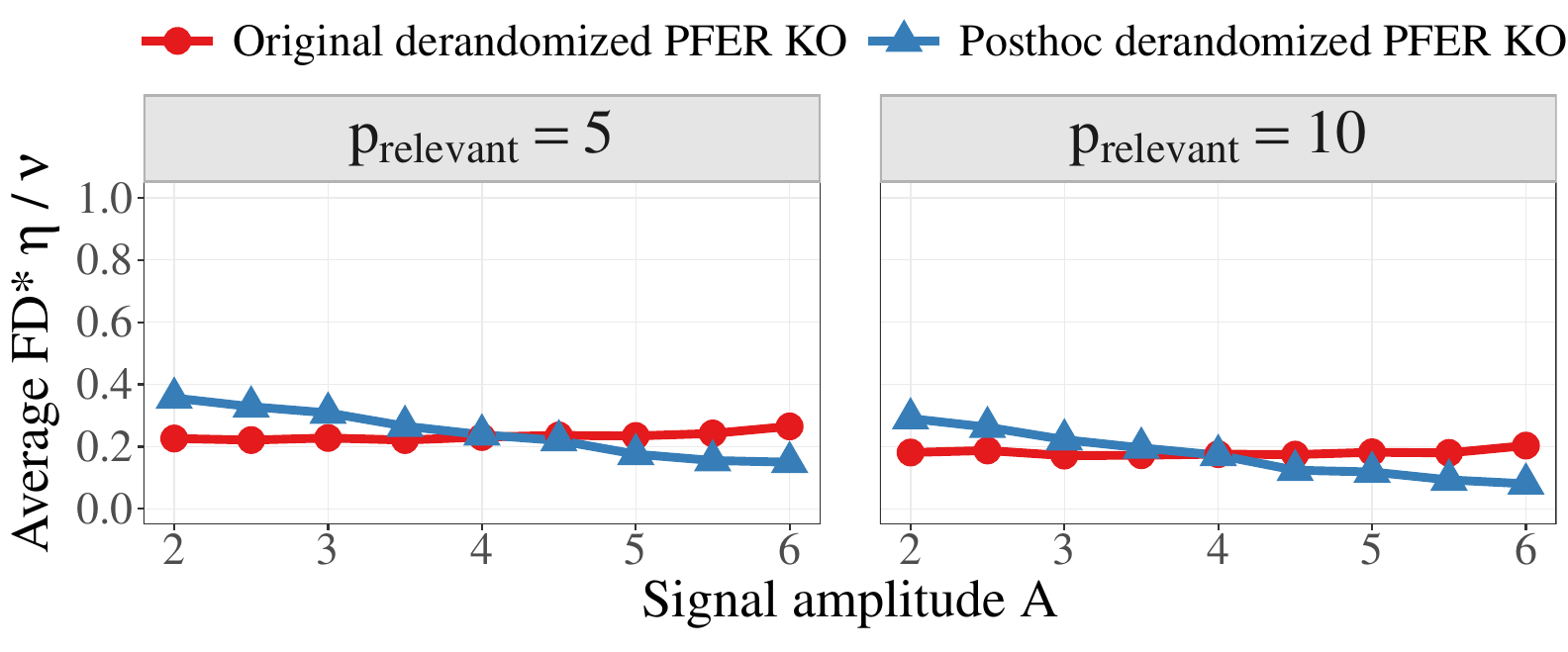}
    \end{minipage}%
}
\ \vrule \ 
\subfloat[\textbf{Logistic}\label{fig:logistic_pfer}]{%
    \begin{minipage}[t]{0.49\textwidth}
        \centering
        \includegraphics[width=\textwidth]{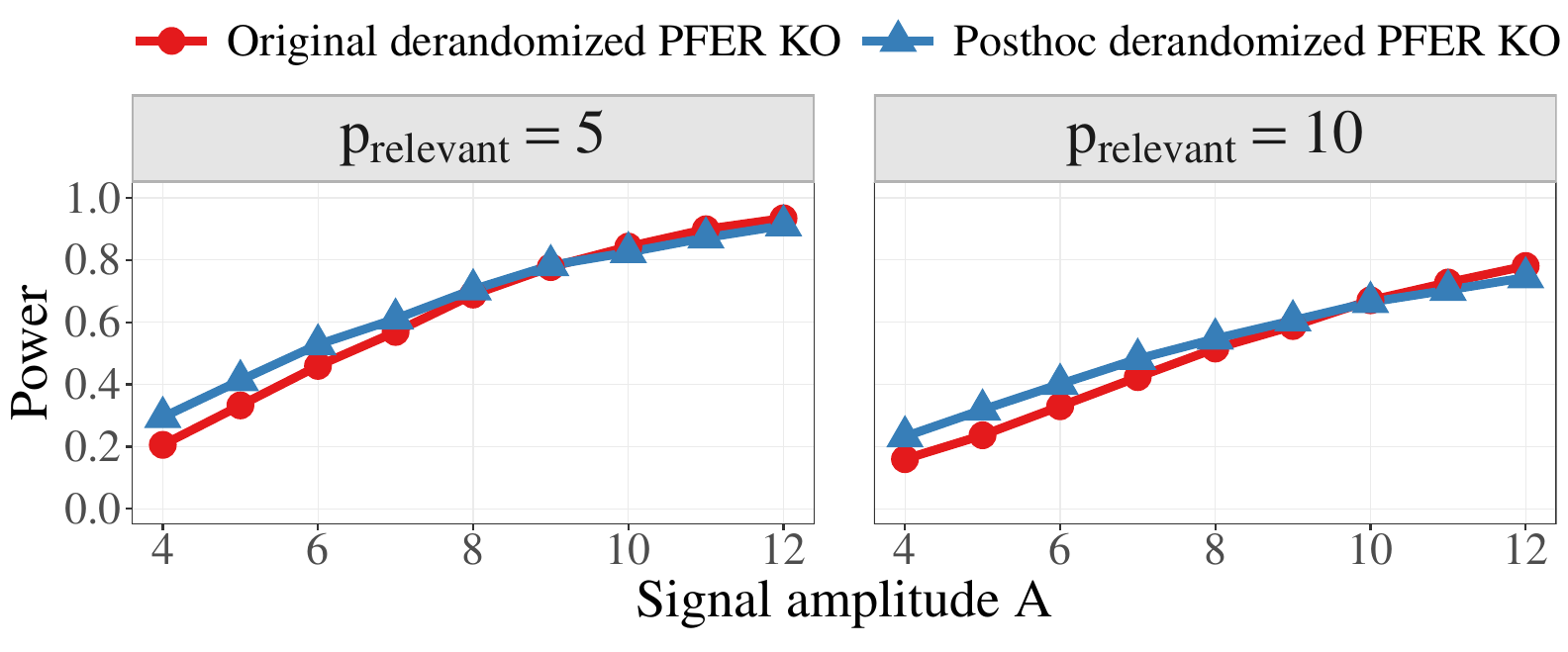}\\
        \vspace{1mm}
        \includegraphics[width=\textwidth]{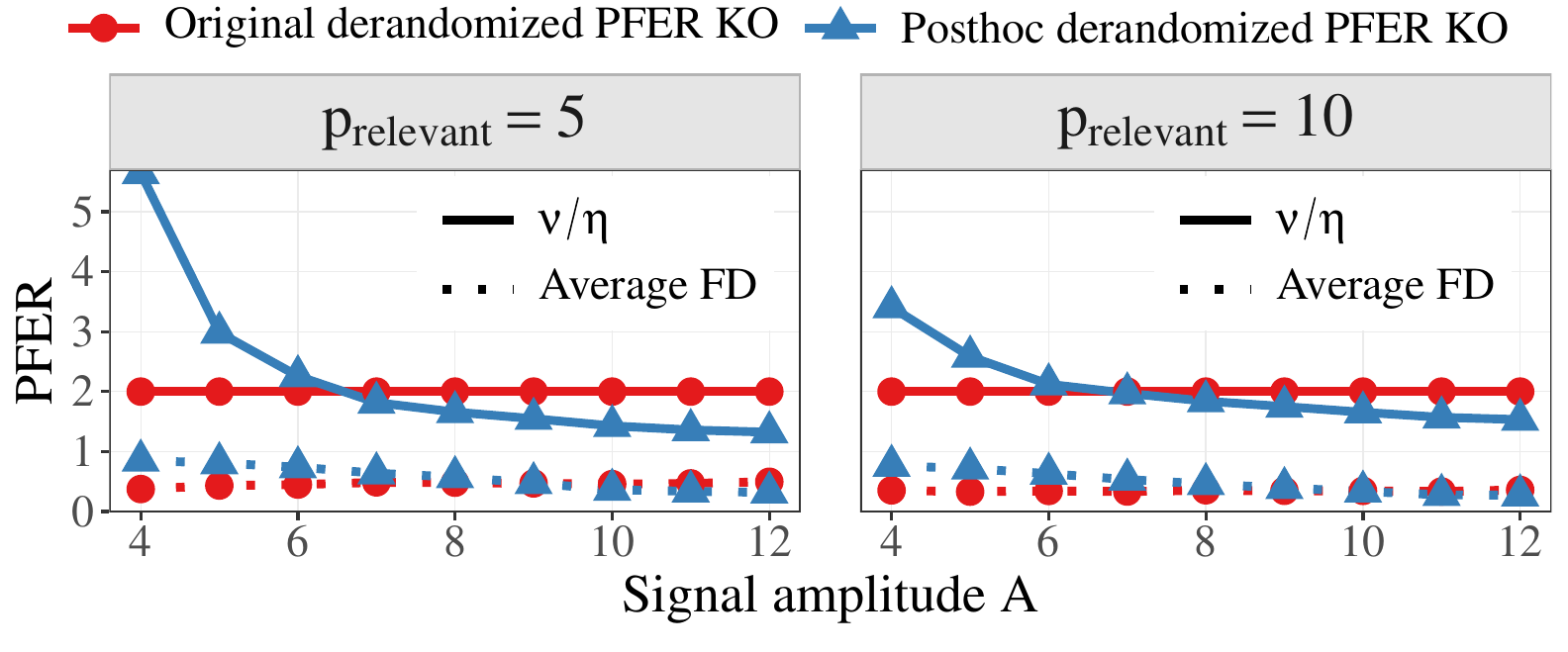}\\
        \vspace{1mm}
        \includegraphics[width=\textwidth]{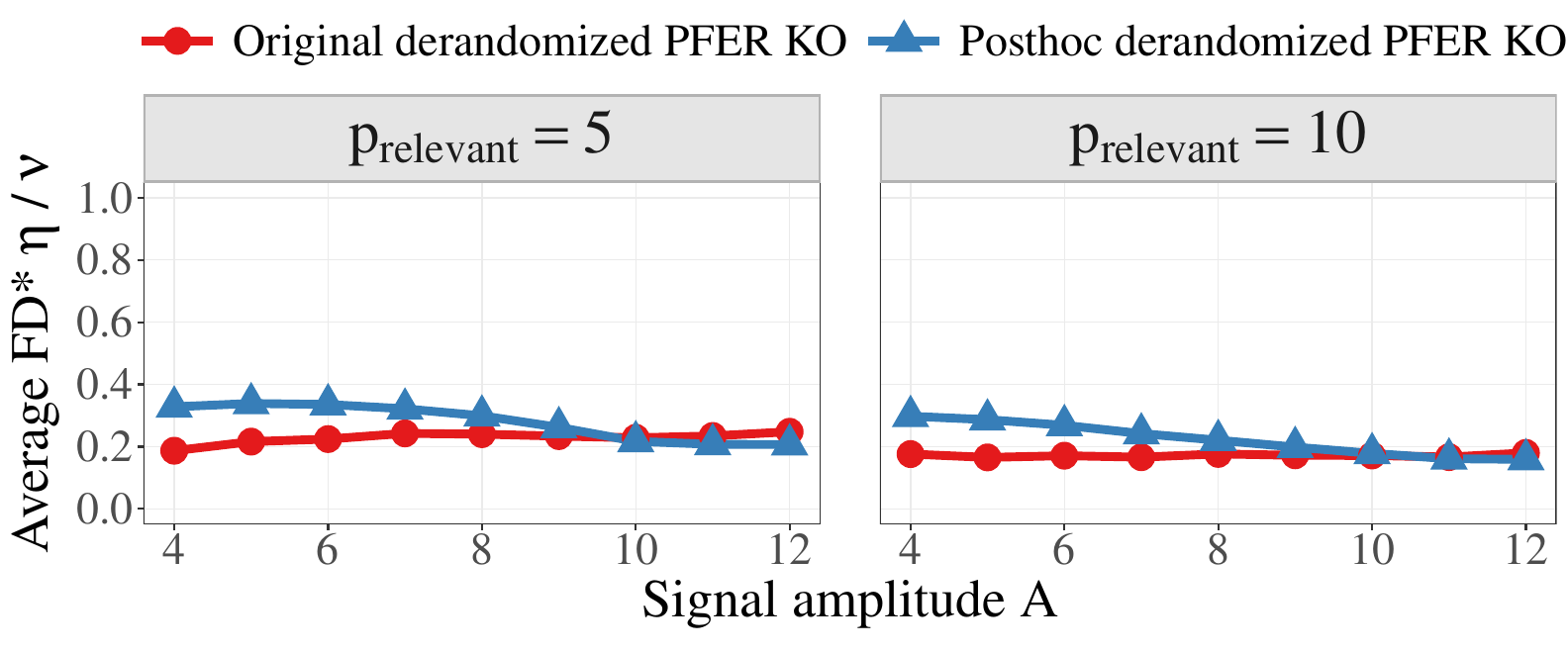}
    \end{minipage}%
}
\caption{Performance comparison between the proposed derandomized post-hoc knockoff procedure for controlling PFER against the original derandomized knockoffs method by \citet{ren2023derandomizing} under Gaussian and Logistic models with $p=50$. Results are reported across signal amplitudes for varying sparsity levels, given by the number of relevant variables $p_{\text{relevant}} = 5$ and $10.$
The first row reports the estimated power, the second row reports the data-dependent level $\alpha$ together with the estimated FD, and the last row reports the average ratio $\textnormal{FD}*\eta/\nu$. All values are averaged over 2000 runs.
}
\label{fig:posthoc_vs_derandomized_PFER}
\end{figure}
Interestingly, for moderate to large signals (Gaussian: amplitudes $>3$; Logistic: amplitudes $> 7$), the post-hoc approach typically yields more stringent error control, leading to smaller estimated values of the expected number of false discoveries (generally below $2$). Finally, across all comparisons the ratio $FD*\eta/\nu$ for both methods remains bellow $1$, confirming that PFER is appropriately controlled.

\textbf{Key takeaways from the comparison with the derandomized knockoffs for controlling PFER:}
The post‑hoc derandomized knockoffs achieve power similar to the original derandomized PFER method, while their adaptive thresholding enables more stringent error control whenever the signal strength is sufficiently high. It should be noted that our post-hoc implementation for PFER control does not always lead to larger power than the original method. However, if the $(R_{\eta}^{\RWC}, \eta)$ combination of the original method seems more desirable, we could also switch to that combination post-hoc.

\section{Applications in analyzing clinical trial data}
\label{sec:clinical_trial} 
In a recent work, \cite{zimmermann2024} illustrate how knockoff-based methods can be applied to clinical trial data to identify prognostic and predictive biomarkers while controlling type-I errors in exploratory variable selection. In this section we will present applications of our post-hoc knockoffs methods in analysing clinical trial data: Subsection \ref{sec:clinical_trial_prognostic} will focus on the identification of prognostic variables using a synthetic dataset designed to mimic a real randomized controlled trial, whereas Subsection \ref{sec:clinical_trial_predictive} will analyse predictive biomarkers using actual clinical trial data from five phase III studies.

\subsection{Derive prognostic variables using post-hoc knockoffs}
\label{sec:clinical_trial_prognostic}
To demonstrate the practical utility of our approach and the advantages of post-hoc knockoffs for analysing clinical trial data, we consider a synthetic dataset designed to mimic a real randomized controlled trial, introduced by \citet{Sun2024}. To generate the data we used \texttt{benchtm} R package \citep{Sun2024}, using Scenario 1 with a continuous outcome and a sample size of $n = 1000$. Each patient is described by $p = 30$ baseline covariates, a dimensionality that reflects realistic clinical settings where tens of demographic, clinical, laboratory, and lifestyle variables are routinely collected and jointly analysed. The outcome for each patient was generated as a function of covariates $\mathbf{X} = (X_1, X_2, \dots, X_{30})$ and the binary treatment assignment $T \in \{0,1\}$ according to
\begin{equation}
f(\mathbf{X}, T) = s \times \big(0.5\,\mathbb{I}(X_1 = \text{Y}) + X_{11}\big) 
+ T\big(\beta_0 + \beta_1 \Phi(20(X_{11} - 0.5))\big), \label{eq:benchtm}
\end{equation}
where $\mathbb{I}(\cdot)$ denotes the indicator function, $s = 2.32$ controls the baseline covariate effect, and $\beta_0 = 0.053$ and $\beta_1 = 0.383$ determine the average treatment effect and the additional effect in the specified subpopulation, respectively. These parameters were tuned to mimic real clinical trial data, and a full description of the data-generating mechanism and simulation design is provided in \citet{Sun2024}. %

We simulate 10 independent RCT datasets following the procedure described above, denoted by R1, ..., R10. To generate the knockoffs we used the sequential knockoff procedure \citep{Kormaksson2021}, implemented in the knockofftools package \citep{zimmermann2024}. Then, we use the Lasso Coefficient Difference (LCD) knockoff statistic, which is computed as the difference between the absolute LASSO coefficient of each original variable and that of its knockoff copy, as the knockoff filter.  Finally, we examine which variables are selected by the two methods (original knockoffs vs post-hoc knockoffs) under a nominal FDR level of 0.20. Figure \ref{fig:heatmaps_original_posthoc} displays the corresponding heatmaps of selected variables. We remind the reader that, according to the data-generating process specified in equation (\ref{eq:benchtm}), the set of truly relevant variables in this scenario is $I_1 = \{T, X_{1}, X_{11}\}.$

\begin{figure}[htb]
  \centering
  \subfloat[Original knockoffs heatmap\label{fig:heatmap_original}]{%
    \includegraphics[width=0.49\textwidth]{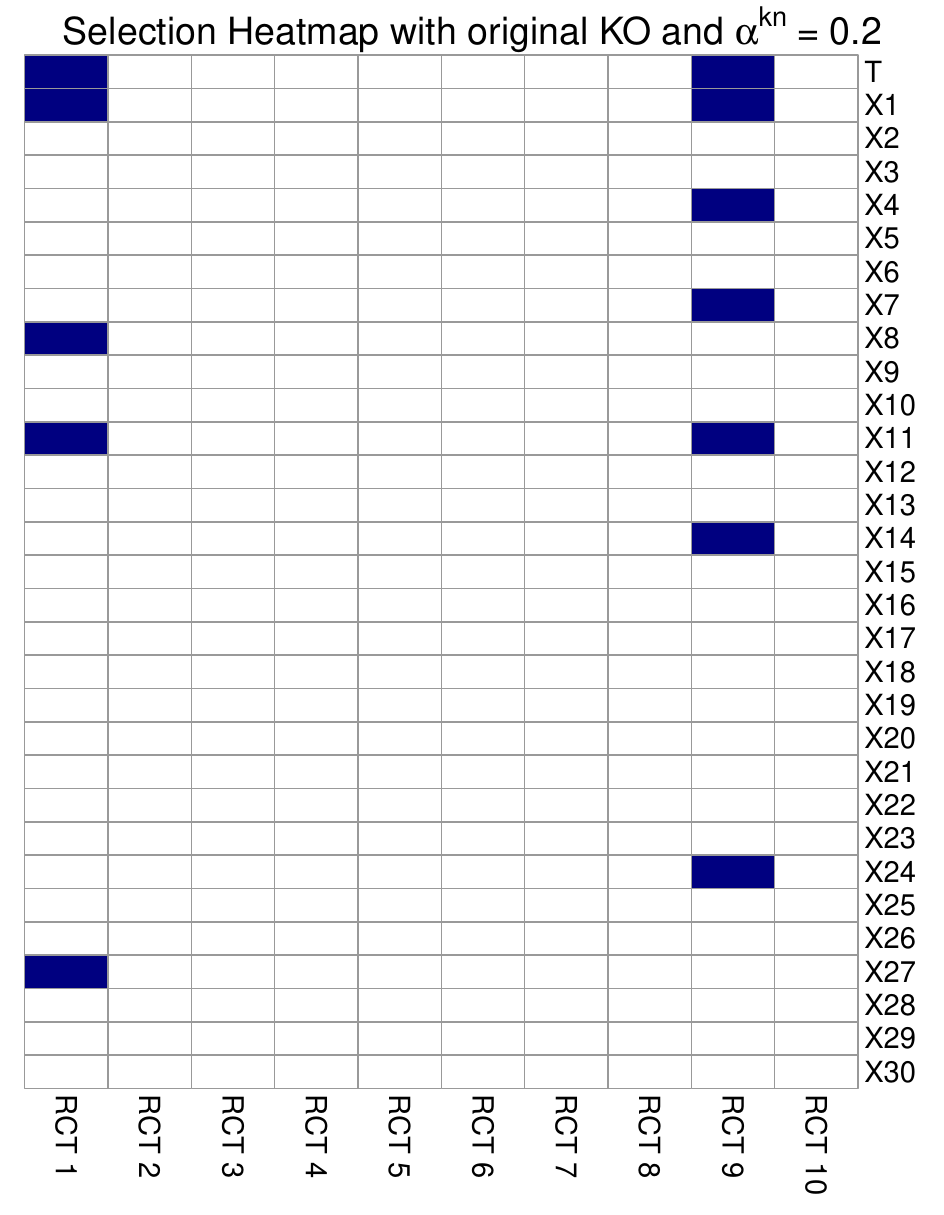}%
  }
  \hfill
  \subfloat[ph-knockoffs heatmap\label{fig:heatmap_ph}]{%
    \includegraphics[width=0.49\textwidth]{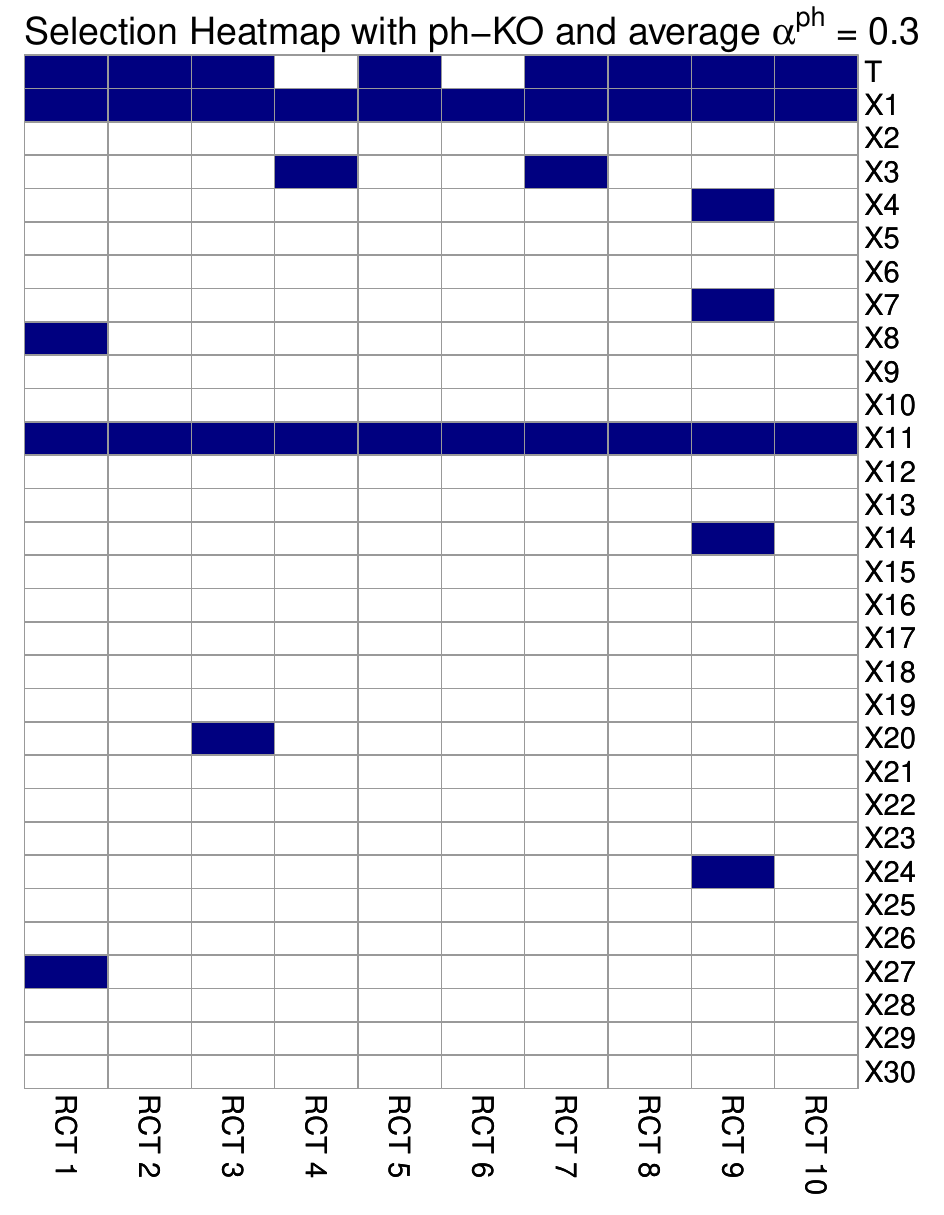}%
  }
  \caption{Variable selection heatmaps across 10 simulated RCT datasets using the original knockoffs method and the post-hoc knockoffs method. Rows correspond to variables and columns to simulation runs. Blue indicates that a variable was selected in a given run, while white indicates that it was not selected. In this setting, the true relevant variables are $X_1$, $X_{11}$ and $T$.}
  \label{fig:heatmaps_original_posthoc}
\end{figure}

As shown in the left panel, the original knockoffs method selects variables only in two runs out of 10, in the RCT 1 and RCT 9. If we use this method in practice, then most of the time the method makes no discoveries, leading to frequent failures in detecting true signals. In contrast, the post‑hoc method (right panel) selects variables $X_{1}$ and $X_{11}$, both true predictors, in all ten runs, and additionally selects the third relevant variable $T$ in 8 out of 10 runs. This indicates that, in practical applications, the method would almost always recover the truly relevant variables. While the average $\alpha^{\ph}$ across the ten runs is 0.3, it is worth emphasizing again that these gains come at no cost, since our method is more informative. In the two runs (RCT 1 and 9) where the original method makes discoveries, our post-hoc procedure identifies the same variable, with nominal levels of $0.20$ and $0.14,$ which is smaller than the nominal level used by the original KO. In the remaining eight runs, where the original KO makes no discoveries, our method provides more informative results: we almost always recover the true signals, and the average nominal $\alpha^{\ph}$ level across these runs is 0.33.
 

\subsection{Derive effect modifiers in real trial}
\label{sec:clinical_trial_predictive}
Identifying effect modifiers, also known as predictive biomarkers, is a central task in the analysis of clinical trial data. Effect modifiers are baseline patient characteristics that influence the magnitude or direction of a treatment effect. Detecting such variables is critical for a range of scientific and clinical activities, including defining subgroups with enhanced treatment benefit, informing treatment-personalization strategies, and generating hypotheses for further biological or mechanistic investigation. In this section study, we focus on the task of identifying effect modifiers from clinical trial data using our knockoffs and controlling FDR. 

We will analyse five Phase~III clinical trials of secukinumab (Cosentyx) for psoriatic arthritis (PsA): FUTURE~1–5 \citep{future1,future2,future3,future4,future5}. The combined dataset contains $n=1937$ patients across the five trials. Covariates with more than 20\% missingness were removed. Remaining variables were imputed using multiple imputation, and completed datasets were aggregated using medians (numeric variables) or modes (categorical variables). This produced a final set of 70 mixed-type baseline variables including demographics, disease history, laboratory biomarkers, quality-of-life assessments, and the proteomic marker BD-2.\citep{Cardner2023} Further details on how we generated the analysis dataset are provided in \citet{sechidis2025drlearner}. The primary endpoint is the binary ACR50 response at week 16 \citep{felson1993american} and treatment heterogeneity is evaluated on the risk-difference scale, $\Pr(Y=1 \mid T=1) - \Pr(Y=1 \mid T=0).$

Building on our previous work by \citet{sechidis2025drlearner}, we first derive pseudo-observations of the individual treatment effect using the doubly robust (DR) learner framework. These pseudo-observations summarize treatment heterogeneity at the patient level in a model-agnostic manner. In that earlier analysis, the authors reported rankings of variable‐importance scores that were obtained using conditional inference trees and permutation importance metrics \citep[Figure 17]{sechidis2025drlearner}. The top ten variables in that ranking were C-reactive protein, age, proteomic marker BD-2, baseline fatigue score, apolipoprotein~A1, alkaline phosphatase, total protein, sex, lymphocytes/leukocytes percentage, and ethnicity. Several of these variables, most notably C-reactive protein, Age, Fatigue Score, and proteomic marker BD-2, correspond to known effect modifiers previously identified in independent analyses \citep{sechidis2021usingknockoffs,bornkamp2023predicting,Cardner2023}.

While such rankings provide valuable descriptive insight, they do not by themselves deliver statistically controlled variable selection. In exploratory subgroup discovery and predictive biomarker research, methods such as knockoffs are essential for producing reproducible findings under explicit FDR control. To generate the knockoffs and compute the corresponding knockoff statistics, we followed the same procedure as in the previous section, using sequential knockoffs and the LCD statistic.

Firstly, we explore the selections obtained from the original derandomized knockoff filter with those from our proposed post-hoc knockoff procedure. For the nominal levels, we adopt the configurations from \citet{ren2024derandomised}: we set the FDR level $\alpha = 0.20$, which corresponds to the $\alpha_{\ebh}$ level of the derandomized knockoffs and serves as the initial significance level in our post-hoc derandomized knockoffs procedure, while in both methods we set  $\alpha^{\kn} = \frac{\alpha_{\ebh}}{2}.$ 
\begin{description}
   \item[Original derandomized knockoffs.] 
   At the prespecified nominal FDR level of $0.20,$ corresponding to the $\alpha_{\text{EBH}}$ threshold of the derandomized knockoff filter, the original derandomized knockoff procedure did not select any variable, indicating that no covariate exhibited sufficiently strong evidence to be flagged as a potential effect modifier.
    \item[Post-hoc derandomized knockoffs.]
    At the post-hoc adjusted significance level $\alpha \approx 0.27$, the derandomized knockoff procedure selected four variables as effect modifiers. Thus, with a small relaxation of the target FDR (from $20\%$ to $27\%$), the post-hoc adjustment yielded a non-empty and interpretable set of promising effect modifiers.
\end{description}

Having established the contrasting selection behavior of the traditional and post-hoc derandomized knockoff procedures, we next examine how these selections relate to the underlying signal structure captured by the post-hoc derandomized knockoff procedure. Figure~\ref{fig:clinical_trial_importance_scores} presents the variable-importance ranking, computed as the average (compound) e-values across the derandomized runs. This ranking offers a transparent view of the covariates exhibiting the strongest evidence of effect modification. The four variables selected by our post-hoc derandomized knockoff method at the relaxed FDR level of $0.27$ are highlighted in red, and they are: C-reactive protein, baseline fatigue score, age and a proteomic marker BD-2. These selections closely mirror findings reported in prior analyses of these datasets, where the same variables were highlighted as leading drivers of treatment heterogeneity in PsA \citep{sechidis2021usingknockoffs,bornkamp2023predicting,Cardner2023}.  
\begin{figure}[htpb]
  \begin{center}
    \includegraphics[width=0.75\textwidth]{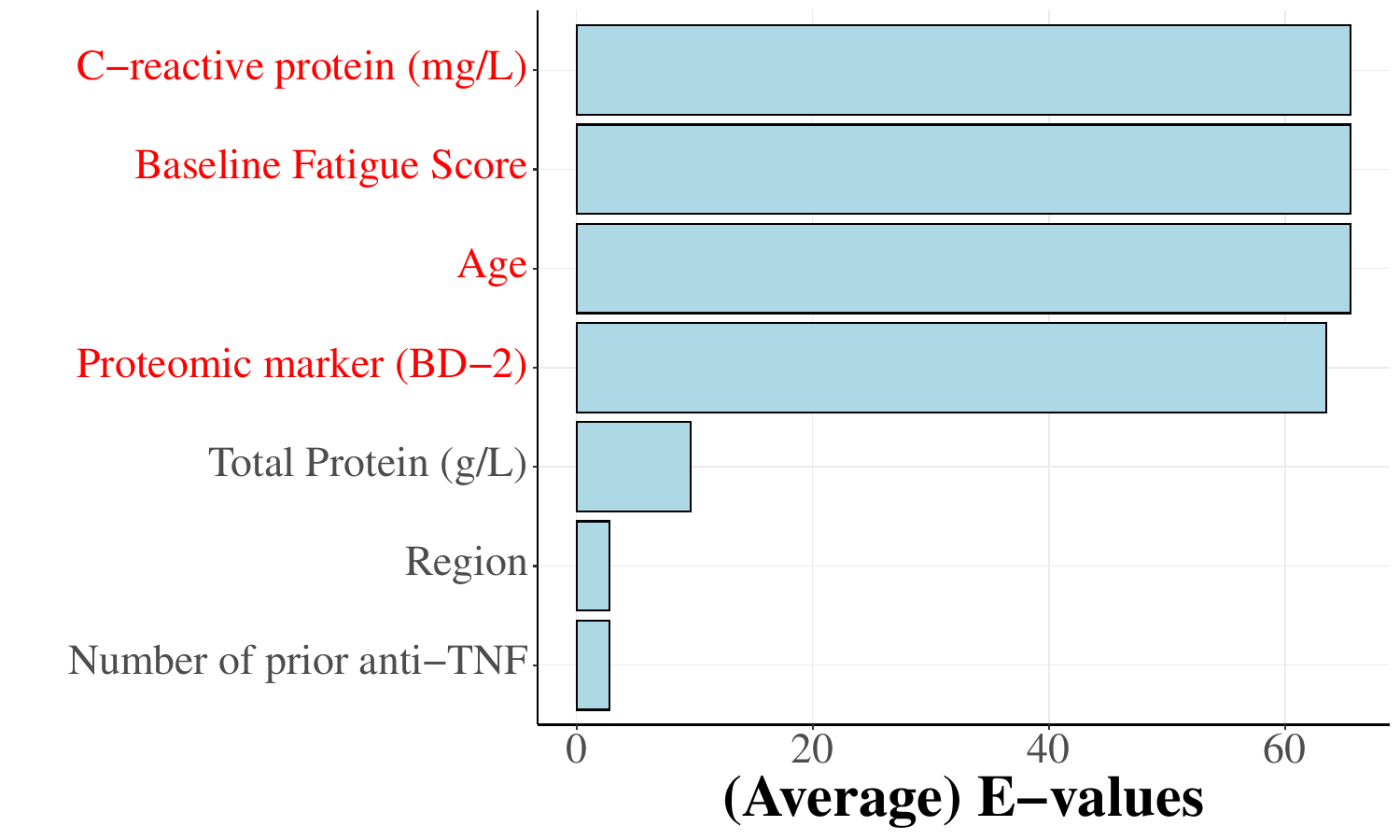}
  \end{center}
  \vspace{-0.5cm}
  \caption{\textbf{Variable importance ranking returned from the post-hoc derandomized knockoffs.} This figure shows the top variables ranked by their importance measured as the  average of (compound) e-values across the derandomized runs. Higher-ranked variables indicate stronger modification effects. We have highlighted with red the variables that got selected by our post-hoc derandomized knockoff procedure, and that control the FDR to the level of $0.27.$}
  \label{fig:clinical_trial_importance_scores}
\end{figure}

As we saw in Section~\ref{sec:rej_first}, the e-value framework allows us to move beyond a fixed rejection set and instead evaluate the specific level of evidence for any subset of interest. To illustrate this flexibility, consider a scenario where a practitioner is interested specifically in the evidence supporting only Age and the proteomic marker. Utilizing the property derived in Proposition~\ref{proposition:rej_first}, we can determine that the smallest nominal level $\alpha$ required to support this specific rejection set is approximately $0.52\%$. While our method enables this capability of determining the error level after choosing the rejection set, further research is needed to fully understand the practical consequences and the specific clinical contexts where such post-hoc decisions are most appropriate.

To sum it up, our post-hoc derandomized knockoff approach recovers clinically plausible modifiers while maintaining explicit FDR control, demonstrating that a modest relaxation of the nominal level can substantially improve power without inflating spurious discoveries. Together, these results illustrate that post-hoc adjustment offers a principled and practical enhancement to derandomized knockoffs, enabling robust and interpretable identification of effect modifiers in high-dimensional clinical trial settings.

\section{Discussion}

In this paper, we have shown how the significance level of the knockoff method can be chosen data-dependently while maintaining a reasonable FDR control. Remarkably, this flexibility comes without any price: (1) if the original method does not make any rejections, we can increase the significance level to potentially make rejections and (2) if the original method makes rejections, we can make the same rejections while often reducing (and never enlarging) the significance level, thus increasing the precision. In this way, the results obtained by our post-hoc knockoff procedure are always more informative than the results offered by the original knockoff procedure. This distinguishes our method from  existing post-hoc $\alpha$ approaches for other testing problems \citep{grunwald2024beyond, koning2023post, xu2022dynamic, gauthier2025values, dhillon2025scores}, where the gained flexibility usually reduces the power compared to the standard methods.

Our post-hoc approach solves the well-known issue of the original knockoff method, which requires at least $1/\alpha^{\kn}$ rejections to make any rejection. This makes our framework particularly suitable for settings with small numbers of variables and/or sparse signals. Furthermore, we integrated the post-hoc idea into derandomized knockoff procedures. While the latter require the choice of some arbitrary parameters ($\alpha_{\ebh}$ for FDR and $\eta$ for PFER control), which influences the power and the guarantee of the procedure, we showed that the parameters can be chosen in a fully data-dependent manner while maintaining reasonable control. \citet{ren2023derandomizing} even mentioned that a data-dependent choice of the parameter $\eta$ would be desirable but is not allowed and thus proposed data-splitting to determine $\eta$ in a data-driven way. Our approach to choose $\eta$ based on the entire data is much more powerful and practical.

In Section \ref{sec:uniform_improvement}, we introduced a uniform improvement of the derandomized (ph‑)knockoff methods based on the e‑Closure Principle. However, due to its substantial computational demands, the approach remains challenging to apply in practice. Identifying computationally efficient shortcuts for this method is an interesting direction for future work.

\subsection*{Acknowledgments}
We are thankful to Zhimei Ren for bringing the conditional calibration paper by \citet{luo2025improving} to our attention, and we also thank Alain Marti for his valuable feedback on the manuscript.

\vskip 0.2in
\bibliography{main}

@article{wasserman2020universal,
  title={Universal inference},
  author={Wasserman, Larry and Ramdas, Aaditya and Balakrishnan, Sivaraman},
  journal={Proceedings of the National Academy of Sciences},
  volume={117},
  number={29},
  pages={16880--16890},
  year={2020},
  publisher={National Acad Sciences}
}

@article{shafer2021testing,
  title={Testing by betting: A strategy for statistical and scientific communication },
  author={Shafer, Glenn},
  journal={Journal of the Royal Statistical Society Series A: Statistics in Society (with discussion)},
  volume={184},
  number={2},
  pages={407--431},
  year={2021},
  publisher={Oxford University Press}
}

@article{ramdas2023game,
  title={Game-theoretic statistics and safe anytime-valid inference},
  author={Ramdas, Aaditya and Gr{\"u}nwald, Peter and Vovk, Vladimir and Shafer, Glenn},
  journal={Statistical Science},
  volume={38},
  number={4},
  pages={576--601},
  year={2023},
  publisher={Institute of Mathematical Statistics}
}

@article{vovk2021values,
  title={E-values: Calibration, combination and applications},
  author={Vovk, Vladimir and Wang, Ruodu},
  journal={The Annals of Statistics},
  volume={49},
  number={3},
  pages={1736--1754},
  year={2021},
  publisher={Institute of Mathematical Statistics}
}

@article{wang2022false,
  title={False discovery rate control with e-values},
  author={Wang, Ruodu and Ramdas, Aaditya},
  journal={Journal of the Royal Statistical Society Series B: Statistical Methodology},
  volume={84},
  number={3},
  pages={822--852},
  year={2022},
  publisher={Oxford University Press}
}

@article{grunwald2020safe,
  title={Safe testing},
  author={Gr{\"u}nwald, Peter and de Heide, Rianne and Koolen, Wouter M},
  journal={Journal of the Royal Statistical Society Series B: Statistical Methodology (with discussion)},
  year={2024},
  publisher={Oxford University Press}
}

@article{howard2021time,
  title={Time-uniform, nonparametric, nonasymptotic confidence sequences},
  author={Howard, Steven R and Ramdas, Aaditya and McAuliffe, Jon and Sekhon, Jasjeet},
  journal={Annals of Statistics},
  year={2021}
}

@article{candes2018panning,
  title={Panning for gold:‘model-{X}’knockoffs for high dimensional controlled variable selection},
  author={Cand\`es, Emmanuel and Fan, Yingying and Janson, Lucas and Lv, Jinchi},
  journal={Journal of the Royal Statistical Society Series B: Statistical Methodology},
  volume={80},
  number={3},
  pages={551--577},
  year={2018},
  publisher={Oxford University Press}
}

@article{benjamini1995controlling,
  title={Controlling the false discovery rate: a practical and powerful approach to multiple testing},
  author={Benjamini, Yoav and Hochberg, Yosef},
  journal={Journal of the Royal Statistical Society Series B: Statistical Methodology},
  volume={57},
  number={1},
  pages={289--300},
  year={1995},
  publisher={Wiley Online Library}
}

@article{marcus1976closed,
  title={On closed testing procedures with special reference to ordered analysis of variance},
  author={Marcus, Ruth and Eric, Peritz and Gabriel, K Ruben},
  journal={Biometrika},
  volume={63},
  number={3},
  pages={655--660},
  year={1976},
  publisher={Oxford University Press}
}

@article{larsson2024numeraire,
  title={The numeraire e-variable and reverse information projection},
  author={Larsson, Martin and Ramdas, Aaditya and Ruf, Johannes},
  journal={arXiv preprint arXiv:2402.18810},
  year={2024}
}

@inproceedings{xu2022dynamic,
  title={Dynamic algorithms for online multiple testing},
  author={Xu, Ziyu and Ramdas, Aaditya},
  booktitle={Mathematical and Scientific Machine Learning},
  pages={955--986},
  year={2022},
  organization={PMLR}
}

@article{waudby2020confidence,
  title={Confidence sequences for sampling without replacement},
  author={Waudby-Smith, Ian and Ramdas, Aaditya},
  journal={Advances in Neural Information Processing Systems},
  volume={33},
  pages={20204--20214},
  year={2020}
}

@article{ramdas2024hypothesis,
  title={Hypothesis testing with e-values},
  author={Ramdas, Aaditya and Wang, Ruodu},
  journal={arXiv preprint arXiv:2410.23614},
  year={2024}
}

@article{xu2025bringing,
  title={Bringing Closure to False Discovery Rate Control: A General Principle for Multiple Testing},
  author={Xu, Ziyu and Solari, Aldo and Fischer, Lasse and de Heide, Rianne and Ramdas, Aaditya and Goeman, Jelle},
  journal={arXiv preprint arXiv:2509.02517},
  year={2025}
}

@article{koning2023post,
  title={Post-hoc $alpha $ Hypothesis Testing and the Post-hoc $ p $-value},
  author={Koning, Nick W},
  journal={arXiv preprint arXiv:2312.08040},
  year={2023}
}

@article{grunwald2024beyond,
  title={Beyond Neyman--Pearson: E-values enable hypothesis testing with a data-driven alpha},
  author={Gr{\"u}nwald, Peter D},
  journal={Proceedings of the National Academy of Sciences},
  volume={121},
  number={39},
  pages={e2302098121},
  year={2024},
  publisher={National Academy of Sciences}
}

@article{Lipkovich2017,
author = {Lipkovich, Ilya and Dmitrienko, Alex and B. D'Agostino Sr., Ralph},
title = {Tutorial in biostatistics: data-driven subgroup identification and analysis in clinical trials},
journal = {Statistics in Medicine},
volume = {36},
number = {1},
pages = {136-196},
year = {2017}
}

@article{spector2022powerful,
  title={Powerful knockoffs via minimizing reconstructability},
  author={Spector, Asher and Janson, Lucas},
  journal={The Annals of Statistics},
  volume={50},
  number={1},
  pages={252--276},
  year={2022},
  publisher={Institute of Mathematical Statistics}
}

@article{Sun2024,
author = {Sun, Sophie and Sechidis, Konstantinos and Chen, Yao and Lu, Jiarui and Ma, Chong and Mirshani, Ardalan and Ohlssen, David and Vandemeulebroecke, Marc and Bornkamp, Björn},
title = {Comparing algorithms for characterizing treatment effect heterogeneity in randomized trials},
journal = {Biometrical Journal},
volume = {66},
number = {1},
pages = {2100337},
year = {2024}
}

@article{sechidis2021usingknockoffs,
  title={Using knockoffs for controlled predictive biomarker identification},
  author={Sechidis, Konstantinos and Kormaksson, Matthias and Ohlssen, David},
  journal={Statistics in Medicine},
  volume={40},
  number={25},
  pages={5453--5473},
  year={2021},
  publisher={Wiley Online Library}
}

@misc{spector2024asymptoticallyoptimalknockoffstatistics,
      title={Asymptotically Optimal Knockoff Statistics via the Masked Likelihood Ratio}, 
      author={Asher Spector and William Fithian},
      year={2024},
      eprint={2212.08766},
      archivePrefix={arXiv},
      primaryClass={stat.ME},
      url={https://arxiv.org/abs/2212.08766}, 
}

@article{huang2020relaxing,
  title={Relaxing the assumptions of knockoffs by conditioning},
  author={Huang, Dongming and Janson, Lucas},
  journal={The Annals of Statistics},
  volume={48},
  number={5},
  pages={3021--3042},
  year={2020},
  publisher={JSTOR}
}

@article{Lipkovich2024,
author = {Lipkovich, Ilya and Svensson, David and Ratitch, Bohdana and Dmitrienko, Alex},
title = {Modern approaches for evaluating treatment effect heterogeneity from clinical trials and observational data},
journal = {Statistics in Medicine},
volume = {43},
number = {22},
pages = {4388-4436},
year = {2024}
}

@article{sechidis2018distinguishing,
  title={Distinguishing prognostic and predictive biomarkers: an information theoretic approach},
  author={Sechidis, Konstantinos and Papangelou, Konstantinos and Metcalfe, Paul D and Svensson, David and Weatherall, James and Brown, Gavin},
  journal={Bioinformatics},
  volume={34},
  number={19},
  pages={3365--3376},
  year={2018},
  publisher={Oxford University Press}
}

@article{luo2025improving,
  title={Improving knockoffs with conditional calibration},
  author={Luo, Yixiang and Fithian, William and Lei, Lihua},
  journal={The Annals of Statistics},
  volume={53},
  number={5},
  pages={2283--2302},
  year={2025},
  publisher={Institute of Mathematical Statistics}
}

@article{ren2024derandomised,
  title={Derandomised knockoffs: leveraging e-values for false discovery rate control},
  author={Ren, Zhimei and Barber, Rina Foygel},
  journal={Journal of the Royal Statistical Society Series B: Statistical Methodology},
  volume={86},
  number={1},
  pages={122--154},
  year={2024},
  publisher={Oxford University Press US}
}

@article{barber2015controlling,
  title={Controlling the false discovery rate via knockoffs},
  author={Barber, Rina Foygel and Cand{\`e}s, Emmanuel J},
  year={2015},
journal={The Annals of Statistics}
}

@article{ren2023derandomizing,
  title={Derandomizing knockoffs},
  author={Ren, Zhimei and Wei, Yuting and Cand{\`e}s, Emmanuel},
  journal={Journal of the American Statistical Association},
  volume={118},
  number={542},
  pages={948--958},
  year={2023},
  publisher={Taylor \& Francis}
}

@article{janson2016familywise,
  title={Familywise error rate control via knockoffs},
  author={Janson, Lucas and Su, Weijie},
  journal={Electronic Journal of Statistics},
  volume={10},
  number={1},
  pages={960--975},
  year={2016},
  publisher={The Institute of Mathematical Statistics and the Bernoulli Society}
}

@article{Sesia2018hmm,
    author = {Sesia, M and Sabatti, C and Candès, E J},
    title = "{Gene hunting with hidden Markov model knockoffs}",
    journal = {Biometrika},
    volume = {106},
    number = {1},
    pages = {1-18},
    year = {2018},
    month = {08},}

@article{Romano2019,
author = {Romano, Yaniv and Sesia, Matteo and Cand{\`{e}}s, Emmanuel},
journal = {Journal of the American Statistical Association},
number = {0},
pages = {1--27},
publisher = {Taylor {\&} Francis},
title = {{Deep Knockoffs}},
volume = {0},
year = {2019}
}

@article{Jordon2019,
author = {Jordon, James and Yoon, Jinsung and {Van Der Schaar}, Mihaela},
journal = {7th International Conference on Learning Representations, ICLR 2019},
pages = {1--25},
title = {{KnockoffGAN: Generating knockoffs for feature selection using generative adversarial networks}},
year = {2019}
}

@article{Bates2020,
author = {Bates, Stephen and Cand{\`{e}}s, Emmanuel and Janson, Lucas and Wang, Wenshuo},
journal = {Journal of the American Statistical Association},
number = {0},
pages = {1--25},
publisher = {Taylor {\&} Francis},
title = {{Metropolized Knockoff Sampling}},
volume = {0},
year = {2020}
}

@article{dhillon2025scores,
  title={E-Scores for (In) Correctness Assessment of Generative Model Outputs},
  author={Dhillon, Guneet S and Gonz{\'a}lez, Javier and Pandeva, Teodora and Curth, Alicia},
  journal={arXiv preprint arXiv:2510.25770},
  year={2025}
}

@article{gauthier2025values,
  title={E-values expand the scope of conformal prediction},
  author={Gauthier, Etienne and Bach, Francis and Jordan, Michael I},
  journal={arXiv preprint arXiv:2503.13050},
  year={2025}
}

@article{ignatiadis2024asymptotic,
  title={Asymptotic and compound e-values: multiple testing and empirical Bayes},
  author={Ignatiadis, Nikolaos and Wang, Ruodu and Ramdas, Aaditya},
  journal={arXiv preprint arXiv:2409.19812},
  year={2024}
}

@article{sechidis2025drlearner,
author = {Sechidis, Konstantinos and Zhang, Cong and Sun, Sophie and Chen, Yao and Spector, Asher and Bornkamp, Björn},
title = {Using Individualized Treatment Effects to Assess Treatment Effect Heterogeneity},
journal = {Statistics in Medicine},
volume = {44},
number = {28-30},
pages = {e70324},
doi = {https://doi.org/10.1002/sim.70324},
year = {2025}
}

@article{Kormaksson2021,
author = {Kormaksson, Matthias and Kelly, Luke J. and Zhu, Xuan and Haemmerle, Sibylle and Pricop, Luminita and Ohlssen, David},
title = {Sequential knockoffs for continuous and categorical predictors: With application to a large psoriatic arthritis clinical trial pool},
journal = {Statistics in Medicine},
volume = {40},
number = {14},
pages = {3313-3328},
doi = {https://doi.org/10.1002/sim.8955},
year = {2021}
}

@article{bornkamp2023predicting,
author = {Bornkamp, Björn and Zaoli, Silvia and Azzarito, Michela and Martin, Ruvie and Müller, Carsten Philipp and Moloney, Conor and Capestro, Giulia and Ohlssen, David and Baillie, Mark},
title = {Predicting subgroup treatment effects for a new study: Motivations, results and learnings from running a data challenge in a pharmaceutical corporation},
journal = {Pharmaceutical Statistics},
volume = {23},
number = {4},
pages = {495-510},
doi = {https://doi.org/10.1002/pst.2368},
year = {2024}
}

@article {Cardner2023,
	author = {Cardner, Mathias and Tuckwell, Danny and Kostikova, Anna and Forrer, Pascal and Siegel, Richard M and Marti, Alain and Vandemeulebroecke, Marc and Ferrero, Enrico},
	title = {Analysis of serum proteomics data identifies a quantitative association between beta-defensin 2 at baseline and clinical response to IL-17 blockade in psoriatic arthritis},
	volume = {9},
	number = {2},
	elocation-id = {e003042},
	year = {2023},
	doi = {10.1136/rmdopen-2023-003042},
	publisher = {BMJ Specialist Journals},
	URL = {https://rmdopen.bmj.com/content/9/2/e003042},
	eprint = {https://rmdopen.bmj.com/content/9/2/e003042.full.pdf},
	journal = {RMD Open}
}

@article{zimmermann2024,
author = {Zimmermann, Manuela R. and Baillie, Mark and Kormaksson, Matthias and Ohlssen, David and Sechidis, Konstantinos},
title = {All that Glitters Is not Gold: Type-I Error Controlled Variable Selection from Clinical Trial Data},
journal = {Clinical Pharmacology \& Therapeutics},
volume = {115},
number = {4},
pages = {774-785},
doi = {https://doi.org/10.1002/cpt.3211},
year = {2024}
}

@article{future1,
author = {Philip J. Mease  and Iain B. McInnes  and Bruce Kirkham  and Arthur Kavanaugh  and Proton Rahman  and Désirée van der Heijde  and Robert Landewé  and Peter Nash  and Luminita Pricop  and Jiacheng Yuan  and Hanno B. Richards  and Shephard Mpofu },
title = {Secukinumab Inhibition of Interleukin-17A in Patients with Psoriatic Arthritis},
journal = {New England Journal of Medicine},
volume = {373},
number = {14},
pages = {1329-1339},
year = {2015},
doi = {10.1056/NEJMoa1412679},

URL = {https://www.nejm.org/doi/full/10.1056/NEJMoa1412679},
}

@article{future2,
  title={{Secukinumab, a human anti-interleukin-{17A} monoclonal antibody, in patients with psoriatic arthritis {(FUTURE 2)}: a randomised, double-blind, placebo-controlled, phase 3 trial}},
  author={McInnes, Iain B and Mease, Philip J and Kirkham, Bruce and Kavanaugh, Arthur and Ritchlin, Christopher T and Rahman, Proton and Van der Heijde, Desiree and Landew{\'e}, Robert and Conaghan, Philip G and Gottlieb, Alice B and others},
  journal={The Lancet},
  volume={386},
  number={9999},
  pages={1137--1146},
  year={2015},
  publisher={Elsevier}
}

@article{future3,
  title={{Efficacy and safety of secukinumab administration by autoinjector in patients with psoriatic arthritis: results from a randomized, placebo-controlled trial (FUTURE 3)}},
  author={Nash, Peter and Mease, Philip J and McInnes, Iain B and Rahman, Proton and Ritchlin, Christopher T and Blanco, Ricardo and Dokoupilova, Eva and Andersson, Mats and Kajekar, Radhika and Mpofu, Shephard and others},
  journal={Arthritis research \& therapy},
  volume={20},
  number={1},
  pages={47},
  year={2018},
  publisher={Springer}
}

@article{future4,
  title={{Efficacy and safety of subcutaneous secukinumab 150 mg with or without loading regimen in psoriatic arthritis: results from the FUTURE 4 study}},
  author={Kivitz, Alan J and Nash, Peter and Tahir, Hasan and Everding, Andrea and Mann, He{\v{r}}man and Kaszuba, Andrzej and Pellet, Pascale and Widmer, Albert and Pricop, Luminita and Abrams, Ken},
  journal={Rheumatology and therapy},
  volume={6},
  number={3},
  pages={393--407},
  year={2019},
  publisher={Springer}
}

@article{future5,
  title={{Secukinumab improves active psoriatic arthritis symptoms and inhibits radiographic progression: primary results from the randomised, double-blind, phase III FUTURE 5 study}},
  author={Mease, Philip and van der Heijde, D{\'e}sir{\'e}e and Landew{\'e}, Robert and Mpofu, Shephard and Rahman, Proton and Tahir, Hasan and Singhal, Atul and Boettcher, Elke and Navarra, Sandra and Meiser, Karin and others},
  journal={Annals of the rheumatic diseases},
  volume={77},
  number={6},
  pages={890--897},
  year={2018},
  publisher={BMJ Publishing Group Ltd}
}

@article{felson1993american,
  title={The American College of Rheumatology preliminary core set of disease activity measures for rheumatoid arthritis clinical trials},
  author={Felson, David T and Anderson, Jennifer J and Boers, Maarten and Bombardier, Claire and Chernoff, Miriam and Fried, Bruce and Furst, Daniel and Goldsmith, Charles and Kieszak, Stephanie and Lightfoot, Robert and others},
  journal={Arthritis \& Rheumatism: Official Journal of the American College of Rheumatology},
  volume={36},
  number={6},
  pages={729--740},
  year={1993},
  publisher={Wiley Online Library}
}

\end{document}